\documentclass{LMCS}

\def\doi{8 (1:18) 2012}
\lmcsheading%
{\doi}
{1--49}
{}
{}
{May~\phantom.16, 2011}
{Mar.~\phantom02, 2012}
{}

\pdfoutput=1
\usepackage{enumerate,hyperref}
\usepackage{stmaryrd}
\usepackage{amssymb}

\theoremstyle{definition}\newtheorem{example}[thm]{Example}
\theoremstyle{definition}\newtheorem{definition}[thm]{Definition}
\theoremstyle{plain}
\theoremstyle{definition}\newtheorem{specification}[thm]{Specification}
\newcommand{\bugnelcodice}[1]{}
\newcommand{\andalso}[0]{\quad}
\newcommand{\Prod}[0]{\mathrm{\Pi}}
\newcommand{\Type}[0]{\mbox{\verb[language=grafite]{Type}}}
\newcommand{\Prop}[0]{\mbox{\verb[language=grafite]{Prop}}}

\newcommand{\match}[3]{\ensuremath{
 \mbox{\verb[language=grafite]+match+}~#1~
 \mbox{\verb[language=grafite]+in+}~#2~
 \mbox{\verb[language=grafite]+return+}~#3~
}}
\newcommand{\letin}[4]{\ensuremath{
 \mbox{\verb[language=grafite]+let+}~(#1:#2) := #3~
 \mbox{\verb[language=grafite]+in+}~#4
}}
\newcommand{\definit}[3]{\ensuremath{
\mbox{\verb[language=grafite]+definition+}~#1 : #2 := #3
}}
\newcommand{\axiom}[2]{\ensuremath{
\mbox{\verb[language=grafite]+axiom+}~#1 : #2
}}
\newcommand{\letr}[7][]{\ensuremath{
\hspace{-0.46em}\begin{array}{l}
\mbox{\verb[language=grafite]+let rec+}~#2 : #3 := #4
    ~\mbox{\verb[language=grafite]+and+}~\ldots #1
    ~\mbox{\verb[language=grafite]+and+}~ #5 : #6 := #7
\end{array}
}}
\newcommand{\letcr}[7][]{\ensuremath{
\hspace{-0.46em}\begin{array}{l}
\mbox{\verb[language=grafite]+let corec+}~#2 : #3 := #4
    ~\mbox{\verb[language=grafite]+and+}~\ldots #1
    ~\mbox{\verb[language=grafite]+and+}~ #5 : #6 := #7
\end{array}
}}
\newcommand{\inductive}[9]{\ensuremath{
\hspace{-0.46em}\begin{array}{l}
 #1\textrm{\verb[language=grafite]+inductive+}
\quad #2~:~#3 := #4 \;|\;\ldots\;|\; #5\\
\qquad\qquad \textrm{\verb[language=grafite]+with+}~\ldots\\
\qquad\qquad \textrm{\verb[language=grafite]+with+}~ 
#6 ~:~ #7 := #8 \;|\;\ldots\;|\; #9 
\end{array}
}}
\newcommand{\inductiveX}[4]{\ensuremath{
\hspace{-0.46em}\begin{array}{l}
 #1\textrm{\verb[language=grafite]+inductive+} #2~:~#3 := #4
\end{array}
}}
\RequirePackage{nicefrac}
\newcommand{\subst}[2]{[{#2}/{#1}]}
\newcommand{\mwhd}{\triangleright_{\mathrm{whd}}}
\usepackage{listings}
\renewcommand{\verb}{\lstinline}

\newcommand{\raisedrule}[2][0em]{\leaders\hbox{\rule[#1]{1pt}{#2}}\hfill}
\newcommand{\placeholder}{\stackrel{\raisedrule[0.165em]{0.3pt}\!\rightarrowtriangle}{\mathrm{placeholder}}}
\renewcommand{\vector}[1]{\overrightarrow{#1}}
\newcommand{\lland}[0]{\quad \land \quad}
\newcommand{\refineslabel}{\leq}
\newcommand{\refines}[2]{#1 \refineslabel #2}
\newcommand{\mldots}{\ensuremath{\stackrel{\rightarrowtriangle}{?}}}
\newcommand{\WF}[1]{\ensuremath{\mathcal{WF}(#1)}}
\newcommand{\h}[1]{} % hide
\newcommand{\nat}{\ensuremath{\mathbb{N}}}
\newcommand{\env}[0]{\ensuremath{\mathsf{Env}}}
\newcommand{\pts}[0]{\ensuremath{\mathsf{PTS}}}

\newcommand{\elim}[0]{\ensuremath{\mathsf{elim}}}
\newcommand{\eatsep}{|\hspace{-0.4em}_{\stackrel{~}{\blacktriangle}}}
\newcommand{\fooA}[1]{\ensuremath{\stackrel{\mathcal{FOOA}}{~\leadsto}}}
\newcommand{\fooB}[1]{\ensuremath{\stackrel{\mathcal{FOOB}}{~\leadsto}}}
\newcommand{\reclabel}{\ensuremath{\smash{}\mathcal{R^{\smash{\Uparrow}}}}}
\newcommand{\recElabel}{\ensuremath{\smash{}\mathcal{R^{\smash{\Downarrow}}}}}
\newcommand{\rec}[1]{\ensuremath{\stackrel{\,\reclabel}{~\leadsto}}}
\newcommand{\recE}[1]{\ensuremath{\stackrel{\recElabel}{~\leadsto}}}
\newcommand{\recO}[1]{\ensuremath{\stackrel{\mathcal{R}}{~\leadsto}}}

\newcommand{\xP}[0]{\ensuremath{\Sigma}}
\newcommand{\xG}[0]{\ensuremath{\Gamma}}
\newcommand{\xS}[0]{\ensuremath{\Phi}}
\newcommand{\metas}[0]{\ensuremath{\mathcal{M}}}
\newcommand{\xPSG}[3]{\xP{#1},~\xS{#2},~\xG{#3}}
\newcommand{\xPSGp}[3]{(\xP{#1},~\xS{#2})~\xG{#3}}
\newcommand{\xPS}[2]{\xP{#1},~\xS{#2}}
\newcommand{\xPSp}[2]{(\xP{#1},~\xS{#2})}
\newcommand{\UNI}{\stackrel{?}{\equiv}}
\newcommand{\unilabel}{\ensuremath{\mathcal{U}}}
\newcommand{\uni}[1]{\ensuremath{\stackrel{\unilabel#1}{~\leadsto~}}}
\newcommand{\eatT}{\ensuremath{\mathcal{E}^T}}
\newcommand{\eatt}{\ensuremath{\mathcal{E}_t}}
\newcommand{\unimany}{\ensuremath{\stackrel{\eatt}{~\leadsto~}}}
\newcommand{\coer}{\ensuremath{\stackrel{\mathcal{C}}{~\leadsto}}}
\newcommand{\forcetotypelabel}{\ensuremath{\stackrel{\mathcal{F}}{~\leadsto}}}
\newcommand{\whdlabel}{\mwhd}
\newcommand{\eatprodslabel}{\ensuremath{\stackrel{\eatT}{~\leadsto~}}}

\newcommand{\infrule}[3][]{#1\quad\frac{\begin{array}{l}#2\end{array}}{#3}}

\newcommand{\eatprods}[9]{\ensuremath{\xPSGp{#1}{#2}{#3} \vdash #4 ~\eatsep~ #5 \eatprodslabel{} #6 :\; #7~ (\xP{}#8,~ \xS{}#9)}}
\newcommand{\refinexbase}[9]{\ensuremath{\xPSGp{#1}{#2}{#3} \vdash #4 #5 \rec{}~ #6 #7 ~ (\xP{}#8,~ \xS{}#9)}}
\newcommand{\refinexbaseE}[9]{\ensuremath{\xPSGp{#1}{#2}{#3} \vdash #4 #5 \recE{}~ #6 #7 ~ (\xP{}#8,~ \xS{}#9)}}
\newcommand{\refinex}[8]{\refinexbase{#1}{#2}{#3}{#4}{}{#5}{:\; #6}{#7}{#8}}
\newcommand{\refinexE}[8]{\refinexbaseE{#1}{#2}{#3}{#4}{:\; #5}{#6}{}{#7}{#8}}

\newcommand{\unifx}[8]{\ensuremath{\xPSGp{#2}{#3}{#4} \vdash #5 \UNI #6 \uni{#1} (\xP{}#7,~ \xS{}#8)}}
\newcommand{\Unimanyx}[9]{\ensuremath{\XPSGp{#1}{#2}{#3} \vdash #4 \UNI #5 \unimany{} #6~\eatsep~#7\h{(\xP{}#8,~ \xS{}#9)}}}
\newcommand{\unimanyx}[9]{\ensuremath{\xPSGp{#1}{#2}{#3} \vdash #4 \UNI #5 \unimany{} #6~\eatsep~#7 ~ (\xP{}#8,~ \xS{}#9)}}

\newcommand{\unifcoercex}[9]{\ensuremath{\xPSGp{#1}{#2}{#3} \vdash #4 : #5 \UNI #6 \coer{}~ #7~ (\xP{}#8,~ \xS{}#9)}}
\newcommand{\forcetotype}[8]{\ensuremath{\xPSGp{#1}{#2}{#3} \vdash #4 \forcetotypelabel{}~ #5 :\; #6~ (\xP{}#7,~ \xS{}#8)}}
\newcommand{\whd}[5]{\ensuremath{\xPSG{#1}{#2}{#3} \vdash #4 \whdlabel{} #5}}

\newcommand{\lookforcx}[9]{\ensuremath{\xPSGp{#1}{#2}{#3} \vdash #4 \rightarrowtail #5 \stackrel{\Delta}{~\leadsto~} #6,~ #7 :\; #8~ (\xP{}#9,\xS{})}}

\newcommand{\XPS}[2]{\h{\xP{#1},~\xS{#2}}}
\newcommand{\XPSG}[3]{\h{\xP{#1},~\xS{#2},~}\xG{#3}}
\newcommand{\XPSGp}[3]{\h{(\xP{#1},~\xS{#2})~}\xG{#3}}
\newcommand{\Eatprods}[9]{\ensuremath{\XPSG{#1}{#2}{#3} \vdash #4 ~\eatsep~ #5 \eatprodslabel{} #6 :\; #7\h{,~ \xP{}#8,~ \xS{}#9}}}
\newcommand{\Refinexbase}[9]{\ensuremath{\XPSG{#1}{#2}{#3} \vdash #4 #5 \rec{}~ #6 #7\h{,~ \xP{}#8,~ \xS{}#9}}}
\newcommand{\RefinexbaseE}[9]{\ensuremath{\XPSG{#1}{#2}{#3} \vdash #4 #5 \recE{}~ #6 #7\h{,~ \xP{}#8,~ \xS{}#9}}}
\newcommand{\Refinex}[8]{\Refinexbase{#1}{#2}{#3}{#4}{}{#5}{:\; #6}{#7}{#8}}
\newcommand{\RefinexE}[8]{\RefinexbaseE{#1}{#2}{#3}{#4}{:\; #5}{#6}{}{#7}{#8}}
\newcommand{\Extendx}[8]{\ensuremath{\xP{}\h{#8} \leadsto \xP{} \cup \h{= \xP{}#1 \land}\{ \Gamma{}#2 \vdash #3 :\; #4 \h{\land}\;,\; \Gamma{}#5 \vdash #6 :~ #7\}}}
\newcommand{\Extendxs}[5]{\ensuremath{\xP{}\h{#5} \leadsto \xP{} \cup \h{= \xP{}#1 \land}\{ \Gamma{}#2 \vdash #3 :\; #4 \}}}
\newcommand{\Unifx}[8]{\ensuremath{\XPSG{#2}{#3}{#4} \vdash #5 \UNI #6 \uni{#1} \h{\xP{}#7,~ \xS{}#8}}}

\newcommand{\Unifcoercex}[9]{\ensuremath{\XPSG{#1}{#2}{#3} \vdash #4 : #5 \UNI #6 \coer{}~ #7\h{,~ \xP{}#8,~ \xS{}#9}}}
\newcommand{\Forcetotype}[8]{\ensuremath{\XPSG{#1}{#2}{#3} \vdash #4 \forcetotypelabel{}~ #5 :\; #6\h{,~ \xP{}#7,~ \xS{}#8}}}
\newcommand{\Whd}[5]{\ensuremath{\XPSG{#1}{#2}{#3} \vdash #4 \whdlabel{} #5}}

\newcommand{\Lookforcx}[9]{\ensuremath{\XPSG{#1}{#2}{#3} \vdash #4 \rightarrowtail #5 \stackrel{\Delta}{~\leadsto~} #6,~ #7 :\; #8\h{,~ \xP{}#9}}}

\begin{document}

\title[A Bi-Directional Refinement Algorithm for CIC]{A Bi-Directional Refinement Algorithm for the Calculus of (Co)Inductive Constructions}

\author[A.~Asperti]{Andrea Asperti\rsuper a}	%required
\address{{\lsuper{a,b,c}}Dipartimento di Scienze dell'informazione\\
 Mura Anteo Zamboni 7\\
 40127, Bologna, Italy}	%required
\email{\{asperti,ricciott,sacerdot\}@cs.unibo.it}  %optional
%\thanks{thanks 1, optional.}	%optional

\author[W.~Ricciotti]{Wilmer Ricciotti\rsuper b}	%optional
\address{\vskip-6 pt}	%required
%\email{ricciott@cs.unibo.it}  %optional
%\thanks{thanks 1, optional.}	%optional

\author[C.~Sacerdoti Coen]{Claudio Sacerdoti Coen\rsuper c}	%optional
\address{\vskip-6 pt}	%required
%\email{sacerdot@cs.unibo.it}  %optional
%\thanks{thanks 1, optional.}	%optional

\author[E.~Tassi]{Enrico Tassi\rsuper d}	%optional
\address{{\lsuper d}Microsoft Research-INRIA Joint Centre\\
Building I, Parc Orsay Universit\'e \\
28, rue Jean Rostand, 91893 Orsay Cedex}	%optional
\email{enrico.tassi@inria.fr}  %optional
%\thanks{thanks 3, optional.}	%optional

%% etc.

%% required for running head on odd and even pages, use suitable
%% abbreviations in case of long titles and many authors:

%% mandatory lists of keywords and classifications:
%\titlecomment{OPTIONAL comment concerning the title, \eg, if a variant
%or an extended abstract of the paper has appeared elsewehere}
%%%%%%%%%%%%%%%%%%%%%%%%%%%%%%%%%%%%%%%%%%%%%%%%%%%%%%%%%%%%%%%%%%%%%%%%%%%

%% the abstract has to PRECEED the command \maketitle:
%% be sure not to issue the \maketitle command twice!

\begin{abstract}
The paper describes the refinement algorithm for the Calculus of
(Co)Inductive Constructions (CIC) implemented in the interactive theorem prover
Matita.

The refinement algorithm is in charge of giving a meaning to the terms,
types and proof terms directly written by the user or generated by using
tactics, decision procedures or general automation. The terms are written
in an ``external syntax'' meant to be user friendly that allows omission
of information, untyped binders and a certain liberal use of user defined
sub-typing. The refiner modifies the terms to obtain related well typed terms
in the internal syntax understood by the kernel of the ITP. In particular,
it acts as a type inference algorithm when all the binders are untyped.

The proposed algorithm is bi-directional: given a term in external syntax
and a type expected for the term, it propagates as much typing information
as possible towards the leaves of the term. Traditional mono-directional
algorithms, instead, proceed in a bottom-up way by inferring the type of
a sub-term and comparing (unifying) it with the 
type expected by its context only at the end.
We propose some novel bi-directional rules for CIC that are particularly
effective. Among the benefits of bi-directionality we have better error
message reporting and better inference of dependent types.
Moreover, thanks to bi-directionality, the coercion system for sub-typing
is more effective and type inference generates simpler unification
problems that are more likely to be solved by the inherently 
incomplete higher order unification algorithms implemented.

Finally we introduce in the external syntax
the notion of vector of placeholders that enables to omit at once an arbitrary
number of arguments. Vectors of placeholders allow a trivial implementation
of implicit arguments and greatly simplify the implementation of primitive
and simple tactics.
\end{abstract}

\maketitle

\bugnelcodice{ scompare nella versione definitva
\newpage
\tableofcontents{}
\newpage}

\section{Introduction}\label{S:one}

In this paper we are interested in describing one of the key ingredients in the
implementation of Interactive Theorem Provers (ITP) based on type theory.  

The architecture of these tools is usually organized in layers 
and follows the so called de Bruijn principle: the correctness of the whole
system solely depends on the innermost component called kernel. 
Nevertheless, from a user perspective, the most interesting layers are the
external ones, the ones he directly interacts with. Among these, the
\emph{refiner} is the one in charge of giving a meaning to the terms and types
he writes. The smarter the refiner is, the more freedom the user has
in omitting pieces of information that can be reconstructed.
The refiner is also the component generating the majority of error
messages the user has to understand and react to
in order to finish his proof or definition.

This paper is devoted to the description of a refinement algorithm
for the Calculus of (Co)Inductive Constructions, the type theory on
which the Matita~\cite{matita-jar-uitp}, Coq~\cite{coq} and 
Lego~\cite{lego} ITPs are based on.

\subsection{Refinement}
In this and in the previous paper~\cite{ck-sadhana} we are interested in the
implementation of interactive theorem provers (ITP) for dependently typed languages
that are %seriously 
heavily based on the Curry-Howard isomorphism. Proofs are represented
using lambda-terms. Proofs in progress are represented using lambda-terms
containing metavariables that are implicitly existentially quantified.
Progression in the proof is represented by instantiation of metavariables
with terms. Metavariables %that occur in terms 
are also useful to represent
%user provided terms with 
missing or partial information, like
untyped lambda-abstractions or instantiation of polymorphic functions to
omitted type arguments.% that are omitted and left to the system to be inferred.

Agda~\cite{agda} and Matita~\cite{matita-jar-uitp} are examples of systems
implemented in this way. Arnaud Spiwack in his Ph.D. thesis~\cite{Spiwack} partially
describes a forthcoming release of Coq 8.4 that will be implemented on the same
principles.

The software architecture of these systems is usually built in layers.
The innermost layer is the \emph{kernel} of the ITP. The main algorithm
implemented by the kernel is the \emph{type checker}, which is based in turn
on \emph{conversion} and \emph{reduction}. The type checker takes as input
a (proof) term possibly containing metavariables and it verifies if the 
partial term is correct so far. 
To allow for type-checking, metavariables are associated to
sequents, grouping their types together with the context (hypotheses) 
available to inhabit the type.
% of the sequent restricts the set of free
% variables that can occur in the terms used to instantiate the metavariable,
% whose type must be the type (conclusion) of the sequent. To allow reduction and
% conversion, substitution in terms containing metavariables needs to be defined.
% To make substitution (hence reduction) commute with metavariable instantiation,
% we turn a substitution that hits a metavariable occurrence into an explicit
% substitution. But for that, 
The kernel does not alter metavariables since no
 instantiation takes place during reduction, conversion or type checking.
 
The kernel has the important role of reducing the trusted code base of the ITP. 
Indeed, the kernel eventually verifies all proofs produced by the outer layers,
detecting incorrect proofs generated by bugs in those layers. Nevertheless,
the user never interacts directly with the kernel and the output of the
kernel is just a boolean that is never supposed to be false when the rest of
the system is bug free. The most interesting layers from the user point of
view are thus the outer layers. The implementation of a kernel for a variant of
the Calculus of (Co)Inductive Constructions (CIC) has been described
in~\cite{ck-sadhana} down to the gory details that make the 
implementation efficient.

The next layer is the \emph{refiner} and is the topic of this paper.
The main algorithm implemented by the refiner is the refinement algorithm
%which, in turn, is based on unification. 
that tries to infer as much information as it is needed to make its input
meaningful. In other words it
%The refinement algorithm 
takes
as input a partial term, written in an ``external syntax'', and tries to obtain
a ``corresponding'' well typed term. The input term can either be user
provided or it can be a partial proof term generated by some 
proof command (called tactic) or automation procedure. 
The gap between the external and internal syntax is rather
arbitrary and system dependent. Typical examples of external syntaxes 
allow for:
 \begin{iteMize}{$\bullet$}
  \item Untyped abstractions. Hence the refiner must perform type inference
   to recover the explicit types given to bound variables. The polymorphism of
   CIC is such that binders are required to be typed to make type checking
   decidable.
  \item Omission of arguments, in particular omission of types used to
   instantiate polymorphic functions. Hence the refiner must
   recover the missing information during type inference to turn implicit
   into explicit polymorphism.
  \item Linear placeholders for missing terms that are not supposed to 
   be discovered during
   type inference. For instance, a placeholder may be inserted by a tactic to
   represent a new proof obligation. Hence the refiner must turn the placeholder
   into a metavariable by constraining the set of free variables that may occur
   in it and the expected type.
  \item Implicit ad-hoc sub-typing determined by user provided cast functions (called coercions)
   between types or type families. Hence the refiner must modify the user
   provided term by explicitly inserting the casts in order to let
   the kernel completely ignore sub-typing.
 \end{iteMize}
Coercions are user provided functions and are thus free to completely ignore
their input. Thus a refiner that handles coercions is actually able to
arbitrarily patch wrong user provided terms turning them into arbitrarily
different but well typed terms. Moreover, the insertion of a coercion between
type families can also introduce new metavariables (the family indexes) that
play the role of proof obligations for pre-conditions of the coercion. For
instance, a coercion from lists to ordered lists can open a proof obligation
that requires the list to be sorted. 
%This leads to an interesting style of
%programming with dependent types where the user types functions using precise
%types that precisely describe the function specification and then writes down
%the code ignoring the dependent types. The refiner patches the code by inserting
%coercions that open proof obligations to grant that the function respects its
%specification. This style has been proposed by Sozeau in~\cite{???} and
%implemented in Coq as an external layer. Nevertheless, it can be directly
%supported without any major effort by the refiner itself.

The refiner is the most critical system component from the user point of view
since it is responsible for the ``intelligence'' of the ITP: the more powerful
the refiner is, the less information is required from the user and the simpler
the outer layers become. For instance, a series of recent techniques that
really improve the user experience have all been focused in the direction of
making the refiner component more powerful and extensible by the user.
Canonical structures~\cite{canonical-structures}, 
unification hints~\cite{unification-hints} and 
type classes~\cite{SozeauO08} are devices that let the user drive some 
form of proof search that is seamlessly integrated in the refinement process. 
While the latter device is directly integrated into the refinement algorithm,
the first two are found in the unification algorithm used by the refiner.

They all make it possible to achieve similar objectives, the second being more general than
the first and the last two being incomparable from the point of view of
efficiency (where the second is best) and expressiveness (where the third is
more flexible). The implementation of type classes done in Coq is actually
provided by an additional layer outside the refiner for historical reasons.

In this paper we will describe only the refinement algorithm implemented in a
refiner for a variant of the Calculus of (Co)Inductive Constructions. The
algorithm is used in the forthcoming major release of the 
Matita\footnote{Matita is free software available at \url{http://matita.cs.unibo.it}} ITP (1.0.x). The
algorithm calls a unification algorithm that will be specified in this paper and
described elsewhere.
We do not consider type classes in our refinement algorithm since we prefer to
assume the unification algorithm to implement unification hints. Nevertheless,
type classes can be easily added to our algorithm with minor modifications and
indeed the relevant bits that go into the refiner are implemented in Matita.

Before addressing bi-directionality, which is a peculiarity of the algorithm
that has not been fully exploited yet\footnote{The refinement algorithm of Coq
8.3, the most widespread implementation of CIC, is almost mono-directional with only the lambda-abstraction case handled in a
bi-directional way. Many other interesting cases of bi-directionality are obtained in this paper for inductive types and constructors.} for the CIC, we just conclude our overview of an ITP architecture by talking about the next layer. The next layer after the refiner is that
of \emph{tactics}. This layer is responsible for implementing commands that
help the user in producing valid proof terms by completely hiding to him the
proof terms themselves. Tactics range from simple ones that capture the
introduction and elimination rules of the connectives (called primitive tactics) to complicated proof automation procedures. The complexity of proof automation is inherent in the problem.
On the other hand, primitive tactics should be as simple as building small
partial proof terms. For instance, to reduce a proof of $A \Rightarrow B$ to
a proof of $B$ given $A$ it is sufficient to instantiate the metavariable
associated to the sequent $\vdash A \Rightarrow B$ with the term
$\lambda x.?$ in external syntax where $?$ is a placeholder for a new proof
obligation.
This is possible when the refinement algorithm is powerful enough to
refine $\lambda x.?$ to $\lambda x:A.?_1$ where $?_1$ is a new metavariable
associated to the sequent $x:A \vdash B$. When this is not the case or when
the refiner component is totally missing, the
tactic is forced to first perform an analysis of the current goal, then explicitly
create a new metavariable and its sequent, and then emit the new proof
term $\lambda x:A.?_1$ directly in the internal syntax. 

\subsection{Bi-directionality}

When the external syntax of our ITP allows to omit types in binders, the
refinement algorithm must perform type inference. Type inference was originally
studied in the context of lambda-calculi typed a la Curry, where no type
information can be attached at all to the binders. The traditional algorithm
for type inference, now called uni-directional, performs type inference by
first traversing the term in a top-down way.
When a binder is met, a new metavariable (usually called type or unification
variable in this context)
is introduced for the type of the bound variable. 
Then type constraints are solved traversing the term in a bottom-up way. 
When the variable or, more
generally, a term is used in a given context, its type (called inferred type)
is constrained to be compatible with the one expected by the context (called
expected type). This triggers a unification problem.

Type inference, especially for the Hindley-Milner type system, gives the possibility to write
extremely concise programs by omitting all types. Moreover, it often detects
a higher degree of polymorphism than the one expected by the user.
Unluckily, it has some drawbacks. A minor one is that types are
useful for program documentation and thus the user desires to add types at
least to top level functions. In practice, this is always allowed by concrete
implementations. Another problem is error reporting: a typing error always
manifests itself as a mismatch between an inferred and an expected type.
Nevertheless, an error can be propagated to a very distant point in the code
before being detected and the position where it is produced. The
mismatch itself can be non informative about where the error actually is. 
Finally, unification quickly becomes undecidable when the expressive power of
the type system increases. In particular, it is undecidable for higher order
logic and for dependent types.

To avoid or mitigate the drawbacks of type inference, bi-directional
type-checking algorithms have been introduced in the 
literature~\cite{piercelocaltype}. 
These algorithms take
as input a $\lambda$-term 
typed a la Curry and an expected top-level type and they
proceed in a top-down manner by propagating the expected type towards the leaves of
the term. Additional expected types are given in local definitions, so that
all functions are explicitly documented. Error detection is improved by making
it more local. The need for unification is reduced and, for simple type systems,
unification is totally avoided. Some terms, in particular $\beta$-redexes, 
are no
longer accepted, but equivalent terms are (e.g. by using a local definition
for the head). An alternative consists of accepting all terms by re-introducing a
dependency over some form of unification.

Bi-directionality also makes sense for languages typed \`a la Church, like the
one we consider here. In this case the motivations are slightly different. First
of all, typing information is provided both in the binders
and at the top-level, in the form of an expected type. Hence information can
flow in both direction and, sooner or later, the need to compare the expected
and inferred types arises. In the presence of implicit polymorphism, unification
is thus unavoidable. Because of dependent types and metavariables for
proof obligations, we need the full power of higher order unification.
%and thus mono-directionality is sufficient to achieve relative completeness.
%Relative completeness is the completeness of the refinement algorithm assuming
%an oracle for the undecidable unification problem.
Moreover, again because of
unification, the problem remains undecidable also via using a bi-directional
algorithm. Hence avoiding unification is no longer a motivation for bi-directionality. The remaining motivations for designing a bi-directional refinement algorithm for CIC are the following:

\subsubsection*{Improved error messages.} A typing error is issued every time a
   mismatch is found between the inferred and expected type. With a
   mono-directional algorithm, the mismatch is always found at the end,
   when the typing information reaches the expected type. In a bi-directional
   setting the expected type is propagated towards the leaves and the inferred
   type towards the root, the mismatch is localized in smaller sub-terms and
   the error message is simpler. For instance, instead of the message
   ``\emph{the provided function has type $A \Rightarrow List~B$ but it is supposed
     to have type $A \Rightarrow List~C$}'' related to a whole function
   definition one could get the simpler message ``\emph{the list element
   has type $B$ but it is supposed to have type $C$}'' related to one particular
   position in the function body.

\subsubsection*{Improvement of the unification algorithm.} To make the system responsive,
   the semi-decidable unification algorithm is restricted to always give an
   answer in a finite amount of time. Hence the algorithm could fail to find
   a solution even when a solution exists. For instance, the algorithms
   implemented in Coq and Matita are essentially backtracking free and they
   systematically favor projections over mimics: when unifying an applied
   metavariable $?_1~a~b~c$ with a $b$ (for some $a,b,c$ closed in a context
   $\Gamma$), the system instantiates $?_1$ with $\lambda x,y,z.y$ rather than
   $\lambda x,y,z.b$ (where $x,y,z \notin \mathit{dom}(\Gamma)$). Moreover, 
   unification for CIC does not
   admit a most general unifier and it should recursively enumerate the set
   of solutions. However, it is usual in system implementations to let
   unification return just one solution and to avoid back-tracking in the
   refinement algorithm\footnote{To the authors knowledge, Isabelle~\cite{isabelle} is the only interactive prover implementing Huet's algorithm~\cite{huet2order} capable of generating all second order unifiers}. 
   Thus, if the solution found by unification is correct
   locally, but not globally, refinement will fail.
   Thanks to bi-directionality, 
   unification problems often become more instantiated and thus simpler, and
   they also admit fewer solutions.
   In particular, in the presence of dependent
   types, it is easy to find practical examples where the unification algorithm
   finds a solution only on the problems triggered by the bi-directional
   algorithm.

   An interesting and practical example that motivated our investigation of
   bi-directionality is the following.
   Consider a dependently typed data-type $(\mbox{Term}~S)$ that represents
   the syntax of a programming language with binders. Type dependency is
   exploited to make explicit the set $S$ of variables bound in the term and
   every variable occurrence must come with a proof that the variable occurs
   in the bound variables list: $(\mbox{Var}~S~x~\mbox{I})$ has type $(\mbox{Term}~S)$ where $x$ is
   a variable name, $\mbox{I}$ is a proof of $\mbox{True}$ and $\mbox{Var}$ has type
   $\forall S.\forall a:\mbox{String}. x \in S \to \mbox{Term}~S$ where $x \in S$ is
   a computable function that reduces to $\mbox{True}$ when $x$ belongs to $S$ and to
   $\mbox{False}$ otherwise. Consider now the term $(\mbox{Lambda}~?~x~(\mbox{Var}~?~x~\mbox{I}))$ in concrete
   syntax that represents $\lambda x.x$ in our programming language. Note
   that no information about the set of bound variables has been provided
   by the user. Thus it is possible to simply define notational macros so that
   the user actually writes $\lambda x.x$ and this is expanded\footnote{User provided notational macros are used to extend the external syntax of an ITP and they are expanded before refinement, yielding a term in external syntax to be refined.} to $\mbox{Lambda}~?~x~(\mbox{Var}~?~x~\mbox{I})$.
  A uni-directional refiner is unlikely to
 accept the given term since it should guess the right value for the second placeholder $?$ such that $x \in~?$ reduces to $\mbox{True}$ and $?$ is the set
 of variables actually bound in the term. The latter information is not local
 and it is still unknown in the bottom-up, uni-directional approach. On the other hand, a
 bi-directional refiner that tries to assign type $\mbox{Term}~\emptyset$ to the term
 would simply propagate $\emptyset$ to the first placeholder and then
 propagate $\emptyset \cup \{x\}$ to the second one, since $\mbox{Lambda}$, which is
 a binder, has type $\forall S.\forall x. \mbox{Term}~(S \cup \{x\}) \to \mbox{Term}~S$.
 Finally, $\mbox{True}$ is the inferred type for $\mbox{I}$, whose expected
type is $x \in \emptyset \cup \{x\}$. The two types are convertible and the input is now accepted without any guessing.
  
\subsubsection*{Improvement of the coercion mechanism.} Coercions are triggered when
   unification fails. 
   They are explicit cast functions, declared by the user, used to fix the
   type of sub-terms.
   Simplifying the unification problem allows to retrieve
   more coercions. For instance, consider a list $[1;2;3]$ of natural numbers
   used as a list of integer numbers and assume the existence of a coercion
   function $k$ from natural to integers.
   In the mono-directional problem, the
   failing unification problem is $(\mbox{List}~\nat{})$ vs $(\mbox{List}~\mathbb{Z})$. The coercion required
   is the one obtained lifting $k$ over lists. The lifting has to be performed
   manually by the user or by the system. In the latter case, the system
   needs to recognize that lists are containers and has to have code to
   lift coercions over containers, like in~\cite{chenPHD}. 
   In the bi-directional case, however, the expected type 
   $(\mbox{List}~\mathbb{Z})$ would
   propagate to assign to each list element the expected type $\mathbb{Z}$ and
   the coercion $k$ would be applied to all integers in the list without need
   of additional machinery. The bi-directional algorithm presented in this paper
   does not allow to remove the need for the coercion over lists in all
   situations, but it is sufficient in many practical ones, like the one just
   considered.

\subsubsection*{Introduction of vectors of placeholders (``\ldots'') in the external syntax.}
   A very common use of dependently typed functions consists in explicitly passing to them
   an argument which is not the first one and have the system infer the
   previous arguments using type dependencies. For instance, if
   $\mbox{Cons}: \forall A. A \to \mbox{List}~A \to \mbox{List}~A$ and $l$ is a list of integers,
   the user can simply write $(\mbox{Cons}~?~2~l)$ and have the system infer that
   $?$ must be instantiated with the type of $2$, which is $\nat{}$.

   This scenario is so common that many
   ITPs allow to mark some function arguments as implicit arguments and let the
   user systematically avoid passing them. This requires additional
   machinery implemented in the ITP and it has the unfortunate drawback that
   sometimes the user needs to explicitly pass the implicit arguments
   anyway, in particular in case of partial function applications. This special
   situation requires further ad-hoc syntax to turn the implicit argument
   into an explicit one. For instance, if we declare the first argument of $\mbox{Cons}$ implicit,
   then the user can simply write $(\mbox{Cons}~2~l)$ for the term presented above,
   but has to write something like $(@\mbox{Cons}~\nat{})$, in Coq
   syntax, to pass the partial function
   application to some higher order function expecting an argument of type
   $\nat{} \to \mbox{List}~\nat{} \to \mbox{List}~\nat{}$.

   An alternative to implicit arguments is to let the user
   explicitly insert the correct number of placeholders ``?'' to be inferred
   by the system. Series of placeholders are neither aesthetic nor robust to
   changes in the type of the function.

   A similar case occurs during the implementation of tactics. Given a lemma
   $L : H_1 \to \ldots \to H_n \to C$, to apply it the tactic
   opens $n$ new proof obligations by refining the term $(L~?~\ldots~?)$
   where the number of inserted placeholders must be exactly $n$.

   In this paper we propose a new construct to be added to the external
   syntax of ITPs: a vector of placeholders to be denoted by $\mldots$
   and to be used in argument position only. In the actual external syntax of
   Matita we use the evocative symbol ``\ldots'' in place of $\mldots$.
   The semantics associated to $\mldots$ is lazy: an $\mldots$ will be
   expanded to the sequence of placeholders of minimal length that makes the
   application refineable, so that its inferred type matches its
   expected type. In a uni-directional setting no expected type is known in
   advance and the implementation of the lazy semantics would require
   computationally expensive non-local backtracking, which is not necessary
   in the bi-directional case.

   \label{ex:notation}
   Thanks to vectors of placeholders the analysis phase of many primitive
   tactics implementation that was aimed at producing terms with the correct
   number of placeholders can now be totally omitted. Moreover, according to
   our experience, vectors of placeholders enable to avoid the implementation
   of implicit arguments: it is sufficient for the user to insert manually
   or by means of a notation a $\mldots$ before the arguments explicitly passed,
   with the benefit that the $\mldots$ automatically adapts to the case of
   partial function application. For example, using the infix notation $::$
   for $(\mbox{Cons}~\mldots)$, the user can both write $2::l$, which is expanded to
   $(\mbox{Cons}~\mldots~2~l)$ and refined to $(\mbox{Cons}~\nat{}~2~l)$, and pass $::$ to an
   higher order function expecting an argument of type
   $\nat{} \to \mbox{List}~\nat{} \to \mbox{List}~\nat{}$. In the latter case, $::$ is expanded to
   $(\mbox{Cons}~\mldots)$ that is refined to $(\mbox{Cons}~\nat{})$ because of the expected type.
   If $::$ is passed instead to a function expecting an argument of type
   $\forall A. A \to \mbox{List}~A \to \mbox{List}~A$, then $(\mbox{Cons}~\mldots)$ will be expanded
   simply to $\mbox{Cons}$ whose inferred type is already the expected one.\bigskip

The rest of the paper explains the bi-directional refinement algorithm
implemented in Matita~\cite{matita-jar-uitp}. The algorithm is presented
in a declarative programming style by means of deduction rules. Many of the
rules are syntax directed and thus mutually exclusive. The implementation
given for Matita in the functional OCaml language takes advantage of the
latter observation to speed up the algorithm. We will clarify in the text
what rules are mutually exclusive and what rules are to be tried in sequence
in case of failure.

The refinement algorithm is presented progressively and in a modular way.
In Section~\ref{sec:mono} we introduce the mono-directional
type inference algorithm for CIC implemented following the
kernel type checker code (that coincides with type inference if the term
is ground)
%in Coq\footnote{The algorithm really implemented in Coq is in between our mono-direction and bi-direction one and it is obtained by adding to the one presented in Section~\ref{sec:mono} only the $\recE{}\mathrm{-lambda}$ rule from Section~\ref{sec:bi}.} and in the old version
of Matita.  The presentation is already adapted to be extended in
Section~\ref{sec:bi} to bi-directional refinement. In these two sections
the external and internal syntaxes coincides. In Section~\ref{sec:raw}
we augment the external syntax with placeholders and vectors of placeholders.
Finally, in Section~\ref{sec:coercions} we add support for coercions.
In all sections we will prove the correctness of the refinement algorithms
by showing that a term in external syntax accepted by the refiner is turned
into a new term that is accepted by the kernel and that has the expected type.
Moreover, a precise correspondence
is established between the input and output term to grant that the refined
term corresponds to the input one.

For the sake of the reader, Appendix~\ref{sec:app} is taken from~\cite{ck-sadhana} with minor modifications and it shows the
type checking algorithm implemented by the kernel. The syntax of the calculus
and some preliminary notions are also introduced in Section~\ref{sec:pre}
before starting the description of the refinement algorithm.

%%%%%%%%%%%%%%%%%%%%%%%%%%%%%%%%%%%%%%%%%%%%%%%%%%%%%%%%%%%%%%%%%%%%%%%%%%%
%%%%%%%%%%%%%%%%%%%%%%%%%%%%%%%%%%%%%%%%%%%%%%%%%%%%%%%%%%%%%%%%%%%%%%%%%%%
%%%%%%%%%%%%%%%%%%%%%%%%%%%%%%%%%%%%%%%%%%%%%%%%%%%%%%%%%%%%%%%%%%%%%%%%%%%
\section{Preliminaries}\label{sec:pre}

\subsection{Syntax}

We begin introducing the syntax for CIC terms and objects 
in Table~\ref{tab:terms} and some naming conventions. 

To denote constants we shall use 
$c, c_1, c_2$ \ldots; the special case of (co)recursively defined constants will
be also denoted using $f, f_1, f_2$ \ldots;
we reserve $x, y, x_1, x_2$ \ldots for variables;
$t, u, v, t', t'', t_1, t_2$ \ldots for terms;
$T, U, V, E, L, R, T', T'', T_1, T_2$ \ldots for types and we use
$s, s', s_1$ \ldots for sorts.  % and
%$\sigma$ for local explicit substitutions.

%CSC: SINTASSI INCOERENTE: $\sigma$ vs $[t; \ldots; t]$ vs $[x/t; \ldots; x/t]$
%vs $[t/x; \ldots; t/x]$!!! Mi sembrava di averla sistemata?

We denote by \xG{} a context made of variables declarations ($x : T$) or 
typed definitions ($x := t : T$). We denote the capture avoiding substitution
of a variable $x$ for a term $t$ by $[x/t]$. The notation $[x_1/t_1;\ldots;x_n/t_n]$ is for simultaneous parallel substitution.

To refer to (possibly empty) sequences of entities of the same nature, we use an
arrow notation (e.g. $\vector{t}$). For the sake of conciseness, it is sometimes
convenient to make the length of a sequence explicit, while still referring to
it with a single name: we write $\vector{t_n}$ to mean that $\vector{t_n}$ is
a sequence of exactly $n$ elements and, in particular, that it is a shorthand
for $t_1~t_2 \ldots t_n$; the index $n$ must be a natural number (therefore the
notation $\vector{t_{n+1}}$ refers to a non-empty sequence). The arrow notation is
extended to telescopes as in $\vector{(x : t)}$ or $\vector{(x_n : t_n)}$ and
used in binders, (co)recursive definitions and pattern matching branches.

As usual, $\Pi x:T_1.T_2$ is abbreviated to $T_1 \rightarrow T_2$ when
$x$ is not a free variable in $T_2$. 
Applications are n-ary, consisting of a term applied to a non-empty sequence of
terms.

\begin{table}[tbp]
\begin{displaymath}
\begin{array}{llll}
 t & ::= & x & \mbox{identifiers}\\
   &  |  & c & \mbox{constants}\\
   &  |  & I_l & \mbox{inductive types}\\
   &  |  & k & \mbox{inductive constructors}\\
   &  |  & \Prop~|~\Type_u & \mbox{sorts}\\
   &  |  & t~\vector{t_{n+1}} & \mbox{n-ary application}\\
   &  |  & \lambda x:t.t & \mbox{$\lambda$-abstraction}\\
   &  |  & \letin{x}{t}{t}{t} & \mbox{local definitions}\\
   &  |  & \Pi x:t.t & \mbox{dependent product}\\
   &  |  & \match{t}{I_l}{t}
           [k_1~\vector{(x : t)}\Rightarrow
   t~|~\ldots~|~k_n~\vector{(x : t)}\Rightarrow t] & \mbox{case analysis} \\
   &  |  & ?_j [t \;;\; \ldots \;;\; t] & \mbox{metavariable occurrence}\\
   &     & \\
 o & ::= & \letr{f_1~\vector{(x : t)}}{t}{t}{f_n~\vector{(x : t)}}{t}{t} & \mbox{recursive definitions}\\
   &  |  & \letcr{f_1~\vector{(x : t)}}{t}{t}{f_n~\vector{(x : t)}}{t}{t} & \mbox{co-recursive definitions}\\
   &  |  & \definit{c}{t}{t} & \mbox{definitions}\\
   &  |  & \axiom{c}{t} & \mbox{axioms}\\
   &  \hspace{-0.46em}\begin{array}{l}|\\~\\~\\\end{array}  
         & \inductive{\Prod \vector{x_l:t_l}.~}
                {I^1_l}{A}{k_{1,1} : t}{k_{1,m_1} : t}
                {I^n_l}{A}{k_{n,1} : t}{k_{n,m_n} : t}
 & \hspace{-0.46em}\begin{array}{l}\mbox{inductives}\\~\\~\\\end{array}
\end{array}
\end{displaymath}
\caption{CIC terms and objects syntax
\label{tab:terms}}
\end{table}

% The notation 
% $\vector{(x : t)}$ is a possible empty
% sequence of variables bound in the body of the (co-)recursive function
% or in the pattern matching branch.

Inductive types $I_l$ are annotated with the number $l$ of arguments that are
homogeneous in the types of all constructors. For example consider the
inductive type of vectors $\mbox{Vect}$ of arity 
$\Pi A : \Type. \nat{} \to \Type$. It takes two
arguments, a type and a natural number representing the length of the vector.
In the types of the two constructors, $\mbox{Vnil} : \mbox{Vect}~A~0$ and
$\mbox{Vcons} : \Pi m, \mbox{Vect}~A~m \to A \to \mbox{Vect}~A~(m+1)$, every
occurrence of $\mbox{Vect}$  is applied to the same argument $A$, that is
also implicitly abstracted in the types of the constructors.
Thus
$\mbox{Vect}$ has one homogeneous argument, and will be represented by the 
object
\[
\inductiveX{\Prod A : \Type.~}{~\mbox{Vect}}{\nat{}\to \Type}
{\\\qquad\mbox{Vnil} : \mbox{Vect}~A~0 ~~|~~ 
 \mbox{Vcons} : \Pi m, A \to \mbox{Vect}~A~m \to \mbox{Vect}~A~(m+1)}
\]
and referred to with $\mbox{Vect}_1$.
This is relevant for the pattern matching construction, since 
the homogeneous arguments are not bound in the patterns because they are
inferred from the type of the matched term. For example, to pattern match over
a vector $v$ of type $(\mbox{Vect}~\nat{}~3)$ the user writes
 \[\match{v}{\mbox{Vect}_1}{T}
    [\mbox{Vnil} \Rightarrow t_1 ~|~ 
     \mbox{Vcons}~(m:\nat{})~(x:\nat{})~(v':\mbox{Vect}~\nat{}~m) \Rightarrow t_2]\]

\noindent The inductive type $I_l$ in
the pattern matching constructor is (almost) redundant, since distinct 
inductive types have distinct constructors; it is given for the sake
of readability and to distinguish the inductive types with no
constructors. In a concrete implementation it also allows to totally drop
the names of the constructors by fixing an order over them: the $i$-th pattern
will be performed on the $i$-th constructor of the $I_l$ inductive type.

Since inductive
types may have non homogeneous arguments, not every branch is required to have
exactly the same type. The term introduced with the \verb+return+ keyword
is a function that computes the type expected by a particular branch and also
the type of the entire pattern matching.
Variables $\vector{(x : t)}$ are abstracted in the right hand side terms of 
$\Rightarrow$.

The definitions of constants $c$ (including (co)recursive constants $f$),
inductive types $I_l$ and constructors $k$ are collected in the syntactic
category of CIC objects $o$.

Metavariable occurrences, 
represented with $?_j [t_1 \;;\; \ldots \;;\; t_n]$, are missing typed
terms equipped with an explicit local substitution. The index $j$ enables
metavariables to occur non-linearly in the term. 
To give an intuition of the role played by the local substitution, the
reader can think of $?_j [t_1 \;;\; \ldots \;;\; t_n]$ as a call to the,
still unknown, function $?_j$ with actual arguments $t_1\ldots t_n$.
The terms $t_1\ldots t_n$ will be substituted for the formal arguments
of the $?_j$ function inside its body only when it will be known.
%Different occurrences of the same metavariable cannot be instantiated
%with exactly the same term since each occurrence lives in its own context
%that defines the set of variables that can occur free in the term.
%Intuitively, the correct notion of instantiation will be the following:
%to instantiate the instance $?_j [t_1 \;;\; \ldots \;;\; t_n]$ with a term
%$M$ where $M = N[x_1/t_1 \;;\; \ldots \;;\; x_n/t_n]$ for some $N$ whose only
%free variables are $x_1,\ldots,x_n$, one needs to instantiate every other
%occurrence $?_j [t'_1 \;;\; \ldots \;;\; t'_n]$ with
%$N[x_1/t'_1 \;;\; \ldots \;;\; x_n/t'_n]$. The typing rules of the calculus
%will force all the local substitutions to have the same length and the same
%type for the arguments in corresponding position.

We omit to write the local substitution when it is the identity substitution
that sends all variables in the current context with themselves. Thus
$?_j$ will be a shorthand for $?_j[x_1 \;;\; \ldots \;;\; x_n]$ when
$x_1,\ldots,x_n$ are the variables bound in the right order in the context
of the metavariable occurrence.

The CIC calculus extended with metavariables 
has been studied in~\cite{munoz} and the flavor of metavariables 
implemented in Matita is described in~\cite{csc-phd}.

%%%%%%%%%%%%%%%%%%%%%%%%%%%%%%%%%%%%%%%%%%%%%%%%%%%%%%%%%%%%%%%%%%%%%%%%%%%
%%%%%%%%%%%%%%%%%%%%%%%%%%%%%%%%%%%%%%%%%%%%%%%%%%%%%%%%%%%%%%%%%%%%%%%%%%%
%%%%%%%%%%%%%%%%%%%%%%%%%%%%%%%%%%%%%%%%%%%%%%%%%%%%%%%%%%%%%%%%%%%%%%%%%%%
\subsection{Typing rules}

The kernel of Matita is able to handle the whole syntax presented in the
previous section, metavariables included. While we report in the
Appendix~\ref{sec:app} the full set of typing rules implemented by the kernel, 
here we summarise only the ones that will be reused by the refinement 
algorithm. We will give a less formal but more intuitive presentation of
these rules, defining them with a more concise syntax. Moreover, we will
put our definition in linear order, while most of them are actually mutually
recursive. 

\begin{definition}[Proof problem $(\xP{})$]\label{def:pp}
A \emph{proof problem} $\xP{}$ is a finite list of typing declarations of
the form $\xG{_{?_j}} \vdash ?_j : T_{?_j}$.% where for each metavariable that
%occurs in in $\xG{_{?_j}}$ or $T_{?_j}$ there exists a corresponding 
%typing judgement in $\xP{}$.
\end{definition}

A proof problem, as well as a CIC term, can refer to constants, that
usually live in an environment that decorates every typing rule
(as in the Appendix~\ref{sec:app}). In
the following presentation we consider a global well formed
environment $\env{}$, basically a collections of CIC objects defining all
constants and inductive types and associating them to their respective types. No
refinement rule will modify this environment that plays no role in this
presentation. In fact it is the task of the kernel to enable well typed
definitions, inductive types and (co-)recursive functions to enter the
environment.

We thus omit the environment $\env{}$ from the input of every judgment.
We will fetch from it the type $T$ of a constant, inductive type or
constructor $r$ writing $(r : T) \in \env{}$.

We regard CIC as a Pure Type System~\cite{BarendregtH:lawcwt}, and
we denote
by \pts{} the set of axioms. We denote by
$s \in \pts{}$ any sort of the PTS, with $(s_1 : s_2) \in \pts{}$ the fact that
$s_2$ types $s_1$, and with $(s_1, s_2, s_3) \in \pts{}$ the fact that a
product over $s_1$ to $s_2$ has sort $s_3$.  CIC is a full but not functional
PTS: all products are well formed but in
$(s_1, s_2, s_3) \in \pts{}$ it may be $s_2 \neq s_3$. This is because the
calculus is parameterized over a predicative hierarchy $\Type_u$ for $u$
in a given set of universe indexes. In a predicative setting, given
$s_1 = \Type_{u_1}$ and $s_2 = \Type_{u_2}$, $s_3$ is defined as
$\Type_{\max{\{u_1,u_2\}}}$ according to some bounded partial order on the
universe indexes. The details for the actual PTS used in Matita are given
in~\cite{ck-sadhana}. We will often write simply $\Type$ when we are not
interested in the universe index (e.g. in examples).
We also write $\Type_{\top}$ for the biggest sort in the hierarchy, if
any,  or a variable universe to be later fixed to be big enough to
satisfy all the required constraints.

We also write $(s_1, s_2) \in \elim(\pts)$ to check if an element of an
inductive type 
of sort $s_1$ can be eliminated to inhabit a type whose sort is $s_2$. This is relevant for CIC
since the sort of propositions, $\Prop$, is non informative and cannot be
eliminated to inhabit a data type of sort $\Type_u$ for any $u$ (but for few exceptions
described in~\cite{ck-sadhana} Section 6).

Proof problems do not only declare missing proofs (i.e. not all
$T_{?_j}$ have sort $\Prop$) but also missing terms and,
of particular interest for this paper, missing types. 

\begin{definition}[Metavariable substitution environment $(\xS{})$]
A \emph{metavariable substitution environment} $\xS{}$ (called simply
 \emph{substitution} when not ambiguous) is a list of judgments of the
 form
\[\xG{_{?_j}} \vdash ?_j := t_{?_j} : T_{?_j}\]
stating that the term $t_{?_j}$ of type $T_{?_j}$ in $\xG{_{?_j}}$ has been 
assigned to $?_j$.
\end{definition}

We now anticipate the typing judgment of the kernel. A formal definition of
well formedness for $\xP{}$ and $\xS{}$ will follow.

\begin{definition}[Typing judgment]
\label{typj}
Given a term $t$, a 
proof problem $\xP{}$ and a substitution $\xS{}$,
all assumed to be well formed, 
we write
\[\xPSG{}{}{} \vdash t : T\]
to state that $t$ is well typed of type $T$.
\end{definition}

When $\xPSG{}{}{} \vdash t:T$ the type $T$ is well typed and its type
is either a metavariable or a sort $s \in \pts$.

The typing judgment implemented in our kernel is an extension
of the regular typing judgment for CIC~\cite{Werner,mohring,CC}.
It is described in~\cite{ck-sadhana} and reported in the Appendix~\ref{sec:app}.
Here we recall the main differences:

\begin{iteMize}{$\bullet$}
\item Substitution of a regular variable $x$ for a term $t$ 
      is extended with the following rule for metavariables:
\[
  ?_j[t_1\;;\;\ldots\;;\;t_n][x/t] = ?_j[t_1[x/t]\;;\;\ldots\;;\;t_n[x/t]]
\]
\item The conversion relation (denoted by $\downarrow$) 
is enlarged allowing reduction to be performed inside explicit
substitution for metavariables:
\[
 \frac
 {\xG{} \vdash t_i \downarrow t'_i \qquad i \in \{ 1 \ldots n \}}
 {\xG{} \vdash ?_j[t_1 \;;\; \ldots \;;\; t_n] \downarrow ?_j[t'_1 \;;\; \ldots \;;\; t'_n]}
\]
\item The following typing rules for metavariables are added:
\[
\frac{
 \begin{array}{l}
 y_1 : T_1 \;;\; \ldots \;;\; y_n:T_n \vdash ?_j:T_{?_j} \in \xP{} \\
 \xG{} \vdash t_i : T_i[y_1/t_1\;;\; \ldots \;;\; y_{i-1}/t_{i-1}]
  \quad\hfill i \in \{1\ldots n\}
 \end{array}
 }{ \xG{} \vdash \WF{?_j[t_1\;;\;\ldots\;;\;t_n]} }
\]
\[
\frac{
 (y_1 : T_1 \;;\; \ldots \;;\; y_n:T_n \vdash ?_j:T_{?_j}) \in \xP{}
 \quad \xG{} \vdash \WF{?_j[t_1\;;\;\ldots\;;\;t_n]}
 }{\xG{}\vdash ?_j[t_1\;;\;\ldots\;;\;t_n] :
 T_{?_j}[y_1/t_1\;;\;\ldots\;;\;y_n/t_n]}
\]
\end{iteMize}

\noindent Moreover, in many situations a metavariable occurrence is also accepted
as a valid sort, marking it so that it cannot be instantiated with anything
different from a sort. This additional labelling will be omitted, being marginal for the refinement algorithm.

The technical judgment $\xG{} \vdash \WF{?_j[t_1\;;\;\ldots\;;\;t_n]}$ states
that a metavariable occurrence $?_j[t_1\;;\;\ldots\;;\;t_n]$ 
is well formed in $\xG{}$.

In all the previous rules we assumed access to a global well formed proof
problem $\xP{}$ and substitution $\xS{}$. Both $\xP{}$ and $\xS{}$ are
never modified by the judgments implemented in the kernel.

We now present the well formedness conditions, corresponding to the judgments
\mbox{$~\vdash \mbox{WF}$} presented in the Appendix~\ref{sec:app}.

\begin{definition}[Metavariables of term/context ($\metas$)]
\label{def:metasofterm}
Given a term $t$, $\metas(t)$ is the set of metavariables
occurring in $t$. Given a context $\xG{}$, $\metas(\xG{})$
 is the set of metavariables occurring in $\xG{}$.
\end{definition}

The function $\metas$ is at the base of the order relation
defined between metavariables.

\begin{definition}[Metavariables order relation ($\ll_\xP{}$)]
Let $\xP{}$ be a proof problem. Let $<_\xP{}$ be the
relation defined as: $?n_1 <_\xP{} ?n_2$ iff $?n_1 \in \metas(\xG{_{?n_2}})
\cup \metas(T_{?n_2})$. Let $\ll_\xP{}$ be the transitive closure
of $<_\xP{}$.
\end{definition}

\begin{definition}[Valid proof problem]
A proof problem $\xP{}$ is a \emph{valid proof problem} if and only
if $\ll_\xP{}$ is a strict partial order (or, equivalently, if and only if
$\ll_\xP{}$ is an irreflexive relation).
\end{definition}

The intuition behind $\ll_\xP{}$ is that the smallest $?_j$  (or one of them
since there may be more than one) does not
depend on any other metavariable (e.g. $\metas(\xG{_{?_j}}) =
\emptyset$ and $\metas(T_{?_j}) = \emptyset$ where $\xG{_{?_j}}
\vdash ?_j : T_{?_j} \in \xP{}$). Thus instantiating every minimal $?_j$ with
a metavariable free term will give a new $\xP{}$ in which there is at
least one $?_j$ not depending on any other metavariable (or $\xP{}$ is
empty).
This definition is the key to avoid circularity in the following definitions.

In the rules given in Appendix~\ref{sec:app} the partial order is left
implicit by presenting $\Sigma$ as an ordered list. However, as proved by
Strecker in his Ph.D. thesis~\cite{strecker}, the order is not preserved by
unification and thus in any realistic implementation $\Sigma$ is to be
implemented as a set and the fact that $\ll_\xP{}$ remains a partial order
must be preserved as an invariant.

\begin{definition}[Well formed context (\WF{\xG{}})]
Given a well formed proof problem $\xP{}$, a context 
$\xG{} = y_1 : T_1, \ldots, y_n : T_n$ is well formed 
(denoted by \WF{\xG{}}) if
$\metas(\xG{}) \subseteq \xP{}$ and for every $i$ 
\[y_1 : T_1, \ldots y_{i-1} : T_{i-1} \vdash y_i : T_i\]
\end{definition}

\begin{definition}[Well formed proof problem $(\WF{\xP{}})$]
A valid proof problem $\xP{}$ is a well-formed proof problem
(denoted by \WF{\xP{}})
if an only if for all $(\xG{_{?_j}} \vdash ?_j : T_{?_j}) \in
\xP{}$ we have $\xP{}, \xG{_{?_j}} \vdash T_{?_j} : s$ and $s \in \pts{}$.
\end{definition}

\begin{definition}[Well formed substitution (\WF{\xS{}})]
Given a well formed proof problem $\xP{}$, a substitution
$\xS{}$ is well formed (denoted by \WF{\xS{}}) 
if for every $(\Gamma_{?_j} \vdash ?_j := t_{?_j} : T_{?_j}) \in \xS{}$ we have
 $\xP,\emptyset,\xG{_{?_j}} \vdash t_{?_j} : T_{?_j}$.
\end{definition}

The well formedness definitions given so far are actually implemented by the
kernel in a more precise but less intuitive way. We thus refer to the kernel
judgments in the following definition, that will be used in the specification
of all refinement rules.

\begin{definition}[Well formed status (\WF{\xPSG{}{}{}})]
Given a proof problem $\xP{}$, a substitution $\xS{}$ and a
context $\xG$, the triple $\xPSG{}{}{}$ is well formed
(denoted by $\WF{\xPSG{}{}{}}$) when $\WF{\xP{}}$ and
$\WF{\xS{}}$ and $\WF{\xG{}}$.
\end{definition}

We shall sometimes omit \xG{}, considering it equal to a default, well
formed context, like the empty one.
The recursive operation of applying a substitution $\xS{}$ to a  term $t$ 
is denoted by $\xS{(t)}$ and acts as the identity for any term but
metavariables contained in $\xS{}$, on which it behaves as follows:
\[\xS{(?_j[t_1\;;\;\ldots;\;t_n])} = t_{?_j}[y_1/t_1\;;\;\ldots;\;y_n/t_n]
\quad \mbox{when} \quad 
(y_1:T_1; \ldots; y_n:T_n \vdash ?_j := t_{?_j} : T_{?_j}) \in \xS{}
\]
\noindent Note that, thanks to the extensions to the type checking rules made in
Definition~\ref{typj}, substitution application is type preserving.
Substitutions do apply also to well formed proof problems in the following way:
\[
 \xS{}(\xG{_{?_j}} \vdash ?_j : T_{?_j}) = \xS{}(\xG{_{?_j}}) \vdash ?_j :
   \xS{}(T_{?_j}) \qquad (\mbox{for each }?_j \in \xP{})
\]

The substitution application operation is seldom used explicitly, since
all judgments take as input and give back a substitution. Nevertheless 
it will be used in the examples.

\begin{definition}[Weak-head normalization ($\whdlabel$)]
Given a context $\xG{}$, substitution $\xS{}$ and proof problem $\xP{}$,
all assumed to be well formed, 
it computes the weak head normal form of a well typed term 
$t$ according to the reduction rules of CIC. It is denoted by:
\[\whd{}{}{}{t}{t'}\]
\end{definition}

\noindent Note that $?_j$ is in weak head normal form iff $~?_j \not\in \xS{}$.
%Also note that in the Appendix, the rules for the kernel refer to the
%same normalization function with the syntax $\rhd_{\mbox{whd}}$.

% \[\begin{array}{c}
% \{\WF{\xG{}},\WF{\xS{}},\WF{\xP{}}, \xS{}(\xG{}) \vdash \xS{}(t) : T \} \\
% \whd{}{}{}{t}{t'} \\
% \{ \xS{}(\xG{}) \vdash \xS{}(t') : T,  \xS{}(\xG{}) \vdash \xS{}(t) \downarrow  \xS{}(t'), t' \mbox{ is in weak-head normal form} \}
% \end{array}\]

By abuse of notation we will write
$\whd{}{}{}{t_1}{\Pi x_1:T_1 \ldots \Pi x_n:T_n.t_{n+1}}$ to mean
that for all $i \in \{1 \ldots n\}$ $\whd{}{}{; x_1 : T_1; \ldots ;x_{i-1} : T_{i-1}}{t_i}{\Pi x_i:T_i.t_{i+1}}$
and $\whd{}{}{}{t_{n+1}}{t_{n+1}}$. Such repeated use of weak head computation
to produce spines of dependent products occur frequently in the kernel and
in the refinement rules, especially when dealing with inductive types.

\begin{definition}[Conversion ($\downarrow$)]
Given a proof problem $\xP{}$, substitution $\xS{}$ and
context $\xG{}$,
all assumed to be well formed, 
and two terms $t_1$ and $t_2$, it verifies if $t_1$ and $t_2$ have a common
normal form according to the rules of CIC given in Appendix~\ref{sec:app}.
It is denoted by: \[\xPSG{}{}{} \vdash t_1 \downarrow t_2\]
\end{definition}

% \subsection{Implementation remarks}

\section{Mono-directional refinement}\label{sec:mono}

We now present the mono-directional refinement algorithm
for CIC implemented in the old versions of Matita (0.5.x) and directly
inspired by the rules for type checking implemented in the kernel.
In this section
we assume the external syntax to coincide with the syntax of terms. Hence
the algorithm actually performs just type inference. Nevertheless, we
already organize the judgments in such a way that the latter extension to
bi-directionality will be achieved just by adding new typing rules.

\subsection{Specification}

To specify what is a refinement algorithm we must first introduce the notion
of proof problem refinement. Intuitively, a pair (proof problem, substitution)
is refined by another pair when the second is obtained by reducing some
proof obligations to new ones. It thus represents an advancement in the proof
discovery process.

\begin{definition}[Proof problem refinement $(\refineslabel)$]
We say that $\xP{'},\xS{'}$ \emph{refines} $\xP{},\xS{}$
(denoted by $\refines{\xP{'},\xS{'}}{\xP{},\xS{}}$)
when
$\xS{} \subset \xS{'}$ and
for every $(\xG{_{?_j}} \vdash ?_j : T_{?_j}) \in \xP{}$ either
$(\xG{_{?_j}'} \vdash ?_j : T'_{?_j}) \in \xP{'}$ or 
$(\xG{_{?_j}'} \vdash ?_j := t_{?_j} : T'_{?_j}) \in \xS{'}$
where $\xG{_{?_j}'} = \Phi'(\xG{_{?_j}})$ and
$T'_{?_j} = \Phi'(T_{?_j})$.
\end{definition}

%\begin{lemma}
%All rules given in the paper of the form
%\[(rule)\quad \frac{\ldots}{\xP{}, \xS{}, \ldots \stackrel{\ldots}{\leadsto} \ldots \xP{'}, \xS{'}}\]
%(where $\xS{'}$ is equal to $\xS{}$ when omitted) are proof problem refinements.
%\end{lemma}

%\begin{definition}[Refined term]
%A refined term in a well formed context \xG{} is a term $t$, a type $T$, a
%possibly empty proof problem $\xP{}$ and a substitution $\xS{}$ such that $t$
%it is well typed of type $T$ and $T$ is a type:  
%\[ \xPSG{}{}{} \vdash t : T \quad\mathrm{and}\quad \xPSG{}{}{} \vdash T : s \]
%\end{definition}

%\noindent
%Refined terms are the output of the \emph{refiner}.

\begin{specification}[Refiner in type inference mode (\reclabel)]
A refiner algorithm $\mathcal{R}$ in type inference mode $\Uparrow$
takes as input a proof problem, substitution and context, all assumed to be
well formed, and a term $t$. It fails or gives in output a new proof problem,
a new substitution, a term $t'$ and a type $T'$.
It is denoted by:
\[\refinex{}{}{}{t}{t'}{T'}{'}{'}\]
Precondition:
\[ \WF{\xPSG{}{}{}} \]
Postcondition (parametric in $\preccurlyeq$):
\[\WF{\xPS{'}{'}} \lland \refines{\xP{'},\xS{'}}{\xP{},\xS{}} \lland \xPSG{'}{'}{} \vdash t' : T' \lland t' \preccurlyeq t
\]
\end{specification}

The specification is parametric in the $\preccurlyeq$ relation that establishes
a correspondence between the term $t$ to be refined and the refiner output
$t'$. In order to prove correctness, we are only interested in admissible
$\preccurlyeq$ relations defined as follows.

\begin{definition}[Admissible relations ($\preccurlyeq$)]
A partial order relation $\preccurlyeq$ is \emph{admissible} when for every term $t_1$ in
external syntax and $t_2$ and $T$ in internal syntax and 
for every variable $x$ occurring free only
linearly in $T$ we have that
$t_1 \preccurlyeq t_2$ implies $T[x/t_1] \preccurlyeq T[x/t_2]$.
\end{definition}
Admissibility for equivalence relations correspond to asking the equivalence
relation to be a congruence.

When the external syntax corresponds to the term syntax and coercions
are not considered, we can provide an implementation that satisfies the
specification by picking the identity for the $\preccurlyeq$ relation.
Combined with $\refines{\xP{'},\xS{'}}{\xP{},}{\xS{}}$, the two postconditions
imply that $\xS{}(t')$ must be obtained from $t$ simply by instantiating some
metavariables. In Sections~\ref{sec:raw} and~\ref{sec:coercions}, we shall use
weaker definitions of $\preccurlyeq$ than the identity, allowing
replacement of (vectors of) placeholders with (vectors of) terms and the
insertion of coercions as results of the refinement process. All the
$\preccurlyeq$ relations considered in the paper will be large partial orders
over terms of the external syntax (that always include the internal syntax).

We will now proceed in presenting an implementation of a refinement algorithm
in type inference mode~\rec{}. The implementation is directly inspired by
the type checking rules used in the kernel.
However, since refinement deals with terms containing flexible parts,
conversion tests need to be replaced with unification tests.  In a higher order and dependently
typed calculus like CIC, unification is in the general case undecidable. What is
usually implemented in interactive theorem provers is an essentially fist order
unification algorithm, handling only some simple higher order cases. The
unification algorithm implemented in Matita goes beyond the scope of this
paper, the interested reader can find more details
in~\cite{csc-phd,unification-hints}. Here we just specify the expected
behavior of the unification algorithm.

\begin{specification}[Unification ($\UNI-\,\unilabel$)]
An unification algorithm takes as input
a proof problem, a substitution and a context, all assumed to be
well formed, and two well typed terms 
$t_1$ and $t_2$. It fails or gives in output a new proof problem and 
substitution.
It is denoted using the following notation where $\bullet$ can either be
$=$ or be omitted. In the former case universe cumulativity (a form
of sub-typing) is not taken in account by unification.
\[\unifx{_\bullet}{}{}{}{t_1}{t_2}{'}{'}\]
Precondition:
 \[ \WF{\xPSG{}{}{}} \lland \xPSG{}{}{} \vdash t_1 : T_1 \lland \xPSG{}{}{} \vdash t_2 : T_2 \]
Postcondition:
 \[ \WF{\xPS{'}{'}} \lland \refines{\xP{'},\xS{'}}{\xP,\xS} \lland
    \xPSG{'}{'}{} \vdash t'_1 \downarrow_\bullet t'_2 \]
\end{specification}

\subsection{Implementation}

\subsubsection{Additional judgments}

For the sake of clarity we prefer to keep the same structure for the 
mono and bi-directional refiners. We thus give the definition of 
some functions that are trivial in the mono-directional case, but will 
be replaced by more complex ones in the following sections.

Even if we presented the syntax of CIC using the same category terms,
types and sorts, some primitive constructors (like the 
$\lambda$ and $\Pi$ abstractions) expect some arguments to be 
types or sorts, and not terms. A type level enforcing algorithm forces a term in
external syntax to be refined to a valid type.

\begin{specification}[Type level enforcing ($\mathcal{F}$)]
A type level enforcing algorithm takes as input a proof problem $\xP{}$,
a substitution $\xS{}$ and a context $\xG{}$, all assumed to be well formed,
and a term $T$. It fails or it returns a new term $T'$, a sort $s$,
a new substitution $\xS{'}$ and proof problem $\xP{'}$.
It is denoted by:\\
\[\forcetotype{}{}{}{T}{T'}{s}{'}{'}\]
Precondition:
 \[ \WF{\xPSG{}{}{}} \]
Postcondition (parametric in $\preccurlyeq$):
 \[ \WF{\xPS{'}{'}} \lland \refines{\xP{'},\xS{'}}{\xP,\xS} \lland
    \xPSG{'}{'}{} \vdash T' : s \lland s \in \pts \lland
    T' \preccurlyeq T
 \]
\end{specification}

Note that one may want to accept a metavariable as the sort $s$, eventually
labelling it in such a way that the unification algorithm will refuse to
instantiate it with a different term. The choice must be consistent with the
one taken in the implementation of the kernel.

The task of checking if a term has the right type is called refinement
in type forcing mode and it will be denoted by $\recE{}$.
In the mono-directional case, $\recE{}$ will be simply implemented
calling the $\coer{}$ algorithm that will handle coercions in Section~\ref{sec:coercions} but which, at the moment, only verifies that no coercion is needed
by calling the unification procedure.

\begin{specification}[Explicit cast ($\mathcal{C}$)]
A cast algorithm takes as input a proof problem $\xP{}$, a substitution
$\xS$ and a context $\xG{}$, all assumed to be well formed, and a term
$t$ with its inferred type $T$ and expected type $T'$.
It fails or it returns a new term $t'$ of type $T'$, a new
proof problem $\xP{'}$ and substitution $\xS{'}$.
It is denoted by:\\
\[\unifcoercex{}{}{}{t}{T}{T'}{t'}{'}{'}\]
Precondition:
 \[ \WF{\xPSG{}{}{}} \lland
    \xPSG{}{}{} \vdash t : T \lland
    \xPSG{}{}{} \vdash T' : s
 \]
Postcondition (parametric in $\preccurlyeq$):
 \[ \WF{\xPS{'}{'}} \lland \refines{\xP{'},\xS{'}}{\xP,\xS} \lland
    \xPSG{'}{'}{} \vdash t' : T' \lland
    t' \preccurlyeq t
 \]
\end{specification}

\begin{specification}[Refiner in type forcing mode (\recElabel)]
A refiner algorithm $\mathcal{R}$ in type forcing mode $\Downarrow$
takes as input a proof problem $\xP{}$, a substitution
$\xS$ and a context $\xG{}$, all assumed to be well formed, and a term
$t$ together with its expected well formed type $T$. 
It fails or returns a term $t'$
of type $T$, a new
proof problem $\xP{'}$ and substitution $\xS{'}$.
It is denoted by:\\
\[\refinexE{}{}{}{t}{T}{t'}{'}{'}\]
Precondition:
\[ \WF{\xPSG{}{}{}} \lland \xPSG{}{}{} \vdash T: s\]
Postcondition (parametric in $\preccurlyeq$):
\[\WF{\xPS{'}{'}} \lland \refines{\xP{'},\xS{'}}{\xP{},\xS{}} \lland \xPSG{'}{'}{} \vdash t' : T \lland t' \preccurlyeq t \]
\end{specification}

\subsubsection{Notational conventions}
The arguments \xP{} and \xS{} will be taken as input and returned as output
in all rules that define the refiner algorithm. To increase legibility
we adopt the following notation, letting \xP{} and \xS{} be implicit.
Each rule of the form

\[
\infrule[(rule)]{
  \xG{} \vdash t \leadsto t'\\
  \xG{} \vdash t' \leadsto t''
}{
  \xG{} \vdash t \leadsto t''
}
\]\smallskip

\noindent has to be interpreted as:
\[
\infrule[(rule)]{
  \xPSGp{}{}{} \vdash t \leadsto t'~ (\xP{'}, \xS{'}) \\
  \xPSGp{'}{'}{} \vdash t' \leadsto t''~ (\xP{''}, \xS{''}) 
}{
  \xPSGp{}{}{} \vdash t \leadsto t''~ (\xP{''}, \xS{''})
}
\]\smallskip

\noindent Moreover we will apply this convention also to rules not returning 
$\xP{}$ or $\xS{}$ as if they were returning the \xP{} or \xS{} 
taken as input.

Note that the $\xP{'}$ and $\xS{'}$ returned by
all rules considered in this paper are well formed and are also a proof
problem refinement of
the $\xP{}$ and $\xS{}$ provided as input. Being a proof problem
refinement is clearly a transitive relation. Thus we have for free that
all the omitted pairs (proof problem, substitution) are refinements of
the initial ones.

\subsubsection{Role of the relations and their interaction}
In this paragraph we shortly present the role played by the relations
$\rec{}$, $\recE{}$, $\coer{}$ and $\forcetotypelabel{}$ introduced so far
and the auxiliary ones $\eatprodslabel{}$ and $\unimany{}$ that will be
specified when needed.

The relation $\rec{}$ links a term with its inferred type, while $\recE{}$
links a term with the type expected by its context. $\recE{}$ will thus
exploit the extra piece of information not only checking that the inferred type
unifies with the expected one, but also propagating this information to
its recursive calls on subterms (when possible). $\rec{}$ and $\recE{}$ will
be defined in a mutually recursive way.

The relation $\forcetotypelabel{}$ links a term with its refinement
asserting that the refinement is a type. This is relevant when typing
binders like $(\lambda x : t. t')$, where $t$ is required to be a type.
In its simplest formulation the relation is a simple assertion,
linking a type with itself.
In Section~\ref{sec:coercions} the refinement relation $\preccurlyeq{}$ 
will admit to link a term $t$ that is not a type with a function applied to 
$t$ that turns its input into a type. 
For example $t$ may be a record containing a type and
$\forcetotypelabel{}$ may link it with $(\pi_n~t)$, where $\pi_n$ is the
projection extracting the type from the record. $\forcetotypelabel{}$
is recursively defined in terms of $\rec{}$

The relation $\coer{}$ links a term $t$, its inferred type $T_1$ and the type 
expected by its context $T_2$ with a refinement of the term $t'$ 
asserting that the refined term has type $T_2$. 
In its simple formulation the relation is a simple assertion that
$T_1$ and $T_2$ are the same and thus links $t$ with itself.
In Section~\ref{sec:coercions} the refinement relation $\preccurlyeq{}$
will admit to 
explicitly cast $t$. For example a natural number $n$ of type $\mathbb{N}$
may be casted into the rationals $\mathbb{Q}$ refining it to
$(\lambda x:\mathbb{N}.x/1)~n$. The $\coer{}$ relation is non recursive.

The relations $\eatprodslabel{}$ and $\unimany{}$ are auxiliary relations
only used to ease the presentation of the $\rec{}$ and $\recE{}$ relations
in the case of applications. Both auxiliary relations are thus recursively
defined with $\rec{}$ and $\recE{}$.

\subsubsection{Rules for terms}

We now give an implementation for the refiner in both modes and for the
auxiliary judgments. The implementation is parametric on the unification
algorithm, that is not described in this paper.
\[
\infrule[(\coer{}\mathrm{-ok})]
 {\Unifx{}{}{}{}{T_1}{T_2}{'}{'}}
 {\Unifcoercex{}{}{}{t}{T_1}{T_2}{t}{'}{'}} 
\qquad
\infrule[(\recE{}\mathrm{-default})]
 {\Refinex{}{}{}{t}{t'}{T'}{'}{'}\\
  \Unifcoercex{'}{'}{}{t'}{T'}{T}{t''}{''}{''}}
 {\RefinexE{}{}{}{t}{T}{t''}{''}{''}}
\]\bigskip
\bugnelcodice{nel codice si prova l'unificazione prima di chiamare la try\_coercions,ma questa inizia con un'unificazione! Quindi si spreca codice e tempo.}
\[
\infrule[(\forcetotypelabel{}\mathrm{-ok})]
 {\Refinex{}{}{}{T}{T'}{s}{'}{'} \qquad
   s_1 \in \pts\\
  \Unifcoercex{'}{'}{}{T'}{s}{s_1}{T''}{''}{''}}
 {\Forcetotype{}{}{}{T}{T''}{s}{''}{''}}\bigskip
\]

\noindent Note that $s_1$ is arbitrary, and the actual code prefers the predicative
sorts $\Type_u$ over \Prop{}. This is the only rule defined in this section
to be non syntax oriented: in case of an incorrect choice of $s_1$, backtracking
is required. The actual algorithm implemented in Matita performs the choice
of $s_1$ lazily to remain backtracking free\footnote{Laziness will be no longer
sufficient to avoid backtracking when we will add additional rules to handle coercions in
Section~\ref{sec:coercions}.}.

%CSC: questa regola non \`e deterministica! Non si sa chi sia $s_1$. Nel codice, in verit\`a, passiamo il primo universo tornato dalla get\_universe e un booleano per evitare l'unificazione. Cosa descrivere nell'articolo?\\

\[
\infrule[(\rec{}\mathrm{-variable})]
 {(x : T) \in \xG{} \quad \mbox{or} \quad (x := t : T) \in \xG{}}
 {\Refinex{}{}{}{x}{x}{T}{}{}}
\]\bigskip
\[
\infrule[(\rec{}\mathrm{-constant})]
 {(r : T) \in \env{} \qquad r \in \{ k, I, c \}}
 {\Refinex{}{}{}{r}{r}{T}{}{}}
\]\bigskip
\[
\infrule[(\rec{}\mathrm{-sort})]
 {(s_1 : s_2) \in \pts{}}
 {\Refinex{}{}{}{s_1}{s_1}{s_2}{}{}}
\]\bigskip
\[
\infrule[(\rec{}\mathrm{-meta})]
 {(\xG{_{?_j}} \vdash ?_j : T_{?_j}) \in \xP \quad \mbox{or} \quad
  (\xG{_{?_j}} \vdash ?_j := t_{?_j} : T_{?_j}) \in \xS\\
  %\xG{_{?_j}} = x_1:T_1;\ldots ; x_n:T_n\\
  \xG{_{?_j}} = \vector{x_n:T_n}\\
  %\RefinexE{_i}{_i}{}{t_i}{T_i[x_1/t'_1;\ldots;x_{i-1}/t'_{i-1}]}{t_i'}{_{i+1}}{_{i+1}}
  \RefinexE{_i}{_i}{}{t_i}{T_i[\vector{x_{i-1}/t'_{i-1}}]}{t_i'}{_{i+1}}{_{i+1}}
  \quad i \in \{1 \ldots n\}\\
 }
 %{\Refinex{_1}{_1}{}{?_j[t_1;\ldots;t_n]}{?_j[t'_1;\ldots;t'_n]}{T_{?_j}[x_1/t'_1;\ldots;x_n/t'_n]}{_{n+1}}{_{n+1}}}
 {\Refinex{_1}{_1}{}{?_j[\vector{t_n}]}{?_j[\vector{t'_n}]}{T_{?_j}[\vector{x_n/t'_n}]}{_{n+1}}{_{n+1}}}
\]\medskip

\noindent Note that the operation of firing a $\beta$-redex must commute with the
operation of applying a substitution \xS{}.
Consider for example the term $v = (\lambda x.?_j[x])~u$ and the substitution
$\xS{} = \{ x:T \vdash ?_j := t(x) : T(x)\}$. 
If one applies the substitution first, and then reduces the redex
obtains $t(u)$, whose type is $T(u)$.
If one fires the redex fist, the fact that $x$ is substituted by $u$ in 
$?_j$ is recorded in the local substitution attached to the metavariable
instance. Indeed 
$\emptyset,\emptyset,\emptyset \vdash v\; \whdlabel ?_j[u]$ and 
$\xS(?_j[u]) = t(u) : T(u)$. Therefore $?_j[u]$ is given the type
$T(u)$ by the rule $(\rec{}\mathrm{-meta})$.
\[
\infrule[(\rec{}\mathrm{-letin})]
 {\Forcetotype{}{}{}{T}{T'}{s}{'}{'} \\
  \RefinexE{'}{'}{}{t}{T'}{t'}{''}{''}\\
  \Refinex{''}{''}{; x := t' : T'}{u}{u'}{T_2}{'''}{'''}}
 {\Refinex{}{}{}{\letin{x}{T}{t}{u}}{\letin{x}{T'}{t'}{u'}}{T_2[x/t'_1]}{'''}{'''}}
\]\bigskip
\[
\infrule[(\rec{}\mathrm{-lambda})]
 {\Forcetotype{}{}{}{T_1}{T_1'}{s_1}{'}{'} \\
  \Refinex{'}{'}{; x : T_1'}{t}{t'}{T}{''}{''}}
 {\Refinex{}{}{}{\lambda x:T_1.t}{\lambda x:T_1'.t'}{\Pi x:T_1'.T}{''}{''}}
\]\bigskip
\[
\infrule[(\rec{}\mathrm{-product})]
 {\Forcetotype{}{}{}{T_1}{T_1'}{s_1}{'}{'} \\
  \Forcetotype{}{}{; x : T_1'}{T_2}{T_2'}{s_2}{''}{''} \\
  (s_1,s_2,s_3) \in \pts{}}
 {\Refinex{}{}{}{\Pi x:T_1.T_2}{\Pi x:T_1'.T_2'}{s_3}{''}{''}}
\]\smallskip

\noindent We now state the correctness theorem holding for all the
 rules presented so far and for the few ones that will follow. The
 proof is partitioned in the following way: here we state the theorem,
 introduce the proof method we adopted and prove the theorem for the
 simple rules presented so far. Then we will introduce more complex
 rules, like the rule for application, and we will prove for each of
 them the correctness theorem.

\begin{thm}[Correctness]
The \coer{}, \forcetotypelabel{}, \rec{}, \recE{}, \eatprodslabel{},
 and \unimany{} algorithms defined by the set of rules presented in
 this section obey their specification for all admissible
 $\preccurlyeq$ relations that include the identity for terms in the
 internal syntax. In particular, the algorithms are correct when the
 identity relation is picked for $\preccurlyeq$.
\end{thm}
\begin{proof}
We assume the unification algorithm to be correct w.r.t. its own
 specification.  For every judgment, the proof is by induction on the
 proof tree.  For each rule, we assume that the precondition of the
 judgment holds for the rule conclusion and that the appropriate
 postcondition holds by induction hypothesis for every hypothesis. We
 need to prove that the precondition of every hypothesis holds and
 that the postcondition of the conclusion holds too.  The proofs are
 mostly trivial for the rules presented so far.  In particular, the
 proof for each rule $\rec\mathrm{-name}$ or $\recE\mathrm{-name}$
 follows from the corresponding rule $\mathcal{K}\mathrm{-name}$
 reported in the Appendix~\ref{sec:app}.\\ We will shortly introduce
 the rules dealing with applicatios together with their correctness
 proofs since applications are handled slightly differently from the
 way they are processed by the kernel.
\end{proof}

The next rule deals with applications which are $n$-ary in our implementation
of CIC.
In a calculus without dependent types, $n$-ary applications could be
handled simply by putting the head function type in the form
of a spine of $n$ products and then by verifying that the type of each argument
matches the corresponding expected type. In the presence of dependent types,
however, it is possible to write functions whose arity depends on the
arguments passed to the function. For instance, a function $f$ could be
given type $\forall n:\nat{}. (\mathrm{repeat}~\nat{}~n)$ where 
$(\mathrm{repeat}~\nat{}~n)$ reduces to
$\nat{} \to \ldots \to \nat{}$ where the number of products is exactly $n$.
For this reason, the only possibility is to process applications one argument
at a time, checking at every step if the function still accepts more
arguments. We implement this with an additional judgment
\[\Eatprods{}{}{}{t~\vector{(x_i := v_i : T_i)} :\; T}{\vector{u_n}}{v}{V}{'}{'}\]
called ``eat products'' to be specified and implemented immediately after the
$(\rec{}\mathrm{-appl})$ rule.
\[
\infrule[(\rec{}\mathrm{-appl})]
 {\Refinex{}{}{}{t}{t'}{T}{'}{'}\\
  \Eatprods{'}{'}{}{t' :\; T}{\vector{u_{n+1}}}{v}{V}{''}{''}}
 {\Refinex{}{}{}{t~\vector{u_{n+1}}}{v}{V}{''}{''}}
\]\bigskip
\bugnelcodice{Codice useless: Caso vettore + tipo atteso: MANCA, ma credo sia codice useless perch\`e la eat\_prods farebbe la stessa cosa!}

\begin{specification}[Eat products (\eatT)]
The $\eatprodslabel$ algorithm refines an $n$-ary application by consuming
an argument at a time. It takes as input a proof problem $\xP{}$, a substitution
$\xS$ and a context $\xG{}$, all assumed to be well formed, the part of the
already processed application $t~(x_1 := v_1 : T_1)~\ldots~(x_r := v_r : T_r)$ together with its type $T$, and the list of arguments yet to be checked.
The notation $(x_i := v_i : T_i)$ means that the $i$-th already processed
argument has type $T_i$ and is consumed by a product that binds the variable
$x_i$.
The algorithm fails or returns the refined application $v$ together with its type $V,$
a new substitution $\xS{'}$ and proof problem $\xP{'}$.
It is denoted by:
\[\eatprods{}{}{}{t~\vector{(x_r := v_r : T_r)} :\; T}{\vector{u_k}}{v}{V}{'}{'}\]
Precondition:
 \[ \WF{\xPSG{}{}{}} \lland
    \xPSG{}{}{} \vdash v_i : T_i \quad i \in \{1 \ldots r \} \lland
    \xPSG{}{}{} \vdash t~v_1~\ldots~v_r : T
 \]
Postcondition (parametric in $\preccurlyeq$):
 \[ \WF{\xPS{'}{'}} \lland \refines{\xP{'},\xS{'}}{\xP,\xS} \lland
    \xPSG{'}{'}{} \vdash v : V \lland
    v \preccurlyeq t~v_1~\ldots~v_r~u_1~\ldots~u_k \]
\end{specification}

\noindent
The applicative case is one of the two most complicated rules. Moreover,
the refinement algorithm for the application does not mimic the one used
in the kernel. Therefore we show the correctness of the
$(\rec{}\mathrm{-appl})$ rule and of the implementation of the \eatprodslabel{}
algorithm.

\begin{proof}[Correctness of $(\rec{}\mathrm{-appl})$]
The only rule precondition is $\WF{\xPSG{}{}{}}$ that is also the
precondition for the first premise.
By induction hypothesis on the first premise we know
$\xPSG{'}{'}{} \vdash t': T$ where $\xPS{'}{'}$ are implicitly
returned by the first call and passed to the second one. Moreover
$\refines{\xPS{'}{'}}{\xPS{}{}}$ and $t' \preccurlyeq t$.
Therefore the preconditions for the second premise are satisfied.
By induction hypothesis on the second premise we know
$\xPSG{''}{''}{} \vdash v : V$ where $\xPS{''}{''}$ are implicitly
returned by the second call and by the rule as a whole. Moreover
$\refines{\xPS{''}{''}}{\xPS{'}{'}}$ and $v \preccurlyeq t'~\vector{u_{n+1}}$.
By transitivity of proof problem refinement, we also have
$\refines{\xPS{''}{''}}{\xPS{}{}}$. Moreover, since $\preccurlyeq$ is
admissible, we also have
$v \preccurlyeq t'~\vector{u_{n+1}} \preccurlyeq t~\vector{u_{n+1}}$.
All post-conditions have been proved and therefore the rule is correct.
\end{proof}

The \eatprodslabel{} algorithm is implemented as follows.

\bugnelcodice{Il codice \`e pi\`u semplice e molto pi\`u stupido
nell'inferenza di tipi dipendenti. Bug risolto nel papero e individuato da Enrico.}

\[
\infrule[(\eatT\mathrm{-empty})]
 {}
 %{\Eatprods{}{}{}{t~(x_1 := v_1 : T_1)~\ldots~(x_r := v_r : T_r) :\; T}{}{t~v_1~\ldots~v_r}{T}{}{}}
 {\Eatprods{}{}{}{t~\vector{(x_r := v_r : T_r)} :\; T}{}{t~\vector{v_r}}{T}{}{}}
\]\bigskip

\noindent Correctness of the $(\eatT{}\mathrm{-empty})$ rule is trivial.

\[
\infrule[(\eatT\mathrm{-prod})]
 {\Whd{}{}{}{T}{\Pi x:U_1.T_1}\\
  \RefinexE{}{}{}{u_1}{U_1}{u_1'}{'}{'}\\
  %\Eatprods{'}{'}{}{t~(x_1 := v_1 : T_1)~\ldots~(x_r := v_r : T_r)~(x := u_1' : U_1) :\; T_1[x/u_1']}{u_2~\ldots~u_n}{v}{V}{''}{'}}
  \Eatprods{'}{'}{}{t~\vector{(x_r := v_r : T_r)}~(x := u_1' : U_1) :\; T_1[x/u_1']}{\vector{u_n}}{v}{V}{''}{'}}
 {\Eatprods{}{}{}{t~\vector{(x_r := v_r : T_r)} :\; T}{u_1~\vector{u_n}}{v}{V}{''}{''}}
 %{\Eatprods{}{}{}{t~(x_1 := v_1 : T_1)~\ldots~(x_r := v_r : T_r)} :\; T}{u_1~\ldots~u_n}{v}{V}{''}{''}}
\]\smallskip

\begin{proof}[Correctness of $(\eatT\mathrm{-prod})$]
Let $\xPS{}{}$ be the well formed pair taken as input by the rule and passed
to the second premise, that returns the well formed pair $\xPS{'}{'}$. Similarly, let $\xPS{''}{''}$ be the well formed pair given in output by the second premise and by the whole rule. By induction hypotheses
$\refines{\xPS{'}{'}}{\xPS{}{}}$ and
$\refines{\xPS{''}{''}}{\xPS{'}{'}}$ and thus
$\refines{\xPS{''}{''}}{\xPS{}{}}$ as required.
By the rule pre-condition, $T$ is well typed in $\xPSG{}{}{}$ and so are
$U_1$ and $T_1$ obtained by reduction. Thus the premises of the second rule
are all satisfied and, by induction hypothesis, $\xPSG{'}{'}{} \vdash u'_1 : U_1 \lland u'_1 \preccurlyeq u_1$. By rules $\mathcal{K}\mathrm{-appl-rec}$ and $\mathcal{K}\mathrm{-appl-base}$ applied to the rule pre-condition
$t~v_1~\ldots~v_r : T$ we get $t~v_1~\ldots~v_r~u'_1 : T_1[x/u'_1]$. Since
all preconditions for the third premise are satisfied, by induction hypothesis
we know $\xPSG{''}{''}{} \vdash v : V$ and $v \preccurlyeq t~v_1~\ldots~v_r~u'_1~u_2~\ldots~u_n$. By admissibility of $\preccurlyeq$ we conclude also
$v \preccurlyeq t~v_1~\ldots~v_r~u_1~\ldots~u_n$. Since all post-conditions
have been proved, the rule is correct.
\end{proof}

\[
\infrule[(\eatT\mathrm{-flexible})]
 {\Whd{}{}{}{T}{?_j} \quad \mbox{or} \quad \Whd{}{}{}{T}{?_j~w_1~\dots w_l}\\
  \Refinex{}{}{}{u_1}{u_1'}{U_1}{'}{'}\\
  %\Extendx{'}{;x_1:T_1;\ldots;x_r:T_r;x:U_1}{?_l}{\Type}{;x_1:T_1;\ldots;x_r:T_r;x:U_1}{?_k}{?_l}{''} \\
  \xP{} \leadsto \xP{} \cup \{ \xG{};\vector{x_r:T_r};x:U_1 \vdash ?_k : \Type_\top\}\\
  %\Unifx{''}{'}{}{T}{\Pi x:U_1.?_k[x_1/v_1;\ldots;x_r/v_r;x/x]}{'''}{''} \\
  \Unifx{}{''}{'}{}{T}{\Pi x:U_1.?_k[\vector{x_r/v_r};x/x]}{'''}{''} \\
  %\Eatprods{'''}{''}{}{t~(x_1 := v_1 : T_1)~\ldots~(x_r := v_r : T_r)~(x:=u_1':U_1) :\; ?_k[x_1/v_1;\ldots;x_r/v_r;x/u_1']\\\quad\quad}{u_2~\ldots~u_n}{v}{V}{''''}{'''}}
  \Eatprods{'''}{''}{}{t~\vector{(x_r := v_r : T_r)}~(x:=u_1':U_1) :\; ?_k[\vector{x_r/v_r};x/u_1']}{\vector{u_n}}{v}{V}{''''}{'''}}
 {
%\Eatprods{}{}{}{t~(x_1 := v_1 : T_1)~\ldots~(x_r := v_r : T_r) :\; T}{u_1~\ldots~u_n}{v}{V}{''''}{'''}
\Eatprods{}{}{}{t~\vector{(x_r := v_r : T_r)} :\; T}{u_1~\vector{u_n}}{v}{V}{''''}{'''}
}
\]\smallskip

\begin{proof}[Correctness of $(\eatT\mathrm{-flexible})$]
The proof is similar to the one for the $\eatT\mathrm{-prod}$ rule.
We only list the major differences here. The fact $\xPSG{'}{'}{} \vdash u'_1 : U_1$ is now obtained by induction hypothesis on the second premise.
The role of $T_1$ is now played by $?_k[\vector{x_r/v_r};x/x]$.
The induction hypothesis on the third premise yields
$\xPSG{'''}{''}{} \vdash T \downarrow \Pi x:U_1.?_k[\vector{x_r/v_r};x/x]$
that was previously given directly by the rule pre-conditions (up to reduction
of $T$). The rest of the proof follows without any changes. The only
remaining check to be performed is the well-formedness of
$?_k[\vector{x_r/v_r};x/x]$ that follows from rule $\mathcal{K}\mathrm{-meta}$
using the rule pre-condition $\xPSG{}{}{} \vdash v_i : T_i \quad i \in \{1 \ldots r\}$.
\end{proof}

Another reason for the complexity of the \eatT algorithm is the need to 
infer a dependent type for the function $f$ when its type is flexible (a
metavariable). We now show an example of this scenario and an execution trace
for the algorithm.

\begin{example}[Inference of a maximally dependent type]
Consider the following input, where $c_1, c_2, c_3, P_1,P_2$ are such that
$\vdash c_1: \nat{}$ and $\vdash c_2: P_1(c_1)$ and $c_3: P_2(c_1,c_2)$:
\[ \{ \vdash ?_F : \Type \}, \emptyset, \emptyset
      \vdash \lambda f : ?_F. f~c_1~c_2~c_3\]
\noindent 
The rule $(\eatT\mathrm{-flexible})$ matches the input and since
the argument $c_1$ has type $\nat{}$, \xP{} is extended as follows:
\[ \begin{array}{lrl}
    \xP{} = \{ &  & \vdash ?_F : \Type; \\
               & x : \nat{} & \vdash ?_S : ?_T; \\
               & x : \nat{} & \vdash ?_T : \Type \}
   \end{array}\]
\noindent 
Then $?_F$ gets unified with $\Pi x : \nat{}.?_S$ obtaining
the following substitution:
\[ \xS{} = \{ \vdash ?_F := \Pi x: \nat{}. ?_S : \Type \} \]
\noindent 
The new type for the head of the application, morally $(f~c_1)$,
represented as $f~(x := c_1)$,
 is $?_S[c_1/x]$.
In the following call to $(\eatT\mathrm{-flexible})$, the argument
$c_2$ has type $P_1(c_1)$. $\xP{}$ is thus extended as follows:
\[ \begin{array}{lrl}
     \xP{} = \{ & x: \nat{}; y : P_1(c_1) & \vdash ?_U : ?_V; \\
              & x: \nat{}; y : P_1(c_1) & \vdash ?_V : \Type; \\
              & x : \nat{} & \vdash ?_S : ?_T; \\
              & x : \nat{} & \vdash ?_T : \Type \}
   \end{array}\]
\noindent 
Then $?_S[x/c_1]$ is unified with $\Pi y : P_1(c_1).?_U[x/c_1]$ obtaining
\[ \begin{array}{lrl}
      \xS{} = \{ & & \vdash ?_F := \Pi x: \nat{}. ?_S : \Type; \\
                 & x : \nat{} & \vdash ?_S := \Pi y : P_1(x).?_U : ?_T; \\
                 & x : \nat{} & \vdash ?_T := \Type : \Type \}
   \end{array}\]
\noindent
The new type for the head of the application $(f~c_1~c_2)$ is $?_U[c_1/x;c_2/y]$.
In the following call to $(\eatT\mathrm{-flexible})$, the argument
$c_3$ has type $P_2(c_1,c_2)$. $\xP{}$ is thus extended as follows:
\[ \begin{array}{lrl}
     \xP{} = \{ & x: \nat{}; y : P_1(c_1); z: P_2(c_1,c_2) & \vdash ?_Z : ?_W; \\
                & x: \nat{}; y : P_1(c_1); z: P_2(c_1,c_2) & \vdash ?_W : \Type; \\
                & x: \nat{}; y : P_1(c_1) & \vdash ?_U : ?_V; \\
                & x: \nat{}; y : P_1(c_1) & \vdash ?_V : \Type \} \\
   \end{array}\]
\noindent 
Then $?_U[x/c_1;y/c_2]$ is unified with $\Pi z : P_2(c_1,c_2).?_Z[x/c_1;y/c_2]$ obtaining
\[ \begin{array}{lrl}
    \xS{} = \{ & & \vdash ?_F := \Pi x: \nat{}. ?_S : \Type; \\
               & x : \nat{} & \vdash ?_S := \Pi y : P_1(x).?_U : ?_T; \\
               & x : \nat{} & \vdash ?_T := \Type : \Type; \\
               & x : \nat{}; y : P_1(c_1) & \vdash ?_U := \Pi z : P_2(x,y).?_Z : ?_V; \\
               & x : \nat{}; y : P_1(c_1) & \vdash ?_V := \Type : \Type \}
   \end{array}\]

\noindent The final instantiation for $?_F$ is thus the maximally dependent type
\[ \xS{}(?_F) = \Pi x:\nat{}.\Pi y : P_1(x). \Pi z : P_2(x,y). ?_Z : \Type \]
where
\[ \begin{array}{lrl}
     \xP{} = \{ & x: \nat{}; y : P_1(c_1); z: P_2(c_1,c_2) & \vdash ?_Z : ?_W; \\
                & x: \nat{}; y : P_1(c_1); z: P_2(c_1,c_2) & \vdash
?_W : \Type \}\rlap{\hbox to 93 pt{\hfill\qEd}} \\
   \end{array}\]
\end{example}

We conclude now the description of the refinement algorithm in type inference
mode. The final missing rule is the most complicated one and deals
with pattern matching. It is reported in Figure~\ref{fig:match}.

\begin{figure}
\[
 \infrule[(\rec{}\mathrm{-match})]
  {(I_l : \Pi \vector{x_l : F_{l}}.\Pi \vector{y_r : G_r}. s) \in \env \\
 (k_j: \Pi \vector{x_l:F_l}.\Pi \vector{y^j_{n_j} :
 T^j_{n_j}}.I_l~\vector{x_l}~\vector{M^j_r}) \in \env \quad \hfill 
 j \in \{1 \ldots n\}\\\\

  \Extendxs{_i}{}{?u_i}{F_i[\vector{x_{i-1}/?u_{i-1}}]} \quad \hfill i \in \{1\ldots l\}\\
  \Extendxs{_i}{}{?v_i}{G_i[\vector{x_l/?u_l};\vector{y_{i-1}/?v_{i-1}}]} \quad \hfill i \in \{1\ldots r\}\\
  \RefinexE{_{l+r+1}}{}{}{t}{I_l~\vector{?u_l}~\vector{?v_r}}{t'}{'}{'}\\\\

  G'_i = G_i[\vector{x_l/?u_l}] \quad \hfill i \in \{1 \ldots r\}\\
  {T'}^j_i = T^j_i[\vector{x_l/?u_l}] \quad \hfill j \in \{1 \ldots n\}, i \in \{1 \ldots n_j\}\\
  {M'}^j_i = M^j_i[\vector{x_l/?u_l}] \quad \hfill j \in \{1 \ldots n\}, i \in \{1 \ldots r\}\\\\

  \Extendxs{'}{'}{?_{1}}{\Type_\top}{''} \\
  \RefinexE{''}{'}{}{T}{\Pi \vector{y_r : G'_r}.
  %\Pi x:I_l~?_{u_1}~\ldots~?_{u_l}~x_{l+1}~ \ldots x_{l+r}. ?_1[]}{T'}{'_1}{'_1}\\
  \Pi x:I_l~\vector{?u_l}~\vector{y_r}. ?_1[]}{T'}{'_1}{'_1}\\
 (s,\xS(?_1)) \in \elim(\pts)\\\\

 \Unifx{}{}{}{;\vector{y^j_{n_j-1} : P^j_{n_j-1}}}{P^j_{n_j}}{{T'}^j_{n_j}}{}{}     \hfill j \in \{1\ldots n\}\\

  \RefinexE{'_i}{'_i}{; \vector{y^j_{n_j} : P^j_{n_j}}}{%\lambda \vector{y^j_{n_j} : P^j_{n_j}}.
t_j}{%\Pi \vector{y^j_{n_j}:{T'}^j_{n_j}}.
T'~\vector{{M'}^j_r}~(k_j~\vector{?u_l}~
   \vector{y^j_{n_j}})}%\\\quad \quad}
  {%\lambda \vector{y^j_{n_j} : {P'}^j_{n_j}}.
t'_j}{'_{i+1}}{'_{i+1}}\quad \hfill j \in \{1\ldots n\}
 }
 {\begin{array}{l}
  \Refinex{_1}{}{}{
  \left(\begin{array}{l}
    \match{t}{I_l}{T} \\
     \;\;[ k_1~(\vector{y^1_{n_1}:P^1_{n_1}})\Rightarrow t_1~|~\ldots~|~
           k_n~(\vector{y^n_{n_n}:P^n_{n_n}})\Rightarrow t_n ]
  \end{array}\!\!\right)\!}
  {\\\quad\quad\left(\begin{array}{l}
    \match{t'}{I_l}{T'} \\
     \;\;[ k_1~(\vector{y^1_{n_1}:{P'}^1_{n_1}})\Rightarrow t'_1~|~\ldots~|~
           k_n~(\vector{y^n_{n_n}:{P'}^n_{n_n}})\Rightarrow t'_n ]
  \end{array}\!\!\right)\!}
   {\!\!T'~\vector{?v_r}~t'}{'_{n+1}}{'_{n+1}}
 \end{array}}
\]
\caption{\label{fig:match}Rule for pattern matching ($\rec{}\mathrm{-match}$)}
\end{figure}
%\\

%TASSI: usare nomi diversi per F e x tra l+1 e l+r in modo da poter usare il 
%vettore $\vector{x_l}\vector{y_r}$
%ok

The rule has been slightly simplified: in the actual implementation of
Matita the test $(s,\Phi(?_1)) \in \elim(\pts)$ is relaxed to accept
elimination of inhabitants of non informative data types in all cases
under the restriction that the data type must be small. Intuitively, smallness
corresponds to the idea that the inhabitant of the data type would be non
informative even if declared in $\Type$. Typical examples are empty types and
the Leibniz equality type. A precise definition of smallness together with
the corresponding rules for pattern matching can be found in~\cite{ck-sadhana}.

Note that the return type $T'$ is usually an anonymous function, beginning
with lambda abstractions. Thus the type inferred for the pattern match
construct is a $\beta$-redex. In fact the actual code of Matita
post-processes that type firing $(r+1)$ $\beta$-redexes. 

%WILMER fissa l - r + 1 se rinomini le variabili
%ok?

%CSC: dire che poi si possono usare implicit (non fatto nel codice).\\

\bugnelcodice{Nota: qui sto cambiando completamente il codice, ma quello ottenuto mi sembra meglio perch\`e usa l'altra direzione ed \`e pi\`u leggibile.}

\bugnelcodice{Nota: manca il caso outtype = implicit dove nel codice espandiamo con dei lambda per far inferire pi\`u tipi dipendenti (perch\`e?). TASSI: perchè coi lambda metti subito i destri e il matchato nel contesto della meta e (forse solo una volta) era piu potente... ora generiamo sempre i lambda nel caso meta applicata contro termine nell'unificazione}

%CSC: ci comportiamo diversamente dal caso applicazione dove evitiamo l'unificazione quando la riduzione \`e sufficiente!\\

\begin{thm}[Termination]
The \rec{} % and \recE{} 
algorithm defined by the set of rules presented in this 
section including \recE{}, \coer{}, \forcetotypelabel{} and \eatprodslabel{}
% and \unimany{}%are 
is terminating.
\end{thm}
\begin{proof}
The proof is by structural induction of the syntax of terms. 
The rules $(\coer{}\mathrm{-ok})$, $(\rec{}\mathrm{-variable})$, 
$(\rec{}\mathrm{-sort})$ and $(\rec{}\mathrm{-constant})$ are base cases. The
first one clearly terminates if the unification algorithm $\UNI{}$ terminates,
while the others terminate since \xG{} and \env{} are finite and 
the test $(s_1, s_2) \in \pts$ is also terminating.

Now that we proved that all the rules for \coer{}, amounting to 
only one for the mono directional refiner,
terminate, we can consider the $(\recE{}\mathrm{-default})$ and
$(\forcetotypelabel{}\mathrm{-ok})$ as aliases
for \rec{}, as if we were inlining their code.

By induction hypothesis \rec{} (and \recE{} being now an alias) 
terminates when called on smaller terms. 
The rule $(\rec{}\mathrm{-meta})$ terminates because
\xS{} and \xP{} are finite, thus lookups are terminating, and 
calls to \recE{} are done on smaller terms, so the induction hypothesis applies.
The rule $(\rec{}\mathrm{-letin})$ calls \forcetotypelabel{}, \recE{} and
\rec{} on smaller terms, thus terminates by induction hypothesis.
The same holds for $(\rec{}\mathrm{-lambda})$ and $(\rec{}\mathrm{-product})$.
To prove that $(\rec{}\mathrm{-appl})$ terminates we use the induction 
hypothesis on the first premise and we are left to prove that 
\eatprodslabel{} terminates as well. Note that \eatprodslabel{} calls
\recE{} and \rec{} on proper subterms of the n-ary application, thus
the induction hypothesis applies and will be used in the next paragraph.

We show \eatT{} terminates by induction on the list
of arguments (i.e. the list of terms after~$\eatsep$) assuming that
the input term $T$ is a well typed type. Thanks to the correctness property 
of \rec{}, $(\rec{}\mathrm{-appl})$ always passes to \eatT{} a well typed type.
The rule $(\eatT\mathrm{-empty})$ clearly terminates.
The recursive call in the rule $(\eatT\mathrm{-prod})$
is on a shorter list of arguments, thus is terminating, and
the call to \recE{} terminates by induction hypothesis. 
The term $T[x/u_1']$ is a
well typed type thanks to the subject reduction property of CIC, and the fact
that the variable $x$ and the term $u_1'$ have the same type (postcondition 
of \recE{}, called with expected type $U_1$).
The call to $\whdlabel{}$ is terminating because $T$ is well typed and 
CIC reduction rules, on well typed  terms, form a terminating rewriting system.

The rule $(\eatT\mathrm{-flexible})$ terminates because of the same
arguments. 
The only non obvious step is that 
$\Pi x: U_1. ?_k[\vector{x_r/v_r};x/x]$ is a well typed type.
The metavariable $?_k$ is declared in \xP{} of type $\Type_\top$,
thus cannot be instantiated with a term. Moreover,
since CIC is a full PTS,
the product $\Pi x: U_1. ?_k$ is well typed of sort $\Type_\top$.

The recursive calls in the last rule $(\rec{}\mathrm{-match})$ are
always on smaller terms. We are left to prove that the expected type
$(T'~\vector{{M'}^j_r}~(k_j~\vector{?u_l}~\vector{y^j_{n_j}}))$ passed to
the recursive call made on the last line is indeed a well typed type. 
The term $T'$
is obtained using \recE{}, and we thus know it is a function
of type $\Pi \vector{y_r : G'_r}.\Pi x:I_l~\vector{?u_l}~\vector{y_r}.?_1$.
The arguments $\vector{{M'}^j_r}$ are as many as expected and
have the right types according to the environment \env{} (first two lines)
and thanks to the fact that the substitutions 
${M'}^j_i = M^j_i[\vector{x_l/?u_l}]$
preserves their types. 
Finally, the term $(k_j~\vector{?u_l}~\vector{y^j_{n_j}})$
has type $I_l~\vector{?u_l}~\vector{M_r^j}$ that is the expected one.
Thus $(T'~\vector{{M'}^j_r}~(k_j~\vector{?u_l}~\vector{y^j_{n_j}}))$ 
has type $?_1$ that is a well typed type according to \xP{}.
\end{proof}
\subsection{Implementation remarks}

The choice of keeping $\xP{}$ and $\xS{}$ separate is important and motivated
by the fact that their size is usually very different.  While the number of
proof problems in $\xP{}$ is usually small, the substitution $\xS{}$ may
record, step by step, the whole proof input by the user and can grow to an
arbitrary size.  Each metavariable must be either declared in $\xP{}$ or
assigned in $\xS{}$, thus to know if a metavariable belongs to $\xS{}$ it is
enough to test if it does not belong to $\xP{}$.  Knowing if a metavariable is
instantiated is a very common operation, needed for example by weak head
normalization, and it must thus be possible to implement it efficiently. 

Another important design choice is to design the kernel of the system 
so that it handles
metavariables~\cite{ck-sadhana}. This enables to reuse a number of
functionalities implemented in the kernel also during the refinement process,
like an efficient reduction machinery.  Also note that the extensions made to
the type checker described in Definition~\ref{typj} become dead code when the
type checker is called on ground terms, and thus do not increase the size of
the trusted code of the system.

Last, it is worth pointing out that the algorithm is mostly independent from
the representation chosen for bound variables. Matita is entirely based on De
Bruijn indexes, but the tedious \emph{lift} function is mostly hidden inside
reduction and only directly called twice
in the actual implementation of this algorithm. In particular it is necessary
only to deal with the types of variables that must be pulled from the context.
This potentially moves the type under all the context entries following 
the variable declaration or definition, thus the type must be lifted 
accordingly.

\bugnelcodice{REGOLA PROD
Nota: qui nel codice non si usa la check\_type non si sa perch\`e. Immagino
perch\`e non torna la sorta, ma io l'ho modificata nel papero per questo.}

\bugnelcodice{
REAGOLA EAT-FLEX
Nota: nel codice della whd c'e' una linea commentata con (* next line grants that ty\_args is a type *): perch\'e non \`e un invariante rispettato?}

%REGOLA EAT-FLEX
%CSC: nel codice la sostituzione fa avoid\_beta\_redexes
%CSC: DIRE CHE NON DIAMO REGOLE PER I CASI IN CUI $T$ RIDUCE A UN MATCH O UN LET-REC APPLICATO A UNA METAVARIABILE (PERDITA DI COMPLETEZZA).

\subsubsection{Rules for objects}

Objects are declarations and definitions of constants, inductive types and
recursive and co-recursive functions that inhabit the environment $\env$.
Exactly like terms, the user writes objects down using the external syntax
and the objects need to be refined before passing them to the kernel for
the final check before the insertion in the environment.

\begin{definition}[Type checking for objects ($\env \vdash \mbox{WF}$)]
The type checking algorithm for CIC objects
takes as input a proof problem $\xP{}$ and a substitution
$\xS$, all assumed to be well formed, and an object
$o$.
It is denoted by:\\
\[\env \cup (\xP{},\xS{}, o) \vdash \mbox{WF}\]
and states that $o$ is well typed.
\end{definition}

This algorithm is part of the kernel, and described in Appendix~\ref{sec:app}.
It is the basis for the construction of the corresponding refinement
algorithm for objects, that is specified as follows.

\begin{specification}[Refiner for objects $(\mathcal{R})$]
A refiner algorithm $\mathcal{R}$ for CIC objects
takes as input a proof problem $\xP{}$ and a substitution
$\xS$, all assumed to be well formed, and an object
$o$. It fails or returns an object $o'$, a
proof problem $\xP{'}$ and substitution $\xS{'}$.
It is denoted by:\\
\[ \xPSp{}{} \vdash o \recO{} o' ~ (\xPS{'}{'})\]
Precondition:
\[ \WF{\xPS{}{}} \]
Postcondition (parametric in $\preccurlyeq$):
\[\WF{\xPS{'}{'}} \lland \refines{\xP{'},\xS{'}}{\xP{},\xS{}} \lland 
\env \cup (\xP{'},\xS{'}, o') \vdash \mbox{WF}
\lland o' \preccurlyeq o \]
\end{specification}

Note that an object can be a block of mutually recursive definitions or
declarations, each one characterized by a different type.
Thus the \recO{} rule does not return a single type, but a new object
together with a new metavariable environment \xP{'} and substitution \xS{'}.
When $\xP{'}$ is not empty, all the metavariables in $\xP{'}$ correspond to
proof obligations to be proved to complete the definition of the object,
necessary to convince the system to
accept the object definition. This is especially useful
for instance in the formalization of category theory where definitions
of concrete categories are made from definitions of terms (objects and morphisms)
together with proofs that the categorical axioms hold. In such definitions,
objects and functors are immediately fully specified, while the proof parts
are turned into proof obligations. Spiwack's Ph.D. thesis~\cite{Spiwack} discusses
this issue at length as a motivation for a complete re-design of the
data type for proofs in Coq. The new data type is essentially the one used
in this paper and it will be adopted in some forthcoming version of Coq.
The old data type, instead, did not take the Curry-Howard isomorphism seriously
in the sense that partial proofs were not represented by partial proof terms
and the refinement of an object could not open proof obligations. This
problem was already partially addressed by Sozeau~\cite{russell} where he
added a new system layer around the refiner to achieve the behavior that
our refiner already provides.

\[
\infrule[(\recO{}\mathrm{-axiom})]
 { \vdash T \forcetotypelabel T' : s }
 { \XPS{_1}{_1} \vdash \axiom{c}{T} \recO{} \axiom{c}{T'} }
\]\bigskip
\[
\infrule[(\recO{}\mathrm{-definition})]
 { \vdash T \forcetotypelabel T' : s \qquad
   \vdash t : T' \recE{} t' }
 { \XPS{_1}{_1} \vdash \definit{c}{T}{t} \recO{} \definit{c}{T'}{t'} }
\]\bigskip
\[
\hspace{-2em}
\infrule[(\recO{}\mathrm{-inductive})]
 { 
\left.
\begin{array}{l}
   \vdash \Prod \vector{x_l:L_l}.A_i \forcetotypelabel V_i' 
     \qquad\qquad\qquad\qquad\qquad\qquad\qquad\;\\
   \vdash V_i' \whdlabel{} \Prod \vector{x_l:L_l'}.A_i' \qquad%\\
   \vdash A_i' \whdlabel{} \Prod \vector{y_{r_i} : R_{r_i}'}.s_i\\
\end{array}
\right\} \begin{array}{l} i \in \{1\ldots n\} \end{array}\\
\left.
\hspace{-0.46em}\begin{array}{l}
   \vector{x_l:L_l'}; \vector{I_n : V_n'} 
     \vdash T_{i,k} \forcetotypelabel T_{i,k}' : s_{i,k} \\
   \vector{x_l:L_l'}; \vector{I_n : V_n'} 
     \vdash T_{i,k}' \whdlabel %T_{i,k}''\\
   \Prod \vector{z_{p_{i,k}} : V_{p_{i,k}}}.V_{i,k}\\
   \xP{} \leadsto \xP{} \cup \{ \vector{x_l : L_l'};\vector{y_{j-1} : R_{j-1}'}\vdash ?_j : R_j' \} \quad j \in \{1 \ldots r_i\}\\
   \vector{x_l:L_l'}; \vector{I_n : A_n'} ;
     \vector{z_{p_{i,k}} : V_{p_{i,k}}} \vdash
     V_{i,k} \UNI I_i~\vector{x_l}~\vector{?_{r_{i}}} \uni{_=}
\end{array}
\right\}
\begin{array}{l} i \in \{1\ldots n\} \\ k \in \{1\ldots m_n\} \end{array}
 }{
\begin{array}{l}
\XPS{_1}{_1} \vdash 
\left(
       \inductive{\Prod \vector{x_l:L_l}.~}
                {I_1}{A_1}{k_{1,1} : T_{1,2}}{k_{1,m_1} : T_{1,m_1}}
                {I_n}{A_n}{k_{n,1} : T_{n,1}}{k_{n,m_n} : T_{n,m_n}}
\right)
\recO{}\\
\quad
\left(
       \inductive{\Prod \vector{x_l:L'_l}.~}
                {I_1}{A_1'}{k_{1,1} : T_{1,2}'}{k_{1,m_1} : T_{1,m_1}'}
                {I_n}{A_n'}{k_{n,1} : T_{n,1}'}{k_{n,m_n} : T_{n,m_n}'}
\right)
\end{array}
}
\]\smallskip

\noindent The loop from $1$ to $n$, ranges over all mutually inductive types
of the block using the index $i$ (for inductive). The other loop, 
from $1$ to $m_n$, ranges over the $m_n$ constructors of the $n$-th 
inductive type using index $k$ (for constructor). $T_{i,k}$ is the
type of the $k$-th constructor of the $i$-th inductive. Every inductive type
in the block has the same number of homogeneous parameters $l$ of type
$L_\alpha$ for some $\alpha$ in $1 \dots l$, and $r_i$ extra arguments of type 
$R_\beta$ for some $\beta$ in $1 \ldots r_i$. As explained in 
Section~\ref{sec:pre}, homogeneous arguments are not abstracted explicitly 
in the types of the constructors, thus their context includes not
only the inductive types $\vector{I_n}$ but also the homogeneous
arguments $\vector{x_l}$. Note that in this rule we used $I_n$ to
mean the $n$-th inductive, in contrast with the rest of the paper where
the index $n$ means the number of homogeneous arguments, $l$ here.
The complete arity of the inductive 
types $\vector{V'}$ is a closed term. In fact that type is generated in an
empty context in the first premise. This makes the context in which 
$T_{i,k}$ is processed valid: there is no variable capture when 
$\vector{x_l}$ is put before $\vector{I_n}$. 
Moreover the successful refinement of
$\Prod \vector{x_l:L_l}.A_i$ in the empty context grants that 
the types of the homogeneous arguments do not depend on the
inductive types $\vector{I_n}$.
The last three premises just check
that the type of each constructor is actually a product targeting the
inductive type. 

Whilst being already quite involved, this rule is only partial.
It lacks the checks for positivity conditions, that are only implemented
by the kernel. Since the kernel of Matita is able to deal with
metavariables we can test for these conditions using directly the kernel
after the refinement process. Nevertheless, when the inductive type fed to 
the kernel is partial, the checks cannot be
precise: all non positive occurrences will be detected, but nothing will
prevent the user from instantiating a missing part with a term containing
a non positive occurrence. One could label metavariables in a such a way
that the unification algorithm refuses to instantiate them with a term
containing non positive occurrences of the inductive type, but our current 
implementation does not. Anyway, once the definition is completed by the 
user, another call to the kernel is  made, and all non positive occurrences 
are detected.

Moreover, the kernel also checks that the
sort $s_{i,k}$ of the type of every constructor $T_{i,k}$ is properly 
contained in the sort of the corresponding inductive $s_i$. Finally,
one should also check that any occurrence of $I_i$ 
in the types $T_{i,k}$ of the constructors is applied to $\vector{x_l}$.
This test is also omitted since it is performed by the kernel during
the test for positivity.

\[
\infrule[(\recO{}\mathrm{-letrec})]
 {
\left.
\begin{array}{l}
  \vdash \Pi \vector{x^i_{p_i} : T^i_{p_i}}.T^i_{p_i+1} \forcetotypelabel T'_i : s_i \\
  \vdash T'_i \whdlabel \Pi \vector{x^i_{p_i} : {T'}^i_{p_i}}. {T'}^i_{p_i+1} \\
  \vector{f_n : T_n}; \vector{x^i_{p_i}:{T'}^i_{p_i}}
  \vdash t_i : {T'}^i_{p_i+1} \recE{} t_i'
\end{array}
\right\} i \in \{1\ldots n\} \\
} {
  \begin{array}{l}
   \XPS{_1}{_1} \vdash 
   \left(
   \letr[\\]{f_1 (\vector{x^1_{p_1} : T^1_{p_1}})}
        {T^1_{p_1+1}}{t_1}
        {f_n (\vector{x^n_{p_n} \!:\! T^n_{p_n}})}
        {T^n_{p_n+1}}{t_n}
   \right)\recO{}\!\!\\
   \quad
    \left(
   \letr[\\]{f_1 (\vector{x^1_{p_1} \!:\! {T'}^1_{p_1}})}
        {{T'}^1_{p_1+1}}{t'_1}
        {f_n (\vector{x^n_{p_n} \!:\! {T'}^n_{p_n}})}
        {{T'}^n_{p_n+1}}{t'_n}
    \right)
     \h{,~ \xP{_{2n+1}},~ \xS{_{2n+1}}}
   \end{array}
  }
\]\smallskip

\noindent As for inductive types, this rule is only partial: it lacks
 the checks for guardedness conditions (termination or productivity
 tests), that are delegated to the kernel. We omit the rule for
 co-recursive functions, since it is identical to the one presented
 above.

%%%%%%%%%%%%%%%%%%%%%%%%%%%%%%%%%%%%%%%%%%%%%%%%%%%%%%%%%%%%%%%%%%%%%%%%%%%
%%%%%%%%%%%%%%%%%%%%%%%%%%%%%%%%%%%%%%%%%%%%%%%%%%%%%%%%%%%%%%%%%%%%%%%%%%%
%%%%%%%%%%%%%%%%%%%%%%%%%%%%%%%%%%%%%%%%%%%%%%%%%%%%%%%%%%%%%%%%%%%%%%%%%%%
\section{Bi-directional refinement}\label{sec:bi}

To obtain a bi-directional implementation of the refiner, we add new
rules to the \recE{} algorithm. These ad-hoc rules for particular
cases must take precedence over the generic
$(\recE{}\mathrm{-default})$ rule. The ad-hoc rules are responsible for
propagating information from the expected type towards the leaves of the
term.
% All these rules are redundant, in the sense that they do not change
%the relative completeness of the algorithm. Nevertheless, in practice
%they can improve unification and BLA BLA.

The new rule for lambda-abstraction is well known in the 
literature~\cite{piercelocaltype}
and it is also the only one implemented in Coq.
The rule for let-in statements is given to allow the system infer more
concise types. The one for application of constructors is completely novel
and it takes advantage of
additional knowledge on the constant parameters of an inductive type. It
is thus peculiar of the Calculus of (Co)Inductive Constructions. This is
also the rule that, according to our experience, mostly affects the behavior
of the refiner. It makes it possible to refine many more terms to be refined in frequently
occurring situations where, using a mono directional algorithm, 
more typing information had to be given by hand.

\begin{thm}[Correctness]
The new rules given in this section do not alter the correctness of the
$\recE{}$ algorithm w.r.t. its specification
for all admissible $\preccurlyeq$ relations that include the
identity for terms in the internal syntax. In particular, the algorithm is
correct when the identity relation is picked for $\preccurlyeq$. \qed
\end{thm}

\[
\infrule[(\recE{}\mathrm{-lambda})]
 {\Whd{}{}{}{E}{\Pi x:E_1.E_2} \\
  \Forcetotype{}{}{}{T}{T'}{s}{'}{'} \\
  \Unifx{_=}{'}{'}{}{T'}{E_1}{''}{''}\\
  \RefinexE{''}{''}{; x : T'}{t}{E_2}{t'}{'''}{'''}}
 {\RefinexE{}{}{}{\lambda x:T.t}{E}{\lambda x:T'.t'}{'''}{'''}}
\]\smallskip

\noindent Note that to type $t$ we push into the context the declared type
for $x$ and not its expected type $E_1$. This is to avoid displaying
a confusing error message in case $t$ is ill-typed, since the user
declared $x$ to have type $T$, and not $E_1$ (that in principle can be
arbitrarily different from $T$).

\bugnelcodice{mettiano nel contesto $x: T_1'$ invece di $x: E_1$? Vi \`e useless code poich\`e in ambo i casi mettiamo $T_1'$ quando il codice sembrava pensato per $E$.}

\[
\infrule[(\recE{}\mathrm{-letin})]
 {\Forcetotype{}{}{}{T}{T'}{s}{'}{'} \\
  \RefinexE{'}{'}{}{t}{T'}{t'}{''}{''}\\
  \RefinexE{''}{''}{; x := t' : T'}{u}{E[t'/x]}{u'}{'''}{'''}}
 {\RefinexE{}{}{}{\letin{x}{T}{t}{u}}{E}{\letin{x}{T'}{t'}{u'}}{'''}{'''}}
\]\\
Where we denote by $[t'/x]$ the operation of substituting
all occurrences of $t'$ with $x$. Note that this operation
behaves as an identity up to conversion (since $x$ holds the value $t'$).
Nevertheless, it enables the bi-directional type inference algorithm to propagate
smaller types towards the leaves and, according to our observation, it leads
to more readable inferred typed for sub-terms of $u'$.

\begin{thm}[Termination]
The \rec{} % and \recE{} 
algorithm defined by the set of rules presented above with the 
addition of
$(\recE{}\mathrm{-letin})$ and $(\recE{}\mathrm{-lambda})$
terminates.
\end{thm}
\begin{proof}
The two functions terminate because all recursive calls are on smaller
terms and because $\whdlabel{}$, $\forcetotypelabel{}$ and $\UNI{}$ terminate.
\end{proof}

The next rule deals with applications of constructors to arguments and it is
only triggered when the expected type is an inductive type. In that case the
application must be total. In CIC, the types of constructors of inductive types are constrained to have a particular shape. Up to reduction, their type must be of the form
$\Pi x_1:F_1 \ldots \Pi x_n:F_n.I~x_1~\ldots~x_l~t_{l+1}~\ldots~t_m$
where $l$ is the number of uniform parameters of the inductive type.
Therefore the application of a constructor to a list $u_1~\ldots~u_n$
of arguments has type $I~u_1~\ldots~u_l~v_{l+1}~\ldots~v_m$ for
some $v$s. Reversing the reasoning, once we know that the expected type for
the application of a constructor is $I~u_1~\ldots~u_l~v_{l+1}~\ldots~v_m$
we already know that the first $l$ arguments of the application must be
equal to $u_1~\ldots~u_l$ up to conversion. It is thus possible to propagate
them following the bi-directional spirit. This is achieved by the
following $(\recE{}\mathrm{-appl-}k)$ that calls a new function
denoted by $\eatt$ that consumes the first $l$ arguments unifying them with
the expected values. The remaining arguments are consumed as in the generic
case of applications.

\[
\infrule[(\recE{}\mathrm{-appl-}k)]
 %{\Whd{}{}{}{E}{I_l~v_1~\ldots~v_n~w_1~\ldots~w_l} \\
 {\Whd{}{}{}{E}{I_l~\vector{v_l}~\vector{w_n}} \\
  %\Unimanyx{}{}{}{t_1~\ldots~t_m}{v_1~\ldots~v_n}{t_1'~\ldots~t_n'}{u_1~\ldots~u_l}{'}{'}\\
  \Unimanyx{}{}{}{\vector{t_m}}{\vector{v_l}}{\vector{t_l'}}{\vector{u_o}}{'}{'}\\
  (k : T) \in \env{}\\
  %\Whd{'}{'}{}{T}{\Pi x_1:S_1.\ldots\Pi x_n:S_n.T'}\\
  \Whd{'}{'}{}{T}{\vector{\Pi x_l:S_l}.T'}\\
  %\Eatprods{'}{'}{}{k~t'_1~\ldots~t'_n:\; T'[x_1/t_1';\ldots;x_n/t_n']}{u_1~ \ldots~u_l}{r}{R}{''}{''}\\
  \Eatprods{'}{'}{}{k~\vector{t'_l}:\; T'[\vector{x_l/t_l'}]}{\vector{u_o}}{r}{R}{''}{''}\\
  %\Unifx{''}{''}{}{R}{I_l~v_1~\ldots~v_n~w_1~\ldots~w_l}{'''}{'''}}
  \Unifx{_=}{''}{''}{}{R}{I_l~\vector{v_l}~\vector{w_n}}{'''}{'''}}
 {\RefinexE{}{}{}{k~\vector{t_m}}{E}{r}{'''}{'''}}
\]\bigskip
Note that if $E$ does not reduce to an applied inductive type,
the implemented algorithm falls back to the standard rule for
application.

The rule presented only propagates information related to uniform parameters.
Uniform parameters must be used consistently in every occurrence of the
inductive type in the type of its constructors and not only in the
occurrence at the end of the product spine (i.e. in the return type of the
constructors). The variant of CIC implemented
in Coq also considers non uniform parameters. Non uniform parameters
must be used consistently only in the return type of the constructors and
not in the premises. We do not consider non uniform parameters in this paper,
but we remark that the $(\recE{}\mathrm{-appl-}k)$ rule is also valid when
the first $l$ parameters are non uniform.

\begin{specification}[Eat arguments ($\eatt$)]
The $\unimany$ algorithm takes a list of arguments for an application and a list
of terms, and it verifies that an initial prefix of the arguments is equal to
the given terms, up to unification.
It takes as input a proof problem $\xP{}$, a substitution
$\xS$ and a context $\xG{}$, all assumed to be well formed, the list of
arguments and the list of terms.
It fails or it returns the list of arguments split into the consumed
ones and the ones yet to be considered.
It is denoted by:
\[\unimanyx{}{}{}{t_1~\ldots~t_m}{v_1~\ldots~v_n}{t_1'~\ldots~t_n'}{u_1~\ldots~u_l}{'}{'}\]
Precondition:
 \[ \WF{\xPSG{}{}{}} \lland
    \xPSG{}{}{} \vdash v_i : T_i \quad i \in \{1 \ldots n\}
 \]
Postcondition (parametric in $\preccurlyeq$):
 \[ \WF{\xPS{'}{'}} \lland \refines{\xP{'},\xS{'}}{\xP,\xS} \lland
    \xPSG{'}{'}{} \vdash t'_i \downarrow v_i \quad i \in \{1 \ldots n\} \lland
\]
\[
    t'_1~\ldots~t'_n~u_1~\ldots~u_l \preccurlyeq t_1~\ldots~t_m
\]
\end{specification}

\[
\infrule[(\eatt\mathrm{-empty})]
 {}
 %{\Unimanyx{}{}{}{u_1~\ldots~u_l}{}{}{u_1~\ldots~u_l}{}{}}
 {\Unimanyx{}{}{}{\vector{u_l}}{}{}{\vector{u_l}}{}{}}
\]

\[
\infrule[(\eatt\mathrm{-base})]
 {\Refinex{}{}{}{t_1}{t_1'}{T_1}{'}{'}\\
  \Unifx{_=}{'}{'}{}{t_1'}{v_1}{''}{''}\\
  \Unimanyx{''}{''}{}{\vector{t_m}}{\vector{v_n}}{\vector{t_n'}}{\vector{u_l}}{'''}{'''}}
 {\Unimanyx{}{}{}{t_1~\vector{t_m}}{v_1~\vector{v_n}}{t_1'~\vector{t_n'}}{\vector{u_l}}{'''}{'''}}
\]

\begin{thm}[Termination]
The \rec{} % and \recE{} 
algorithm defined by the set of rules presented above with the 
addition of
$(\recE{}\mathrm{-appl}-k)$,
$(\eatt\mathrm{-empty})$
and
$(\eatt\mathrm{-base})$
terminates.
\end{thm}
\begin{proof}
The rule $(\recE{}\mathrm{-appl}-k)$ terminates because $\UNI{}$
and $\whdlabel{}$ terminate and $\eatT{}$ is called on smaller terms. Moreover
the term $T'[\vector{x_l/t_l'}]$ is a well typed type because $\eatt{}$
grants that $\vector{t_l'}$ are convertible with $\vector{v_l}$ and thus 
have the same types. Also notes that all the calls to \rec{} made by $\eatt{}$
are on sub-terms of the input of $(\recE{}\mathrm{-appl}-k)$.

We thus show that $\eatt$ terminates by induction on the second list of 
arguments (the one between $\UNI{}$ and $\stackrel{\eatt}{~\leadsto~}$).
Rule $(\eatt\mathrm{-empty})$ is the base case and clearly terminates.
Rule $(\eatt\mathrm{-base})$ terminates because \rec{} and $\UNI{}$ terminate
and because the recursive call terminates by induction hypothesis.
\end{proof}

\subsection{Remarks}
We present here a simple but frequently occurring case that explains why
the bi-directional rule for application of constructors enables to refine
many more terms w.r.t. the mono-directional algorithm. A more complicated
example was already discussed in the introduction and deals with dependent
data types to represent the syntax of languages with binders.

Consider the inductive type used to define the
existential quantification.
\[
\inductiveX{\Prod T : \Type.\Prod P : T \to \Prop.~}{~\mbox{Ex}}{\Prop}
{\\\qquad\mbox{Ex\_intro} : \Pi x : T.P~x \to \mbox{Ex}~T~P}
\]
\noindent Note that $T$ and $P$ are homogeneous arguments. 

\begin{example}[Use of $(\recE{}\mathrm{-appl-}k)$]
Consider the conjecture $\exists x:\nat{}. x > 0$,
encoded in CIC as 
   \[\mbox{Ex}~\nat{}~(\lambda x:\nat{}.x > 0) \]
\noindent Given a context $\xG{}$ containing the assumption $p$ stating
that $2 > 0$, one may want to use the following proof term to prove the
conjecture
   \[ t = \mbox{Ex\_intro}~?_T~?_P~?_x~p \]
\noindent A mono directional refiner encounters a hard unification problem
involving the type of $p$, $2 > 0$, and its expected type.
   \[ \{ \vdash~ ?_T ~:\; \Type;~ \vdash~ ?_P ~:\; ?_T \to \Prop;~ 
         \vdash~ ?_x ~:\; ?_T \}, 
      \emptyset, \xG{} \vdash 2 > 0 ~\UNI~ ?_P~?_x \]
\noindent Clearly, the desired solution is to instantiate $?_x$ with
$2$ and $?_P$ with $(\lambda x:\nat{}.x > 0)$ obtaining a proof
term of type $\exists x : \nat{}. x > 0$. 
Unfortunately, this is not the only
possible solution. An undesired solution, but as reasonable as the correct one, is 
   \[\xS{} = \{ ?_T := \nat{};~ ?_x := 0;~ ?_P := \lambda x:\nat{}.2 > x \}\]
\noindent under which the resulting proof term $\xS{(t)}$ has type 
$\exists x:\nat{}. 2 > x$, that is not the expected one. 
Why one should prefer the former to the
latter is also unclear from a computer perspective.
Thanks to the polymorphism of CIC, another undesired and less expected
solution is also possible:
   \[\xS{'} = \{ ?_T := \Prop; ?_x := 2 > 0; ?_P := \lambda x.x \}\]
\noindent The proof term $\xS{'(t)}$ would then be of type 
$\exists x:\Prop. x$, again different from the desired one.

Using the expected type, $?_T$ and $?_P$ are easily inferred 
looking at the homogeneous argument of the expected type. 
   \[\xS{''} = \{ ?_T := \nat{};~ ?_P := \lambda x:\nat{}.x > 0 \}\]
\noindent Then inferring $?_x$ is easy.
Applying $\xS{''}$ to the right hand side of the unification problem we obtain:
   \[ 2 > 0 ~\UNI~ (\lambda x:\nat{}.x > 0)~?_x \]
\noindent Then it is sufficient to reduce the right hand side and then perform
a simple, first order, unification algorithm to obtain the desired
instantiation for $?_x$.
\qed{}
\end{example}

The rule $(\recE{}\mathrm{-appl-}k)$ is only fired when the term to be refined
is syntactically the application of a constructor. Because of conversion,
the term under analysis could be reducible to an application of a constructor.
However, we cannot reduce the term first to try to match the rule.
The first motivation is that terms in the external syntax may
contain placeholders (see Section~\ref{sec:raw}) and may not be well typed. Duplication of placeholders
and substitution into them is not admitted. Moreover,
reducing an ill typed term may lead to divergence. Secondly, reduction of
proof terms correspond to cut elimination that is known to yield proofs terms
of arbitrary size.

\bugnelcodice{nella regola ho ottimizzato rispetto al codice.}

\bugnelcodice{nel codice invece di chiamare il refiner passando un tipo atteso
come nella regola chiamiamo con None e poi unifichiamo a mano: perch\`e'?}

%Caso NEL CODICE IN CASO DI FALLIMENTO RICOMINCIAMO SENZA APPLICARE LA REGOLA APPOSITA PER IL CONST SPERANDO DI METTERE UNA COERCION: MANCA\\

%%%%%%%%%%%%%%%%%%%%%%%%%%%%%%%%%%%%%%%%%%%%%%%%%%%%%%%%%%%%%%%%%%%%%%%%%%%
%%%%%%%%%%%%%%%%%%%%%%%%%%%%%%%%%%%%%%%%%%%%%%%%%%%%%%%%%%%%%%%%%%%%%%%%%%%
%%%%%%%%%%%%%%%%%%%%%%%%%%%%%%%%%%%%%%%%%%%%%%%%%%%%%%%%%%%%%%%%%%%%%%%%%%%
\section{Extension to placeholders}\label{sec:raw}

We consider here the first extension of our external syntax, obtained
introducing linear placeholders for missing terms and for vectors of
missing terms of unknown length. The latter are only accepted in argument
position, even if we will enforce this only in the refinement algorithm and
not in the syntax.
The new syntax is obtained extending the one given in Table~\ref{tab:terms} with the new productions of Table~\ref{tab:imptermsx}. Placeholders are also
called implicit arguments in the literature, but that terminology is
ambiguous since it is also used for arguments that can be entirely
omitted in the concrete syntax.

In a concrete implementation, user defined notations are used to further enlarge
the external syntax. User defined notations behave as macros; macro expansion
gives back a term in the external syntax we consider here. In particular,
thanks to user defined notations, it is possible to entirely omit the typing
information in binders, like in calculi typed \`a la Curry. Omitted
types are turned into placeholders during the macro expansion phase.
Implicit arguments can also be simulated by defining notations that insert
into applications a fixed  number of placeholders or vectors of placeholders
in appropriate positions.

\begin{table}[t]
\begin{displaymath}
\begin{array}{llll}
 t & ::= & \ldots & \\
   &  |  & ? & \mbox{placeholder} \\
   &  |  & \mldots & \mbox{arbitrarily many placeholder arguments} \\
\end{array}
\end{displaymath}
\caption{
\label{tab:imptermsx}
CIC terms syntax II - Placeholders}
\end{table}

A placeholder $?$ differs from a metavariable $?_i$ in the fact that
it has no sequent associated to it. The intended associated sequent allows
in $?$ occurrences of all variables in the context of $?$. Moreover, the
type of $?$ is meant to be the one determined by the context. This
information is made explicit by the refinement algorithm that turns each
placeholder into a corresponding metavariable.

Placeholders occur only linearly in the term (i.e. every occurrence of a
placeholder is free to be instantiated with a different term). Non linear
placeholders are not allowed since two occurrences could be in contexts
that bind different set of variables and instantiation with terms that
live in one context would make no sense in the other one.

Similarly, substitution is not allowed on placeholders since a placeholder
occurrence does not have a corresponding explicit substitution.

For both previous reasons, reduction is not allowed on terms in the external
syntax that contain placeholders: the reduction, conversion and unification
judgements only make sense on refined terms.

Intuitively, a vector of placeholders can be instantiated by the refiner with
zero or more metavariables. In our algorithm we adopt a lazy semantics:
vectors of placeholders can only be used in argument position and a
vector is expanded to the minimal number of metavariables that make the
application well typed. Bi-directionality, i.e. the knowledge about the
expected type for the application, is required for the lazy semantics.
Indeed, without the expected type, an expansion could produce a locally
well typed application whose inferred type will not match later on with the
expected one.

We extend the \rec{}, \recE{}, \eatprodslabel{} and \unimany{} algorithms
with new rules for single placeholders and for vectors of placeholders.

\begin{thm}[Correctness]
The \coer{}, \forcetotypelabel{}, \rec{}, \recE{}, \eatprodslabel{}, and \unimany{} algorithms extended with the set of rules presented in this section obey
their specification for all admissible $\preccurlyeq$ that include the
$\preccurlyeq'$ relation defined as follows:
$t' \preccurlyeq' t$ when $t'$ is obtained from $t$ by replacing single placeholders with terms and vectors of placeholders with vectors of terms.
In particular, the algorithms are correct w.r.t. $\preccurlyeq'$.
\end{thm}
\begin{proof}
Every admissible $\preccurlyeq$ that includes $\preccurlyeq'$ also includes
the identity. Thus we do not need to re-establish the result on the rules
given in the previous sections.
Correctness of the new rules given in this section is established by rule
inspection.
\end{proof}

\[
\infrule[(\rec{}\mathrm{-placeholder})]
 { \xP{} \leadsto \xP{} \cup \{
     \xG{} \vdash ?_l :\; \Type_\top \;,\; 
     \xG{} \vdash ?_k :\; ?_l \;,\;
     \xG{} \vdash ?_j :\; ?_k \}}
 {\Refinex{}{}{}{?}{?_j}{?_k}{'}{}}
\]\bigskip
\[
\infrule[(\recE{}\mathrm{-placeholder})]
 { \xP{} \leadsto \xP{} \cup \{
     \xG{} \vdash ?_k :\; T \}}
 {\RefinexE{}{}{}{?}{T}{?_k}{'}{}}
\]\bigskip
\[
%\infrule[(\eatT\mathrm{-\placeholder-0})]
\infrule[(\eatT\mathrm{-\placeholder-0})]
 %{\Eatprods{}{}{}{t~(x_1 := v_1 : T_1)~\ldots~(x_r := v_r : T_r) :\; T}{u_1~\ldots~u_n}{v}{V}{'}{'}}
 {\Eatprods{}{}{}{t~\vector{(x_r := v_r : T_r)} :\; T}{\vector{u_n}}{v}{V}{'}{'}}
 %{\Eatprods{}{}{}{t~(x_1 := v_1 : T_1)~\ldots~(x_r := v_r : T_r) :\; T}{\mldots~u_1~\ldots~u_n}{v}{V}{'}{'}}
 {\Eatprods{}{}{}{t~\vector{(x_r := v_r : T_r)} :\; T}{\mldots~\vector{u_n}}{v}{V}{'}{'}}
\]\bigskip
\[
\infrule[(\eatT\mathrm{-\placeholder+1})]
 {\Whd{}{}{}{T}{\Pi x_1:U_1.T_1}\\
  %\Eatprods{}{}{}{t~(x_1 := v_1 : T_1)~\ldots~(x_r := v_r : T_r) :\; T}{?~\mldots~u_1~\ldots~u_n}{v}{V}{'}{'}}
  \Eatprods{}{}{}{t~\vector{(x_r := v_r : T_r)} :\; T}{?~\mldots~\vector{u_n}}{v}{V}{'}{'}}
 %{\Eatprods{}{}{}{t~(x_1 := v_1 : T_1)~\ldots~(x_r := v_r : T_r) :\; T}{\mldots~u_1~\ldots~u_n}{v}{V}{''}{''}}
 {\Eatprods{}{}{}{t~\vector{(x_r := v_r : T_r)} :\; T}{\mldots~\vector{u_n}}{v}{V}{''}{''}}
\]\bigskip
The rule $(\eatT\mathrm{-\placeholder-0})$ is meant to take
precedence over $(\eatT\mathrm{-\placeholder+1})$.
The second
is applied when the first one fails (local backtracking). 

\begin{thm}[Termination]
The \rec{} % and \recE{} 
algorithm defined by the set of rules presented above with the 
addition of
$(\rec{}\mathrm{-placeholder})$,
$(\recE{}\mathrm{-placeholder})$,\\
$(\eatT\mathrm{-\placeholder-0})$ and
$(\eatT\mathrm{-\placeholder+1})$
terminates.
\end{thm}
\begin{proof}
Rules $(\rec{}\mathrm{-placeholder})$ and $(\recE{}\mathrm{-placeholder})$
terminate. The proof that $\eatT$ terminates is, as before, by induction
on the list of arguments that follow $\eatsep$. The rule
$(\eatT\mathrm{-\placeholder-0})$ terminates by induction hypothesis.
The rule $(\eatT\mathrm{-\placeholder+1})$ deserves an accurate treatment.
The check over $T$, asking it to be a product, is to avoid divergence.
Since the input $T$ is a well typed type, also $T_1$ is, and thus it
admits a normal form $T_1'$ in which $x$ may occur. The recursive call
does necessarily trigger the rule $(\eatT{}\mathrm{-product})$ that
will substitute a metavariable $?_j$ for $x$ in $T_1$. Thanks to the reduction
rules of CIC, reported in the appendix, substituting a variable for a 
metavariable declared in \xP{} (and not in \xS{}) does not change the normal 
form, meaning that $T_1'[x/?j] \whdlabel T_1'[x/?j]$. Thus the rule
$(\eatT\mathrm{-\placeholder+1})$ can be applied only a finite number
of times, and the number of products in $T$ is an upper bound.
\end{proof}

%CSC: non ci metterei la mano sul fuoco! terminazione non banale!

\bugnelcodice{la $has\_some\_more\_pis$ nel codice in verit\`a non dice quello
 che il nome indica. Quella nell'articolo cambiata.}

The next two rules for the \unimany{} judgment follow the same schema of the
ones for \eatprodslabel.

\[
\infrule[(\eatt\mathrm{-\placeholder-0})]
 %{\Unimanyx{}{}{}{t_1~\ldots~t_m}{v_1~\ldots~v_n}{t_1'~\ldots~t_n'}{u_1~\ldots~u_l}{'}{'}}
 {\Unimanyx{}{}{}{\vector{t_m}}{\vector{v_n}}{\vector{t_n'}}{\vector{u_l}}{'}{'}}
 %{\Unimanyx{}{}{}{\mldots~t_1~\ldots~t_m}{v_1~\ldots~v_n}{t_1'~\ldots~t_n'}{u_1~\ldots~u_l}{'}{'}}
 {\Unimanyx{}{}{}{\mldots~\vector{t_m}}{\vector{v_n}}{\vector{t_n'}}{\vector{u_l}}{'}{'}}
\]

\[
\infrule[(\eatt\mathrm{-\placeholder+1})]
 %{\Unimanyx{}{}{}{?~\mldots~t_1~\ldots~t_m}{v_1~\ldots~v_n}{t_1'~\ldots~t_n'}{u_1~\ldots~u_l}{'}{'}}
 {\Unimanyx{}{}{}{?~\mldots~\vector{t_m}}{\vector{v_n}}{\vector{t_n'}}{\vector{u_l}}{'}{'}}
 %{\Unimanyx{}{}{}{\mldots~t_1~\ldots~t_m}{v_1~\ldots~v_n}{t_1'~\ldots~t_n'}{u_1~\ldots~u_l}{'}{'}}
 {\Unimanyx{}{}{}{\mldots~\vector{t_m}}{\vector{v_n}}{\vector{t_n'}}{\vector{u_l}}{'}{'}}
\]\bigskip
\begin{thm}[Termination]
The \rec{} % and \recE{} 
algorithm defined by the set of rules presented above with the 
addition of
$(\eatt\mathrm{-\placeholder-0})$ and
$(\eatt\mathrm{-\placeholder+1})$
terminates.
\end{thm}
\begin{proof}
The rule $(\eatt\mathrm{-\placeholder-0})$ makes a recursive call
on the same list of arguments $\vector{v_n}$ but consumes a $\mldots$,
and no other rule of the refiner adds one, so it can be repeated only
a finite number of times. The recursive call in the rule 
$(\eatt\mathrm{-\placeholder+1})$ can trigger only rules
$(\eatt\mathrm{-empty})$
and
$(\eatt\mathrm{-base})$. The former terminates immediately, the latter
will do a recursive call consuming one argument in $\vector{v_n}$, and thus
terminates.
\end{proof}
Note that inlining the latter would lead to a rule whose termination is
trivial to see, but we preferred to present the algorithm in a more
modular way.

%\subsection{Remarks}
\begin{example}[Vector of placeholders]
Assume a theorem $\tau \in \env{}$ shows that
$\forall x : \nat{}. P~x \to Q~x$.
The proof context may contain a natural number $y$ and optionally a 
proof $H$ that $y$ validates $P$. Different proofs or proof styles
may use the same theorem $\tau$ in different ways. 
For example, one may want to perform
forward reasoning, and tell the system to assume $(Q~y)$ providing the
following proof for it
\[ y : \nat{}; H : P~y \vdash \tau~H\]
Nevertheless, sometimes $H$ is not known, and the user may want to
tell the system he has intention to use the theorem $\tau$ on $y$,
and prove $(P~y)$ later.
\[ y : \nat{} \vdash \tau~y\]
While the latter application is well typed, the first is not, since
the first argument of $\tau$ must be of type \nat{}. Nevertheless, the type
of $H$ depends on $y$, thus the term $(\tau~?~H)$ would refine to the
well typed term $(\tau~y~H)$ of type $(Q~y)$.

The vector of placeholders enables the system
to accept both terms originally written by the user.
In the first case $(\tau\mldots H)$ would expand to $(\tau~?~H)$ thanks to 
$(\eatt\mathrm{-\placeholder+1})$, and refine to $(\tau~y~H)$.
In the second case $(\tau\mldots y)$ would refine to $(\tau~y)$
thanks to $(\eatt\mathrm{-\placeholder-0})$. This suggests
defining $(\tau\mldots)$ as a notation for the theorem
$\tau$, obtaining a cheap implementation of what other systems
call prenex implicit arguments: the first n arguments of an application
whose head has a dependent type like
$\Pi \vector{x_n : T_n}.\Pi \vector{y_m : P(\vector{x_n})}.T$ 
can be omitted, and are inferred thanks to the 
dependencies in the types of the m following arguments.
As a bonus, in case the user wants to pass one of the implicit
arguments there is no need to temporarily disable the mechanism, since
the expansion of $\mldots$ is computed on the fly and automatically adapts to
its context.
\qed{}
\end{example}

%%%%%%%%%%%%%%%%%%%%%%%%%%%%%%%%%%%%%%%%%%%%%%%%%%%%%%%%%%%%%%%%%%%%%%%%%%
\section{Coercions}\label{sec:coercions}

Coercions are explicit type casts. While the
literature~\cite{coercivesubtyping} considers them mostly as a device to mimic
sub-typing in a calculus lacking it, they have other interesting applications.
The refiner of Matita inserts coercions in three locations: 
\begin{iteMize}{$\bullet$}
\item around the argument of an application
\item around the head of an application
\item around the type of an abstraction
\end{iteMize}
The first case is the most common one, and is the one that can easily be
explained in terms of sub-typing. For example, if one applies an operation
defined over integers $\mathbb{Z}$ to an argument lying in the type
of natural numbers $\nat{}$, the system injects the argument into the
right type by means of the obvious, user declared, cast function mapping
naturals into the non negative fragment of $\mathbb{Z}$.

The second case is handy in two situations. First when the head of the
application is implicit in the standard notation, like in $3x$ where the
intended head constant is the multiplication but in the input it happens to be
$3$. The second is when the head constant has a non ambiguous interpretation as
a function, but is not. For example a set may act as its characteristic
function.

The last case is recurrent when algebraic structures are encoded as 
dependently typed records~\cite{pollackFAC02} embedding the type (or carrier)
for the elements together with the operations and properties defining the
structure. In that case, one may want to state a theorem quantifying over a
structure $G$, say a group, and some elements in that group. However the
statement $\forall G:\mbox{Group}.\forall x,y:G.P(x,y)$ is ill-typed since $G$
is a term (of type $\mbox{Group}$) but is used as a type for $x$ and $y$. The
intended meaning is clear: $x$ and $y$ lie in the carrier type of $G$. The
system can thus insert around $G$ the projection for the carrier component of
the $\mbox{Group}$ record.

\begin{definition}[Coercion set ($\Delta$)] A coercion set $\Delta$
is a set of pairs $(c, k)$ where $c$ is a constant in \env{} and $k$ 
is a natural number smaller than the arity of $c$ (i.e. $k$ points to
a possible argument of $c$)
\end{definition}

In the literature the coercion set is usually represented as a graph.
Given a coercion $(c, k)$ such that
$(c:\Pi x_1 : T_1 \ldots \Pi x_k : T_k \ldots \Pi x_n : T_n. T) \in \env{}$,
$T_k$ and $T$ %(or some approximation of them, like their head constant)
are nodes in the graph, and $c$ is an edge from $T_k$ to $T$.
Most coercion implementation, like the one of Coq, Lego and Plastic, assume 
$\Delta$ to be a graph validating a property called coherence. 
This property states that $\Delta$ is an acyclic graph with at most one path
linking every pair of nodes. This property enables to employ a straightforward 
algorithm to look for a sequence of coercions linking two non adjacent 
nodes in the graph. 

In Matita, for various reasons detailed in~\cite{tassi-phd}, $\Delta$ is
not a graph, but a set of arcs for the transitive closure of the graph.
Every time a coercion $c$ is declared by the user, and thus added to $\Delta$, 
the following set of automatically generated composite coercions is also 
added to $\Delta$.
\[
\{ c_i \circ c \circ c_j | c_i \in \Delta \land c_j \in \Delta\}
\cup
\{ c \circ c_j | c_j \in \Delta\}
\cup
\{ c_i \circ c | c_i \in \Delta\}
\]
Of course the $\circ$ operator here is partial, and only well typed
composite coercions are actually considered.
This design choice enables the coercion lookup operation to be single step,
since the set is already transitively closed. Moreover, since composite 
coercions are defined constants in \env{}, the term resulting after a cast
is smaller if compared with the one obtained inserting the corresponding
chain of user declared coercions. Last, allowing $k$ to differ from $n$ is 
a peculiarity of Matita. When $k \neq n$ the application of the coercion
creates new uninstantiated metavariables that correspond to proof
obligations. This will be detailed later on.

The last detail worth mentioning is that, all systems known to the authors with
the notable exception of Plastic~\cite{callaghan00coherence}, adopt some
approximated representation for the nodes in the coercion graph, usually the
name of the head constant of the source and target types.  This results in a faster
lookup in the coercion graph, but the coherence check is also strengthened.
In particular, in a calculus with dependent types, different, but similar, 
coercions may not be allowed to be declared. Matita drops
the coherence check, or better changes it into a warning, and enables the user
to attach to coercions a priority: coercions from and to the same approximation
of types are all tried according to user defined priorities.

\begin{specification}[Coercion lookup $(\rightarrowtail-\;\Delta)$]
Given a context $\xG{}$, substitution $\xS{}$ and proof problem
$\xP{}$, 
all assumed to be well formed, 
two types $T_1$ and $T_2$, this function 
returns an explicit cast $c~?_1~\ldots~?_k~\ldots~?_n$ for the metavariable
of index $k$ and its type $T'$.
It is denoted by:\\
\[\lookforcx{}{}{}{T_1}{T_2}{k}{c~?_1~\ldots~?_k~\ldots~?_n}{T'}{'}\]
Precondition (parametric in $\approx$):
\[\WF{\xPSG{}{}{}} \lland 
(c, k) \in \Delta \lland 
(c:\Pi x_1 : T_1 \ldots \Pi x_k : T_k \ldots \Pi x_n : T_n. T) \in \env{} \lland
\]
\[
T_k \approx T_1 \lland T \approx T_2
\]
Postcondition:
\[\WF{\xP{'}} \lland 
\xPSG{'}{}{} \vdash c~?_1~\ldots~?_k~\ldots~?_n : T'
\]
\end{specification}
\noindent We denoted by $\approx$ the approximated comparison test
used to select from $\Delta$ a coercion $c$ from 
$T_1$ to $T_2$. A proper definition of
$\approx$ is not relevant for the present paper, but we can anyway say that 
Matita compares the first order skeleton of types obtained by dropping bound 
variables, metavariables and higher order terms,
and that this skeleton can be made less precise 
on user request. We will give an account of this facility in the example
that will follow.

The new metavariables $?_1,\ldots,?_n$ generated by the lookup operation
are all declared in the new proof problem $\xP{'}$.
The number of metavariables to which $c$ is applied to is defined when the coercion is declared and may be less than the arity of $c$. In the latter case
$T$ is a product and the coercion casts its $k$-th argument to be a function.
The position $k$ of the casted argument is user defined as well.  The
coerced term has then to be later unified with~$?_k$. 

\begin{thm}[Correctness]
The \coer{}, \forcetotypelabel{}, \rec{}, \recE{}, \eatprodslabel{}, and \unimany{} algorithms extended with the set of rules presented in this section obey
their specification where $\preccurlyeq''$ is the following admissible order
relation: $t' \preccurlyeq'' t$ when $t'$ is obtained from $t$ by 
replacing single placeholders with terms, vectors of placeholders with vectors of terms, and terms $u_k$ with terms convertible to
$(c~u_1~\ldots~u_k~\ldots~u_n)$ where $c$ is a coercion declared in
$\Delta$ for its $k^{th}$ argument.
\end{thm}
\begin{proof}
We do not need to re-establish correctness for the rules given in the previous
sections since $\preccurlyeq''$ is admissible and includes $\preccurlyeq'$.
Correctness of the new rules given in this section is established by rule
inspection as usual.
\end{proof}

In the following rule the coercion $c$ is applied to its argument
$t$ unifying it with $?_k$. The returned term $t'$ can still contains 
metavariables: $?_1 \ldots ?_{k-1}$ may appear in the type of $?_k$, thus
unifying $?_k$ with $t$ may instantiate them\footnote{In the case of
dependent types the unification of the types is a
necessary condition for the unification of the two terms, as claimed
by Strecker~\cite{strecker}.}, but $?_{k+1} \ldots ?_n$ do not appear in
the type of $?_1 \ldots ?_k$, and thus cannot be all instantiated.
This rule is applied as a fall back in case $\coer{}\mathrm{-ok}$ fails.

\[
\infrule[(\coer{}\mathrm{-coercion})]
 %{\Lookforcx{}{}{}{T_1}{T_2}{k}{c~?1~\ldots~?_k~\ldots~?n}{T_2'}{'}\\
 {\Lookforcx{}{}{}{T_1}{T_2}{k}{c~\vector{?_m}~?_k~\vector{?_n}}{T_2'}{'}\\
  %\Unifx{'}{}{}{?_k}{t}{''}{'} \\
  \Unifx{_=}{'}{}{}{?_k}{t}{''}{'} \\
  \Unifx{}{''}{'}{}{T_2'}{T_2}{'''}{''}}
 %{\Unifcoercex{}{}{}{t}{T_1}{T_2}{c~?1~\ldots~?_k~\ldots~?n}{'''}{''}}
 {\Unifcoercex{}{}{}{t}{T_1}{T_2}{c~\vector{?_m}~?_k~\vector{?_n}}{'''}{''}}
\]\bigskip
\begin{proof}[Correctness of $(\coer{}\mathrm{-coercion})$]
Since $?_k$ is unified with $t$ in the second premise of the rule,
by definition of unification we have $\xPSG{}{}{} \vdash ?_k \downarrow t$, 
%and thus $c~?1~\ldots~?_k~\ldots~?n \preccurlyeq t$.
and thus $c~\vector{?_m}~?_k~\vector{?_n}\preccurlyeq'' t$.
Moreover, the postconditions of coercion lookup 
$\stackrel{\Delta}{~\leadsto~}$ grant that 
$c~\vector{?_m}~?_k~\vector{?_n}$ has type 
$T_2'$ that is later unified with $T_2$. Thus the postconditions of the
unification algorithm allow us to prove that 
$c~\vector{?_m}~?_k~\vector{?_n}$ has a type convertible with $T_2$.
\end{proof}

The $\coer{}\mathrm{-coercion}$ rule automatically takes care of the insertion
of coercions around arguments of an application and around the types of an
abstraction.

The following extension to $\eatprodslabel$ take cares of insertion
around the head of an application.
% \[
% \infrule[(\eatT\mathrm{-coercion})]
%  {\Extendx{}{}{?_{1'}}{\Type}{}{?_{1}}{?_{1'}}{'} \\
%   \Extendx{}{;x_1:?_{1}}{?_{2'}}{\Type}{;x_1:?_{1}}{?_{2}}{?_{2'}}{'} \\
%   \Unifcoercex{'}{}{}{t}{T}{\Pi x:?_1.?_2}{t'}{''}{'} \\
%   \Eatprods{''}{'}{}{t' :\; \Pi x:?_1.?_2}{u_1~\ldots~u_n}{v}{V}{'''}{''}}
%  {\Eatprods{}{}{}{t :\; T}{u_1~\ldots~u_n}{v}{V}{'''}{''}}
% \]\\
% 
\[
\infrule[(\eatT\mathrm{-coercion})]
 {\Extendx{}{}{?_{1'}}{\Type_\top}{}{?_{1}}{?_{1'}}{'} \\
  \Extendx{}{;x_1:?_{1}}{?_{2'}}{\Type_\top}{;x_1:?_{1}}{?_{2}}{?_{2'}}{'} \\
  \Unifcoercex{'}{}{}{t~\vector{v_r}}{T}{\Pi x:?_1.?_2}{c~\vector{w_s}}{''}{'}\\
  \xG \vdash w_i : W_i \quad i \in \{1 \ldots s\}\\
  \Eatprods{''}{'}{}{c~\vector{(x_s := w_s : W_s)} :\; \Pi x:?_1.?_2}{u_1~\vector{u_n}}{v}{V}{'''}{''}}
 {\Eatprods{}{}{}{t~\vector{(x_r := v_r : T_r)}:\; T}{u_1~\vector{u_n}}{v}{V}{'''}{''}}
\]\bigskip

\begin{proof}[Correctness of $\eatT\mathrm{-coercion}$]
Its correctness follows trivially from the correctness of 
$(\coer{}\mathrm{-coercion})$.
\end{proof}

\begin{thm}[Termination]
The \rec{} % and \recE{} 
algorithm defined by the set of rules presented above with the 
addition of
$(\coer{}\mathrm{-coercion})$ and
$(\eatT\mathrm{-coercion})$
terminates.
\end{thm}
\begin{proof}
Rule $(\coer{}\mathrm{-coercion})$ clearly terminates.
The rule $(\eatT\mathrm{-coercion})$ issues a recursive call to
$\eatprodslabel$ without consuming $u_1$, but the only rule
that can be triggered is $(\eatT\mathrm{-product})$, that will immediately
consume $u_1$.
\end{proof}
Inlining $(\eatT\mathrm{-product})$ would
result in a rule that consumes some input and thus clearly terminates, 
but would be way less readable.

\subsection{Implementation remarks}

Since we allow coercion arguments not to be inferred automatically (like proof
obligations) their type may depend on the coerced term (e.g. the proof that the
coerced integer is greater than zero has an instance of the coerced integer in
its type, and the corresponding metavariable will have index greater than $k$).

\begin{example}[Coercion with side conditions]
Consider the following coercion set, declaring the coercion 
$\mathrm{v\_to\_nel}$ from vectors to non empty lists.
\[
\Delta = \{ (\mathrm{v\_to\_nel, 3}) \}
\]
The environment holds the following type for the coercion:
\[
(\mathrm{v\_to\_nel}:
 \Pi A : \Type.\Pi n : \nat{}.\Pi v : \mathrm{Vect}~A~n.n > 0 \to
 \exists l : \mathrm{List}~A, \mathrm{length}~l > 0 ) \in \env{}
\]
Now consider the term $t = (\mathrm{Vcons}~\nat{}~0~(\mathrm{Vnil}~\nat{})~2)$ 
and the following coercion problem:
\[
\frac{
 \begin{array}{l}
 \Lookforcx{}{}{}
   {\mathrm{Vect}~\nat{}~(0+1)}
   {(\exists l : \mathrm{List}~\nat{},\mathrm{length}~l > 0)}
   {\\\qquad 3}{\mathrm{v\_to\_nel}~?_1~?_2~?_3~?_4}
   {(\exists l : \mathrm{List}~?_1,\mathrm{length}~l > 0)}{'}\\
 \Unifx{_=}{'}{}{}{?_3}{t}{''}{'} \\
 \Unifx{}{''}{'}{}
    {(\exists l : \mathrm{List}~?_1,\mathrm{length}~l > 0)}
    {(\exists l : \mathrm{List}~\nat{},\mathrm{length}~l > 0)}{'''}{''}
  \end{array}
}{
\Unifcoercex{}{}{}{t}
  {\mathrm{Vect}~\nat{}~(0+1)}
  {(\exists l : \mathrm{List}~\nat{},\mathrm{length}~l > 0)}
  {\mathrm{v\_to\_nel}~?_1~?_2~?_3~?_4}{'}{'}
}
\]
\noindent
where the final proof problem and substitutions are:
\[
\xP{} = \{ \xG{} \vdash ?_4 :\; ?_2 > 0\}
\]
\[
\xS{} = \{ \xG{} \vdash ?_1 := \nat{} : \Type,~ 
           \xG{} \vdash ?_2 := 0+1 : \nat {},~
           \xG{} \vdash ?_3 := t : \mathrm{Vect}~\nat{}~(0+1)
\}
\]
Note that $?_4$ is still in \xP{}, thus it represent a proof obligation
the user will be asked to solve.
Also note that the following coercion could be declared as well,
with a higher precedence. It is useful since it does not
open a side condition when the type of the coerced vector 
is explicit enough to make the proof that it is not empty
constant (not depending on \xG{} nor on the vector but just on its type) 
and thus embeddable in the body of
the coercion.
\[
(\mathrm{nev\_to\_nel},\; 
 \Pi A : \Type.\Pi n : \nat{}.\Pi v : \mathrm{Vect}~A~(n+1).
 \exists l : \mathrm{List}~A, \mathrm{length}~l > 0 ) \in \env{}
\]
The system would thus try $\mathrm{nev\_to\_nel}$ first, and fall back
to $\mathrm{v\_to\_nel}$ whenever needed.
\qed{}
\end{example}

Also note that this last coercion can be indexed as a cast from
$(\mathrm{Vect}~\_~(\_+1))$ to 
$(\exists l : \mathrm{List}~\_, \mathrm{length}~l > 0)$ or in a 
less precise way.
For example the approximation of the source type could be relaxed
to $(\mathrm{Vect}~\_~\_)$. This will force the system to try to apply
this coercion even if the casted term is a vector whose length
is not explicitly mentioning $+1$, but is something that unifies with $?_j+1$.
For example the length $1*2$ would unify, since its normal form is $(0+1)+1$.

%REGOLA EAT-T-COERCE
%CSC: Ho messo io la creazione delle due meta. Nel codice si usa $\Pi x:?.?$ e c'\`e un flag per non fare unificazione, che \`e pi\`u performante. Che cosa descrivere?\\
%%%%%%%%%%%%%%%%%%%%%%%%%%%%%%%%%%%%%%%%%%%%%%%%%%%%%%%%%%%%%%%%%%%%%%%%%%%
%%%%%%%%%%%%%%%%%%%%%%%%%%%%%%%%%%%%%%%%%%%%%%%%%%%%%%%%%%%%%%%%%%%%%%%%%%%
%%%%%%%%%%%%%%%%%%%%%%%%%%%%%%%%%%%%%%%%%%%%%%%%%%%%%%%%%%%%%%%%%%%%%%%%%%%

\section{Comparison with related work on Type Inference}
Type inference is a very widely studied field of computer science.  Nevertheless
to the authors' knowledge there is no precise account of a type inference
algorithm for the full CIC calculus in the literature. 

The extension to the typing algorithm of CIC with explicit casts in~\cite{saibi-inheritance} 
follows the same spirit of our refinement
algorithm for raw terms. However the work by Saibi does not handle placeholders
nor metavariables, and the presentation is in fact quite distant from the 
actual implementation in the Coq interactive prover. 

Another work in topic is \cite{phd-norell} where Norell describes the
bi-directional type inference algorithm implemented in the Agda interactive
prover. He presents the rules for a core dependently typed calculus
enriched with dependent pairs. Unfortunately he omits the rules 
for its extension with inductive types.
It is thus hard to tell if Agda exploits the type expected by
the context to type check inductive constructors as in rule
$(\recE{}\mathrm{-appl-}k)$.\\
Agda does not provide an explicit $\mldots$ placeholder but uses the expected
type to know when it is necessary to pad an application with meta variables
in order to reduce the arity of its type.
In our setting this is equivalent to the following transformation: 
every application $(f~\vector{a})$ is turned into $(f~\vector{a}~\mldots)$
whenever its expected type is known (i.e. not a metavariable).

One aspect that allows for a direct comparison with Coq and Agda is the
handling of implicit arguments. In both Agda and Coq, abstractions corresponding
to arguments the user can freely omit are statically labelled as such.
The systems automatically generate fresh metavariables as arguments to
these binders and the type inference algorithm eventually instantiates them.
Both systems give the user the possibility to locally override the implicit
arguments mechanism.  In Coq the user can prefix the name of a constant with
the @ symbol, while in Agda the user can mark actual arguments as implicit
enclosing them in curly braces.  This escaping mechanism is required because
many lemmas admit multiple and incompatible lists of implicit arguments.\\
As an example, consider a transitivity lemma
$\mbox{eqt} : \forall x,y,z. x = y \to y = z \to x = z$.
When used in a forward proof step the user is likely to pass as arguments
a proof $p$ that $a=b$ and a proof $q$ that $b=c$ like in $(\mbox{eqt}~p~q)$
to put in his context the additional fact $a=c$.
In that case values for $x, y$ and $z$ are determined by the types of
$p$ and $q$. On the contrary if the lemma is used in backward proof step
to prove that $a=c$, no value for $y$ can be inferred, thus the user 
is likely to use the lemma as in $(\mbox{eqt}~b)$ and expect the system
to open two new goals: $a=b$ and $b=c$. 
The two different uses of \mbox{eqt} make it impossible to 
statically attach to it a single list of implicit arguments and 
at the same time to never resort to an escape mechanism to temporarily
forget that list.\\
In Matita the user can simply use the $\mldots$ placeholder, thus no escaping
mechanism is required. In fact the type inference algorithm described in this
paper lets the user write $(\mbox{eqt}~\mldots~p~q)$ in the first case as well
as $(\mbox{eqt}~\mldots~c)$ in the second one.\footnote{A trailing $\mldots$ is
automatically added to any term used in backward proof step.}\\

The lack of a complete and formal study of type inference for raw CIC terms is
probably due to the many peculiarities of the CIC type system, in particular
inductive and dependent types, explicit polymorphism and the fact that type
comparison is not structural, but up to computational equivalence. We thus try
to position our work with respect to some of the main approaches adopted by
type inference algorithms designed for programming languages.

\subsection{Greedy versus delayed constraint solving}
The most notable example of type inference algorithm based on constraint
solving is the one adopted for the Agda system~\cite{phd-norell}. Agda 
is based on a dependently typed programming language quite similar to CIC, 
but is designed for programming and not for writing proofs. The type 
inference algorithm collects constraints and checks for their satisfiability.
Nevertheless, their solution is not recorded in the terms. This enables
the user to remove an arbitrary part of an already type checked term
and have the typing of its context not influenced by the term just removed.
While this ``compositionality'' property is desirable for
programming in a language with dependent types, it is not vital for proof 
systems, where one seldom edits by hand type checked terms. 
% Moreover,
% basing the type inference algorithm on a constraint solver may make 
% error reporting extremely fragile and inaccurate.

A strong characteristic of constraint based type inference is precise
error reporting, as described in~\cite{StuckeySW06}. Even if it the heuristics
adopted in Matita~\cite{mcs2008} to discard spurious error reports are slightly
more complex than the ones proposed by Stuckey, we believe that they provide a
similar precision.

% On the contrary, 
Greedy algorithms~\cite{dunfieldgreedy}, like the one presented in this paper,
are characterized by a very predictable behavior, at the cost of being forced
to take early decisions leading to the rejection of some possibly well typed
terms. Also remember that unification has to take computation into account,
and user provided functions are known to be total only if they are well typed.
Thus the resolution of type constraints cannot be delayed for long.  According
to our experience, predictability compensates for the extra type annotations the
user is sometimes required to produce to drive the greedy algorithm towards a
solution.

\subsection{Unification based versus local constraint solving}
Many algorithms to infer a polymorphic type for a program prefer
to avoid the use of unification~\cite{piercelocaltype} since
unification variables may represent type constraints coming from 
distant, loosely related, sub-terms. Moreover a bi-directional 
approach pushes the type constraints of the context towards sub-terms,
making it effectively possible to drop unification altogether.
These approaches also scaled up to types with some sort of dependency over 
terms, as in~\cite{pfenning3}. 

Interactive provers based on type theory are for (good and) historical reasons
based on the two twin approaches. A small kernel based on decision procedures
type checks (placeholder free) terms, and a refiner based
on heuristics deals with terms with holes performing
type inference. Since the kernel is the key component of the system, the one
that must be trusted, the language is designed to allow the type checking 
algorithm to be as simple as possible. Explicit polymorphism makes type
checking CIC terms decidable while allowing the same degree of polymorphism
as $\mathcal{F}^\omega$. These explicit type annotations are
usually left implicit by the user and represent long-distance constraints.
In this context unification seems to be a necessary device. Moreover,
the most characterizing feature of CIC 
is that types are compared taking computation into account, and that types
can contain terms, in particular functions applications. Thus the kernel
is equipped with a quite elaborate machinery to compute recursive functions
and unfold definitions. Type inference has to provide a similar machinery,
and possibly extend it to handle types containing metavariables.
This extension is commonly named higher order unification, and it is 
a really critical component of an interactive prover. Recent important
developments~\cite{canonical-structures} heavily rely on a user-extensible 
unification algorithm~\cite{unification-hints}, using it as a predictable
form of Prolog-like inference engine. In other words, unification can
be employed to infer terms (content) while type inference is employed
to infer types and type annotations in the case of explicit polymorphism.
For these many reasons, we believe that developing type inference on top of
unification is a sound decision probably necessary to scale to a rich
type system like CIC.

%%%%%%%%%%%%%%%%%%%%%%%%%%%%%%%%%%%%%%%%%%%%%%%%%%%%%%%%%%%%%%%%%%%%%
\section{Conclusion}

In this paper we studied the design of an effective refinement algorithm
for the Calculus of (Co)Inductive Constructions. Its effectiveness has been
validated in all the formalizations carried on using the Matita interactive 
theorem prover~\cite{TS11,armentano,asperti-ricciotti,matitapoplmark}, 
whose refiner is based on the algorithm described in this
paper. Once again we stress that the refiner component, while not being
critical for the correctness of the prover, is the user's main interlocutor,
and is thus critical for the overall user's experience.

This algorithm is also the result of the complete rewrite the Matita ITP
underwent in the last couple of years. The refiner algorithm described in this
paper amounts to approximatively 1600 lines of OCaml code, calling the higher
order unification algorithm that amounts to a bit less than 1900 lines. To give
a term of comparison to the reader, the kernel of Matita, written by the same
authors, amounts to 1500 lines of data structures definitions and basic
operations on them, 550 lines of conversion algorithm and 1400 of type
checking. More than 300 lines of the type checking algorithm are reused
by the refiner for checking inductive types positivity conditions and
recursive or co-recursive functions termination or productivity. 

% 
%   143 nCicRefineUtil.ml
%  1262 nCicRefiner.ml
%   216 nCicCoercion.ml
% -----
%  1621
% 
%   549 nCicMetaSubst.ml
%   408 nCicUnifHint.ml
%   929 nCicUnification.ml
% -----
%  1886
% 
%   424 nCicEnvironment.ml
%   156 nCic.ml
%   370 nCicPp.ml
%   157 nReference.ml
%   351 nCicUntrusted.ml
%    60 nUri.ml
% -----
%  1518 total
% 
%   456 nCicReduction.ml
%   103 nCicSubstitution.ml
% -----
%   559 total
% 
% 1402 nCicTypeChecker.ml

On top of this refinement algorithm all
primitive proof commands have been reimplemented.  In the old implementation
they were not taking full advantage of the refiner, partially for historical
reasons, partially because it was lacking support for placeholder vectors and
bi-directionality was not always exploited. The size of the code is now
48.9\% of what it used to be in the former implementation.  In particular, it
became possible to implement many proof commands as simple ``notations'' for
lambda-terms in external syntax.

A particularity of this work is that the presented algorithm deals with
completely raw terms, containing untyped placeholders, whose only precondition
is to be syntactically well formed. In addition it also supports a very
general form of coercive sub-typing, where inserting the explicit cast may leave
uninstantiated metavariables to be later filled by the user. This eased the
implementation of subset coercions in the style of~\cite{russell}, but that
topic falls outside the scope of the present paper and is thus not discussed.

The algorithm could be enhanced adding more rules, capable of propagating more 
typing information. For instance, a specific type forcing rule for 
$\beta$-redexes (suggested by a referee) could be in the form

\[
\infrule[(\recE{}\mathrm{-beta})]
 {\Forcetotype{}{}{}{U}{U'}{s}{'}{'} \\
  \RefinexE{'}{'}{}{u}{U'}{u'}{''}{''}\\
  \RefinexE{''}{''}{; x:U'}{t}{E[u'/x]}{t'}{'''}{'''}}
 {\RefinexE{}{}{}{(\lambda x:U.t)\; u}{E}{(\lambda x:U'.t')\; u'}{'''}{'''}}
\]\smallskip

\noindent enabling the system to propagate the expected type to the abstraction 
(compare this rule with ($\recElabel\mathrm{-letin}$)). In practice the 
advantage of this rule is limited, since it is quite infrequent for a user to 
write $\beta$-redexes. A type forcing rule for pattern matching based on the 
same principles, propagating the expected type to the return type of the 
\mbox{\verb[language=grafite]+match+}, could be of greater value, since this 
construct is more likely to come from the user input. We will consider adding
such rules in a future implementation.

The refinement algorithm we presented already validates many desired 
properties, like correctness and termination. 
Nevertheless we did not even state
the relative completeness theorem. In a simpler framework, admitting
most general unifiers, one could have stated that given an oracle for 
unification, the algorithm outputs a well typed refinement every time it
is possible and that any other refinement is less general than the
produced one. Unluckily CIC is higher order and does not admit most
general unifiers. To state the relative completeness theorem one
has to make the oracle aware of the whole refinement procedure and the oracle
has to guess a unifier (or all of them) such that the remaining
refinement steps succeed. This makes the theorem way less interesting.
Alternatively one would have to add backtracking to compensate for errors made
by the oracle, and make the algorithm distant from the implemented one,
that is essentially greedy and backtracking free.

The algorithm presented in the paper is clearly not relatively complete.
For example, the rules given in Section~\ref{sec:mono} do not accept
the term $f~c$ where $c: \mathbb{N}$ and $f: \match{?_1}{\mathbb{N}}{\lambda x.\Type}[O \Rightarrow \mathbb{N} ~|~ S~(x: \mathbb{N}) \Rightarrow \mathbb{N} \to \mathbb{N}]$. The term is however refineable, for instance by instantiating
$?_1$ with $S~O$. To obtain a relatively complete algorithm, we could add
additional rules based on the invocation of the unifier on difficult problems.
For instance, for the example just shown it would be sufficient to unify
$\mathbb{N} \to ?_2$ with $\match{?_1}{\mathbb{N}}{\lambda x.\Type}[O \Rightarrow \mathbb{N} ~|~ S~(x: \mathbb{N}) \Rightarrow \mathbb{N} \to \mathbb{N}]$ for
a fresh metavariable $?_2$. However, we know in advance that the efficient
but incomplete algorithm implemented in Matita always fails on such difficult
unification problems. The same holds for the similar algorithm implemented
in Coq. Therefore a relatively complete version of the algorithm would remain
only of theoretical interest.

The following weaker theorem, which establishes completeness on well typed terms
only, can be easily proved by recursion over the proof tree and by inspection
of all cases under the hypothesis that every pair of convertible terms are
unified by the identity metavariable instantiation.
\begin{thm}[Completeness for well typed terms]For all well formed $\xPSG{}{}{}$ and for all $t$ and $T$ such
that $\xPSG{}{}{} \vdash t:T$ we have 
$\RefinexE{}{}{}{t}{T}{t}{'}{'}$. \qed
\end{thm}

% ne parlavo oggi con Ugo
% una ocsa che si puo' rimarcare
% che perfino lo statement e' un casino
% perche' vorresti dire che
% se l'unificazione ogni volta ritorna QUELL'UNIFICATORE CHE E' GLOBALMENTE VALIDO
% allora hai completezza relativa
% cioe' l'unificazione non deve tornare solamente un unificatore a caso quando ce ne e' uno
% deve tornare quello che poi globalmente fa tipare il tutto
% 
% (la globalita' e' colpa dei tipi dipendenti)
% detto cosi', e' difficilmente enunciabile
% in pratica, per enunciarlo, bisognerebbe prevedere ovunque back-tracking
% unificazione inclusa
% ma sarebbe barare rispetto a quanto implementato che non fa backtracking

% \section*{Acknowledgement}
%  The authors wish to acknowledge fruitful discussions with A and B.

%% in general the use of bibtex is encouraged

\subsection*{Acknowledgments}
We deeply thank Jacques~Carette and the anonymous referees for their many 
observations and corrections.

\appendix
\label{sec:app}
\section{Syntax-directed type-checking rules}

The following appendix is an extract of the paper~\cite{ck-sadhana}
in which the reader can find all the details of the type checking
algorithm implemented in the Matita interactive prover.
A few aesthetic changes have been made to the adopted syntax to increase its
consistency with respect to the syntax adopted in this paper. The main
differences are summarised in the following list:
\begin{iteMize}{$\bullet$}
\item We use the membership relation over the \pts{} set to type sorts and
      products
\item The check for the consistency of the metavariable local substitution
      has been inlined in the rule
\item A new generic judgement $(r : T) \in \env$ has been introduced to
      provide a more compact syntax for the lookup of the type of a 
      generic object into the environment
\item We inlined several auxiliary functions that were used in the
      presentation of the typechecking rule for case analysis.

      This was made
      possible by the following abuse of notation:
      $\env,\Sigma,\Phi,\emptyset \vdash t_1 \mwhd \Pi \vector{x_i:T_i}.t_{n+1}$
      is a shortcut to mean that for all $i \in \{1 \ldots n\}$ \\
      $\whd{}{}{; x_1 : T_1; \ldots ;x_{i-1} : T_{i-1}}{t_i}{\Pi x_i:T_i.t_{i+1}}$ and $\whd{}{}{}{t_{n+1}}{t_{n+1}}$.

      Moreover, the rule presented in~\cite{ck-sadhana} is more liberal than the
      one presented here that just uses the test $(s,s') \in \elim(\pts)$ to
      check that a non informative data is never analyzed to obtain an
      informative one. The actual rules used in the kernel and the refiner of
      Matita also allow
      in every situation the elimination of inhabitants of singleton inductive
      types, whose definition is given in~\cite{ck-sadhana}.
\end{iteMize}

\label{sec:formal}
\noindent In this section, $\mathcal{I}$ will be short for
\[
\begin{array}{l}
\Prod \vector{x_l:U_l}. \textrm{\verb[language=grafite]+inductive+}\\
~ \qquad \qquad I^1_l : A_1 := k^1_1 : K^1_1 \ldots k^{m_1}_1 :
K^{m_1}_1\\ \qquad \textrm{\verb[language=grafite]+with+}~\ldots\\ \qquad
\textrm{\verb[language=grafite]+with+}
~I^n_l : A_n := k^1_n : K^1_n \ldots k^{m_n}_n : K^{m_n}_n
\end{array}
\]

\subsection{Environment formation rules}
Environment formation rules (judgement $\env \vdash WF$, function \verb+typecheck_obj+)

$$\infrule
 {}
 {\emptyset \vdash WF}
$$

$$\infrule
 {\env \vdash WF \andalso
  \mbox{$d$ undefined in $\env$} \andalso
  \env,\Sigma \vdash WF \andalso
  \env,\Sigma,\Phi \vdash WF \andalso\\
  \env,\Sigma,\Phi,\emptyset \vdash T: S \andalso 
  \env,\Sigma,\Phi,\emptyset \vdash S \mwhd S'~\mbox{where $S'$ is a sort}\\
  \env,\Sigma,\Phi,\emptyset \vdash b: T' \andalso
  \env,\Sigma,\Phi,\emptyset \vdash T \downarrow T'
 }
 {\env \cup ( \Sigma,\Phi,\mbox{\verb[language=grafite]+definition+}~d:T:=b) \vdash WF}
$$

$$\infrule
 {\env \vdash WF \andalso
  \mbox{$d$ undefined in $\env$} \andalso
  \env,\Sigma \vdash WF \andalso
  \env,\Sigma,\Phi \vdash WF \andalso\\
  \env,\Sigma,\Phi,\emptyset \vdash T: S \andalso
  \env,\Sigma,\Phi,\emptyset \vdash S \mwhd S' ~\mbox{ where $S'$ is a sort}
 }
 {\env \cup ( \Sigma,\Phi,\mbox{\verb[language=grafite]+axiom+}~d:T) \vdash WF}
$$

$$\infrule
 {\env \vdash WF \andalso
  \mbox{$\vector{f_n}$ undefined in $\env$} \andalso
  \env,\Sigma \vdash WF \andalso
  \env,\Sigma,\Phi \vdash WF \andalso\\
  \env,\Sigma,\Phi,\emptyset \vdash T_i: S_i \andalso
  \env,\Sigma,\Phi,\emptyset \vdash S_i \mwhd S'_i~\mbox{ where $S'_i$ is a sort}\\
  T_i = \Pi \vector{x^i_{p_i} : T^i_{p_i}}.T^i_{p_i + 1} \\
  \left.
  \begin{array}{l}
  \env,\Sigma,\Phi,[f_1:T_1;\ldots;f_n:T_n;\vector{x^i_{p_i} : T^i_{p_i}}] \vdash t_i: {T'}^i_{p_i + 1} \\
  \env,\Sigma,\Phi,[f_1:T_1;\ldots;f_n:T_n;\vector{x^i_{p_i} : T^i_{p_i}}] \vdash T^i_{p_i + 1} \downarrow {T'}^i_{p_i + 1}
  \end{array}
  \right\} i \in \{1 \ldots n\} \\
  \vector{t_n}~\mbox{guarded by destructors (\cite{ck-sadhana}, Sect. 6.3)}
 }
 {\env \cup \left( 
  \begin{array}{l}
   \Sigma,\Phi,\\
   \letr[\\]{f_1 (\vector{x^1_{p_1} : T^1_{p_1}})}
        {T^1_{p_1+1}}{t_1}
        {f_n (\vector{x^n_{p_n} \!:\! T^n_{p_n}})}
        {T^n_{p_n+1}}{t_n}
  \end{array}
  \right)
  \vdash WF}
$$

$$\infrule
 {\env \vdash WF \andalso
  \mbox{$\vector{f_n}$ undefined in $\env$} \andalso
  \env,\Sigma \vdash WF \andalso
  \env,\Sigma,\Phi \vdash WF \andalso\\
  \env,\Sigma,\Phi,\emptyset \vdash T_i: S_i \andalso
  \env,\Sigma,\Phi,\emptyset \vdash S_i \mwhd S'_i~\mbox{ where $S'_i$ is a sort}\\
  T_i = \Pi \vector{x^i_{p_i} : T^i_{p_i}}.T^i_{p_i + 1} \\
  \left.
  \begin{array}{l}
  \env,\Sigma,\Phi,[f_1:T_1;\ldots;f_n:T_n;\vector{x^i_{p_i} : T^i_{p_i}}] \vdash t_i: {T'}^i_{p_i + 1} \\
  \env,\Sigma,\Phi,[f_1:T_1;\ldots;f_n:T_n;\vector{x^i_{p_i} : T^i_{p_i}}] \vdash T^i_{p_i + 1} \downarrow {T'}^i_{p_i + 1}
  \end{array}
  \right\} i \in \{1 \ldots n\} \\
  \vector{t_n}~\mbox{guarded by constructors (\cite{ck-sadhana}, Sect. 6.3)}
 }
 {\env \cup \left( 
  \begin{array}{l}
   \Sigma,\Phi,\\
   \letcr[\\]{f_1 (\vector{x^1_{p_1} : T^1_{p_1}})}
        {T^1_{p_1+1}}{t_1}
        {f_n (\vector{x^n_{p_n} \!:\! T^n_{p_n}})}
        {T^n_{p_n+1}}{t_n}
  \end{array}
  \right)
  \vdash WF}
$$

$$\infrule
 {\env \vdash WF \andalso
  \mbox{$I^1_l,\ldots,I^n_l,k_1^1,\ldots,k_n^{m_n}$ undefined in $\env$} \andalso
  \env,\Sigma \vdash WF \andalso
  \env,\Sigma,\Phi \vdash WF \andalso\\
  \mbox{all the conditions in \cite{ck-sadhana}, Sect. 6.1 are satisfied}
 }
 {\env \cup ( \Phi, \Sigma, \mathcal{I} ) \vdash WF}
$$

\subsection{Metasenv formation rules}
Metasenv formation rules (judgement $\env, \Sigma  \vdash WF$, function \verb+typecheck_metasenv+)

$$\infrule
 {}
 {\env,\emptyset \vdash WF}
$$

$$\infrule
 {\env,\Sigma \vdash WF \andalso
  \mbox{$?_i$ undefined in $\Sigma$} \andalso
  \env,\Sigma,\emptyset,\Gamma \vdash WF \\
  \env,\Sigma,\emptyset,\Gamma \vdash T:S \andalso
  \env,\Sigma,\emptyset,\Gamma \vdash S \mwhd S' \mbox{where $S'$ is a sort}}
 {\env,\Sigma \cup ( \Gamma \vdash ?_i: T) \vdash WF}
$$

\subsection{Subst formation rules}
Subst formation rules (judgement $\env, \Sigma, \Phi  \vdash WF$, function \verb+typecheck_subst+)

$$\infrule
 {}
 {\env,\Sigma,\emptyset \vdash WF}
$$

$$\infrule
 {\env,\Sigma,\Phi \vdash WF \andalso
  \mbox{$?_i$ undefined in $\Sigma$ and in $\Phi$} \andalso
  \env,\Sigma,\Phi,\Gamma \vdash WF \andalso\\
  \env,\Sigma,\Phi,\Gamma \vdash T:S \andalso
  \env,\Sigma,\Phi,\Gamma \vdash S \mwhd S'~\mbox{where $S'$ is a sort} \\
  \env,\Sigma,\Phi,\Gamma \vdash t:T' \andalso
  \env,\Sigma,\Phi,\Gamma \vdash T \downarrow T'}
 {\env,\Sigma,\Phi \cup ( \Gamma \vdash ?_i: T := t) \vdash WF}
$$

\subsection{Context formation rules}
Context formation rules (judgement $\env, \Sigma, \Phi, \Gamma  \vdash WF$, function \verb+typecheck_context+)

$$\infrule
 {}
 {\env,\Sigma,\Phi,\emptyset \vdash WF}
$$

$$\infrule
 {\env,\Sigma,\Phi,\Gamma \vdash WF \andalso
  \mbox{$x$ is undefined in $\Gamma$} \\
  \env,\Sigma,\Phi,\Gamma \vdash T: S \andalso
  \env,\Sigma,\Phi,\Gamma \vdash S \mwhd S'~\mbox{where $S'$ is a sort}}
 {\env,\Sigma,\Phi,\Gamma \cup ( x:T ) \vdash WF}
$$

$$\infrule
 {\env,\Sigma,\Phi,\Gamma \vdash WF \andalso
  \mbox{$x$ is undefined in $\Gamma$} \\
  \env,\Sigma,\Phi,\Gamma \vdash T: S \andalso
  \env,\Sigma,\Phi,\Gamma \vdash S \mwhd S'~\mbox{where $S'$ is a sort}\\
  \env,\Sigma,\Phi,\Gamma \vdash t: T' \andalso
  \env,\Sigma,\Phi,\Gamma \vdash T \downarrow T' \andalso
  }
 {\env,\Sigma,\Phi,\Gamma \cup ( x:T := t ) \vdash WF}
$$

\subsection{Term typechecking rules}

Term typechecking rules (judgement $\env,\Sigma,\Phi,\Gamma \vdash t:T$, function \verb+typeof+)

$$\infrule[(\mathcal{K}\mathrm{-variable})]
{( x:T ) \in \Gamma \quad \mbox{ or } \quad
( x:T := t ) \in \Gamma}
{\env,\Sigma,\Phi,\Gamma \vdash x:T}
\quad
\infrule[(\mathcal{K}\mathrm{-sort})]
{ ( s_1, s_2 ) \in \pts}
{\env,\Sigma,\Phi,\Gamma \vdash s_1 : s_2}
$$

$$\infrule[(\mathcal{K}\mathrm{-meta})]
{ ( x_1:T_1;\ldots ; x_n:T_n \vdash ?_i : T ) \in \Sigma \quad\mbox{or}\quad
  ( x_1:T_1;\ldots ; x_n:T_n \vdash ?_i : T := t ) \in \Phi \\
  \env,\Sigma,\Phi,\Gamma \vdash t_i : T_i[\vector{x_{i-1}/t_{i-1}}] \quad i \in \{1 \ldots n\}\\
}
{\env,\Sigma,\Phi,\Gamma \vdash ?_i[t_1;\ldots;t_n]:T[\vector{x_n/t_n}]}
$$

$$\infrule[(\mathcal{K}\mathrm{-constant})]
{(r : T) \in \env{}}
{\env,\Sigma,\Phi,\Gamma \vdash r : T}
$$

$$\infrule[(\mathcal{K}\mathrm{-definition})]
 {( \Sigma',\Phi',\mbox{\verb[language=grafite]+definition+}~d:T:=b) \in \env \quad\mbox{or}\quad
 ( \Sigma',\Phi',\mbox{\verb[language=grafite]+axiom+}~d:T)
 \in \env \\
 \Sigma' = \emptyset \andalso \Phi' = \emptyset }
 {(d : T) \in \env}
$$

$$\infrule[(\mathcal{K}\mathrm{-letrec})]
 {\left( 
  \begin{array}{l}
  \Sigma',\Phi',\\
   \letr[\\]{f_1 (\vector{x^1_{p_1} : T^1_{p_1}})}
        {T^1_{p_1+1}}{t_1}
        {f_n (\vector{x^n_{p_n} \!:\! T^n_{p_n}})}
        {T^n_{p_n+1}}{t_n}
  \end{array}
  \right) \in \env \\
  \Sigma' = \emptyset \andalso \Phi' = \emptyset \andalso 1 \leq i \leq n}
 {(f_i : \Pi \vector{x^i_{p_i} : T^i_{p_i}}.T^i_{p_i + 1}) \in \env}
$$

$$\infrule[(\mathcal{K}\mathrm{-letcorec})]
 {\left( 
  \begin{array}{l}
  \Sigma',\Phi',\\
   \letcr[\\]{f_1 (\vector{x^1_{p_1} : T^1_{p_1}})}
        {T^1_{p_1+1}}{t_1}
        {f_n (\vector{x^n_{p_n} \!:\! T^n_{p_n}})}
        {T^n_{p_n+1}}{t_n}
  \end{array}
  \right) \in \env \\
  \Sigma' = \emptyset \andalso \Phi' = \emptyset \andalso 1 \leq i \leq n}
 {(f_i : T_i) \in \env}
$$

$$\infrule[(\mathcal{K}\mathrm{-inductive})]
 {( \Sigma',\Phi',\mathcal{I}) \in \env\\ 
    \Sigma' = \emptyset \andalso \Phi' = \emptyset \andalso 1 \leq p \leq n}
 {(I^p_l : \Prod \vector{x_l:U_l}.A_p) \in \env}
$$

$$\infrule[(\mathcal{K}\mathrm{-constructor})]
 {( \Sigma',\Phi',\mathcal{I}) \in \env\\ 
    \Sigma' = \emptyset \andalso \Phi' = \emptyset \andalso 1 \leq p \leq n
    \andalso 1 \leq j \leq m_p}
 {(k^j_p : \Prod \vector{x_l:U_l}.K^j_p) \in \env}
$$

$$\infrule[(\mathcal{K}\mathrm{-lambda})]
{\env,\Sigma,\Phi,\Gamma \vdash T : S \\
 \env,\Sigma,\Phi,\Gamma \vdash S \mwhd S' \andalso \mbox{$S'$ is a sort or a meta} \\
 \env,\Sigma,\Phi,\Gamma \cup ( n:T ) \vdash u : U}
{\env,\Sigma,\Phi,\Gamma \vdash \lambda n:T.u : \Prod n:T.U}
$$

$$\infrule[(\mathcal{K}\mathrm{-product})]
{\env,\Sigma,\Phi,\Gamma \vdash T:s_1 \\
 \env,\Sigma,\Phi,\Gamma \cup ( n:T ) \vdash U:s_2 \\
 (s_1, s_2, s_3) \in \pts}
{\env,\Sigma,\Phi,\Gamma \vdash \Prod n:T.U : s_3}
$$

$$\infrule[(\mathcal{K}\mathrm{-letin})]
{\env,\Sigma,\Phi,\Gamma \vdash t : T' \\
 \env,\Sigma,\Phi,\Gamma \vdash T : S \andalso \env,\Sigma,\Phi,\Gamma \vdash T \downarrow T'  \\
 \env,\Sigma,\Phi,\Gamma \cup ( x:T := t ) \vdash u : U}
{\env,\Sigma,\Phi,\Gamma \vdash \mbox{\verb[language=grafite]+let+}~(x:T) :=
t ~\mbox{\verb[language=grafite]+in+}~u : U [x/t]}
$$

$$\infrule[(\mathcal{K}\mathrm{-appl-base})]
{\env,\Sigma,\Phi,\Gamma \vdash h:\Prod x:T.U \\
 \env,\Sigma,\Phi,\Gamma \vdash t:T'
 \andalso \env,\Sigma,\Phi,\Gamma \vdash T \downarrow T'}
{\env,\Sigma,\Phi,\Gamma \vdash h \; t : U [x/t]}
$$

$$\infrule[(\mathcal{K}\mathrm{-appl-rec})]
{\env,\Sigma,\Phi,\Gamma \vdash (h \; t_1) \; t_2 \cdots t_n : T}
{\env,\Sigma,\Phi,\Gamma \vdash h \; t_1 \; t_2  \cdots t_n : T}
$$

$$\infrule[(\mathcal{K}\mathrm{-match})]
{( \Sigma', \Phi',\mathcal{I} ) \in \env \andalso
 \Sigma' = \emptyset \andalso \Phi' = \emptyset \andalso
 \env,\Sigma,\Phi,\Gamma \vdash t : T \\
 \env,\Sigma,\Phi,\Gamma \vdash T \mwhd I^p_l \; \vector{u_l} \; \vector{u'_r} \\
 A_p [\vector{x_l/u_l}] = \Prod \vector{y_r:Y_r}.s \andalso
 K^j_p [\vector{x_l/u_l}] =
 \Prod \vector{x^j_{n_j} : Q^j_{n_j}}.I^p_l~\vector{x_l}~\vector{v_r} \qquad j = 1\ldots m_p\\
 \env,\Sigma,\Phi,\Gamma \vdash U : V \andalso
 \env,\Sigma,\Phi,\Gamma \vdash V \mwhd \Prod \vector{z_r:Y_r}.\Prod z_{r+1}:I^p_l~\vector{u_l}~\vector{z_r}.s' \\
 (s , s') \in \elim{}(\pts{})\\
 \env,\Sigma,\Phi,\Gamma \vdash \lambda \vector{x^j_{n_j} : P^j_{n_j}}.t_j : T_j \qquad \hfill j = 1, \ldots, m_p\\
 \env,\Sigma,\Phi,\Gamma \vdash T_j \downarrow \Pi \vector{x^j_{n_j} : Q^j_{n_j}}. U~\vector{v_r}~(k^p_j~\vector{u_l}~\vector{x^j_{n_j}}) \qquad \hfill j=1,\ldots,m_p}
{\begin{array}{rl}
 \env,\Sigma,\Phi,\Gamma \vdash &
 \match{t}{I^p_l}{U} \\
 &
 \;\; [ k^p_1~(\vector{x^1_{n_1} : P^1_{n_1}}) \Rightarrow t_1~|\ldots~|
        k^p_{m_p}~(\vector{x^{m_p}_{n_{m_p}} : P^{m_p}_{n_{m_p}}}) \Rightarrow t_{m_p}~]
 : U~\vector{u'_r}~t
 \end{array}
}
$$

\subsection{Term conversion rules}
Term conversion rules (judgement
$\env,\Sigma,\Phi,\Gamma \vdash T \downarrow T'$, function
\verb+are_convertible+; $\downarrow_=$ means \verb+test_eq_only = true+;
$\downarrow_\bullet$ means that the current rule must be intended as two
rules, one with all the $\downarrow_\bullet$ replaced by $\downarrow$, the
other with all the $\downarrow_\bullet$ replaced by $\downarrow_=$)

$$\infrule { 
 \env,\Sigma,\Phi,\Gamma \vdash T =_{\alpha} T'
}{
 \env,\Sigma,\Phi,\Gamma \vdash T \downarrow_= T'
}
$$

$$\infrule {
\env,\Sigma,\Phi,\Gamma \vdash T \downarrow_= T'
}{
\env,\Sigma,\Phi,\Gamma \vdash T \downarrow T'
}
$$

$$\infrule { 
 \Type_u \leq \Type_v \andalso \Type_v \leq \Type_u \andalso \textrm{are declared constraints
 (\cite{ck-sadhana}, Sect. 4.3)}
}{
 \env,\Sigma,\Phi,\Gamma \vdash \Type_u \downarrow_= \Type_v
}
$$

$$\infrule { 
 \Type_u \leq \Type_v \andalso \textrm{is a declared constraint
 (\cite{ck-sadhana}, Sect. 4.3)}
}{
 \env,\Sigma,\Phi,\Gamma \vdash \Type_u \downarrow \Type_v
}
$$

$$\infrule { 
}{
 \env,\Sigma,\Phi,\Gamma \vdash \Prop \downarrow \Type_u
}
$$

$$\infrule{
lc = t_1,\ldots,t_n \andalso lc' = t'_1,\ldots,t'_n \\
\mbox{for all $i = 1,\ldots,n$} \qquad \env,\Sigma,\Phi,\Gamma \vdash t_i \downarrow_\bullet t'_i
}{
\env,\Sigma,\Phi,\Gamma \vdash ?_j[lc] \downarrow_\bullet ?_j[lc']
}
$$

$$\infrule { 
 \env,\Sigma,\Phi,\Gamma \vdash T_1 \downarrow_= T_1'
 \andalso
 \env,\Sigma,\Phi,\Gamma \cup ( x : T_1 ) \vdash T_2
 \downarrow_\bullet T_2'
}{
 \env,\Sigma,\Phi,\Gamma \vdash \Prod x:T_1.T_2 \downarrow_\bullet \Prod x:T_1'.T_2'
}
$$

$$\infrule { 
 \env,\Sigma,\Phi,\Gamma \cup ( x : T ) \vdash t
 \downarrow_\bullet t'
}{
 \env,\Sigma,\Phi,\Gamma \vdash \lambda x:T.t \downarrow_\bullet \lambda
 x:T'.t'
}
$$
In the rule above, no check is performed on the source of the
abstractions, since we assume we are comparing well-typed terms whose types are
convertible.

$$\infrule{
\env,\Sigma,\Phi,\Gamma \vdash h \downarrow_\bullet h' \\
\mbox{for all $i = 1,\ldots,n$} \qquad \env,\Sigma,\Phi,\Gamma \vdash t_i
\downarrow_= t'_i
}{
\env,\Sigma,\Phi,\Gamma \vdash h \; \vector{t_n} \downarrow_\bullet h' \;
\vector{t'_n}
}
$$

$$\!\!\!\!\!\!\!\!\infrule{
\env,\Sigma,\Phi,\Gamma \vdash t \downarrow_\bullet t' \andalso
\env,\Sigma,\Phi,\Gamma \vdash U \downarrow_\bullet U' \\
\mbox{for all $i=1,\ldots,m_p$} \quad 
\env,\Sigma,\Phi,\Gamma \vdash \lambda \vector{x^i_{n_i} : P^i_{n_i}}.t_i
\downarrow_\bullet \lambda \vector{x^i_{n_i} : {P'}^i_{n_i}}.t'_i}
{ \begin{array}{ll}\env,\Sigma,\Phi,\Gamma \vdash\!\! &
 \match{t}{I^p_l}{U} [ k^p_1~(\vector{x^1_{n_1} : P^1_{n_1}}) \Rightarrow t_1
 | \ldots | k^p_{m_p}~(\vector{x^{m_p}_{n_{m_p}} : P^{m_p}_{n_{m_p}}}) \Rightarrow t_{m_p} ]~
 \downarrow_\bullet \\ &
 \match{t'}{I^p_l}{U} [ k^p_1~(\vector{x^1_{n_1} : {P'}^1_{n_1}}) \Rightarrow t'_1
 | \ldots | k^p_{m_p}~(\vector{x^{m_p}_{n_{m_p}} : {P'}^{m_p}_{n_{m_p}}}) \Rightarrow t'_{m_p}]
 \end{array}}
$$

$$\infrule{
\env,\Sigma,\Phi,\Gamma \vdash t \mwhd t' \andalso
\env,\Sigma,\Phi,\Gamma \vdash u \mwhd u' \andalso
\env,\Sigma,\Phi,\Gamma \vdash t' \downarrow_\bullet u'}{
\env,\Sigma,\Phi,\Gamma \vdash t \downarrow_\bullet u}
$$

\noindent
In the previous rule, $t'$ and $u'$ need not be weak head normal forms:
any term obtained from $t$ (respectively, $u$) by reduction (even non-head
reduction) will do. Indeed, the less reduction is performed, the more
efficient the conversion test usually is.
~\\

\subsection{Term reduction rules}Term reduction rules.
%(judgement $\env,\Sigma,\Phi,\Gamma \vdash T
%=_{\beta\delta\iota\zeta\mu\nu} T'$ is the reflexive, symmetric, transitive
%and context aware closure of the following reduction rules)

\[
\env,\Sigma,\Phi,\Gamma \vdash (\lambda x:T.u)~t ~\rhd_\beta~ u\subst{t}{x}
\]

\[
 \env,\Sigma,\Phi,\Gamma \vdash \mbox{\verb[language=grafite]+let+}~(x:T) :=
 t ~\mbox{\verb[language=grafite]+in+}~u 
~\rhd_\zeta~ u\subst{t}{x}
\]

$$\infrule{
( \emptyset,\emptyset,\mbox{\verb[language=grafite]+definition+}~d:T:=b )
\in \env}
{\env,\Sigma,\Phi,\Gamma \vdash d ~\rhd_\delta~ b}
$$

$$\infrule{
( \Gamma' \vdash ?_i : T := t ) \in \Phi
}
{\env,\Sigma,\Phi,\Gamma \vdash ?_i [u_1\;;\;\ldots\;;\;u_n]~\rhd_\delta~t[dom(\Gamma')/\vector{u_n}]}
$$

%\noindent
%Note that restricting the unfolding of definition to the case of
%empty substitutions and empty metavariables environment is a 
%simplification actually reflected by the code.

\[
 \begin{array}{rl}
 \env,\Sigma,\Phi,\Gamma \vdash &
 \match{k^p_i~\vector{t_l}~\vector{t'_{n_i}}}{I^p_l}{U} \\
 &
 [ k^p_1~(\vector{x^1_{n_1} : P^1_{n_1}}) \Rightarrow u_1 | \ldots |
   k^p_{m_p}~(\vector{x^{m_p}_{n_{m_p}} : P^{m_p}_{n_{m_p}}}) \Rightarrow u_{m_p} ]
 ~\rhd_\iota~u_i [\vector{x^i_{n_i}}/\vector{t'_{n_i}}]
 \end{array}
\]

$$\infrule{
\left( 
 \begin{array}{l}
 \emptyset, \emptyset, \\
   \letr[\\]{f_1 (\vector{x^1_{p_1} : T^1_{p_1}})}
        {T^1_{p_1+1}}{t_1}
        {f_n (\vector{x^n_{p_n} \!:\! T^n_{p_n}})}
        {T^n_{p_n+1}}{t_n}
 \end{array} \right) \in \env \\
 k \in \{1\ldots n\}}
{\env,\Sigma,\Phi,\Gamma \vdash f_k~u_1 ~ ... (k^i_j~\vector{v_{n_j}}) ... ~ u_m 
 ~\rhd_\mu~
 t_k[\vector{x^k_m} / u_1,\ldots,(k^i_j~\vector{v_{n_j}}),\ldots,u_m]}
$$
\noindent
Notice that $(k^i_j~\vector{v_{n_j}})$ must occur
in the position of the recursive argument of $f_k$. This implies that, for
this reduction to be performed, $f_k$ must be applied at least up to its
recursive argument.

$$\infrule{
\left( 
 \begin{array}{l}
 \emptyset, \emptyset, \\
   \letcr[\\]{f_1 (\vector{x^1_{p_1} : T^1_{p_1}})}
        {T^1_{p_1+1}}{t_1}
        {f_n (\vector{x^n_{p_n} \!:\! T^n_{p_n}})}
        {T^n_{p_n+1}}{t_n}
 \end{array} \right) \in \env \\
 k \in \{1\ldots n\}}
{ \begin{array}{ll}\env,\Sigma,\Phi,\Gamma \vdash &
 \match{f_k~\vector{u_q}}{I^p_l}{U} \\
 &
 [ k^p_1~(\vector{y^1_{n_1} : P^1_{n_1}}) \Rightarrow v_1 | \ldots |
   k^p_{m_p}~(\vector{y^{m_p}_{n_{m_p}} : P^{m_p}_{n_{m_p}}}) \Rightarrow v_{m_p} ]
 ~\rhd_\nu \\ &
 \match{t_k[\vector{x^k_q}/\vector{u_q}]}{I^p_l}{U} \\
 &
 [ k^p_1~(\vector{y^1_{n_1} : P^1_{n_1}}) \Rightarrow v_1 | \ldots |
   k^p_{m_p}~(\vector{y^{m_p}_{n_{m_p}} : P^{m_p}_{n_{m_p}}}) \Rightarrow v_{m_p} ]
\end{array}}
$$

\noindent
Notice that here $q$ can be zero.
\bibliographystyle{plain}
\bibliography{../BIBTEX/helm}

\begin{thebibliography}{10}

\bibitem{armentano}
Andrea Asperti and Cristian Armentano.
\newblock A page in number theory.
\newblock {\em Journal of Formalized Reasoning}, 1:1--23, 2008.

\bibitem{asperti-ricciotti}
Andrea Asperti and Wilmer Ricciotti.
\newblock About the formalization of some results by {C}hebyshev in number
  theory.
\newblock In {\em Proc. of TYPES'08}, volume 5497 of {\em LNCS}, pages 19--31.
  Springer-Verlag, 2009.

\bibitem{matitapoplmark}
Andrea Asperti, Wilmer Ricciotti, Claudio {Sacerdoti Coen}, and Enrico Tassi.
\newblock Formal metatheory of programming languages in the {Matita}
  interactive theorem prover.
\newblock {\em Journal of Automated Reasoning: Special Issue on the Poplmark
  Challenge}.
\newblock Published online, May 2011.

\bibitem{ck-sadhana}
Andrea Asperti, Wilmer Ricciotti, Claudio {Sacerdoti Coen}, and Enrico Tassi.
\newblock A compact kernel for the {C}alculus of {I}nductive {C}onstructions.
\newblock {\em Sadhana}, 34(1):71--144, 2009.

\bibitem{unification-hints}
Andrea Asperti, Wilmer Ricciotti, Claudio {Sacerdoti Coen}, and Enrico Tassi.
\newblock Hints in unification.
\newblock In {\em TPHOLs 2009}, volume 5674/2009 of {\em LNCS}, pages 84--98.
  Springer-Verlag, 2009.

\bibitem{matita-jar-uitp}
Andrea Asperti, Claudio {Sacerdoti Coen}, Enrico Tassi, and Stefano Zacchiroli.
\newblock User interaction with the {M}atita proof assistant.
\newblock {\em Journal of Automated Reasoning}, 39(2):109--139, 2007.

\bibitem{BarendregtH:lawcwt}
Henk Barendregt.
\newblock {Lambda Calculi with Types}.
\newblock In {Abramsky, Samson and others}, editor, {\em {Handbook of Logic in
  Computer Science}}, volume~2. {Oxford University Press}, 1992.

\bibitem{agda}
Ana Bove, Peter Dybjer, and Ulf Norell.
\newblock A brief overview of {A}gda - a functional language with dependent
  types.
\newblock In {\em Theorem Proving in Higher Order Logics, 22nd International
  Conference, TPHOLs 2009, Munich, Germany, August 17-20, 2009. Proceedings},
  volume 5674 of {\em LNCS}, pages 73--78. Springer, 2009.

\bibitem{callaghan00coherence}
P.~Callaghan.
\newblock Coherence checking of coercions in {P}lastic.
\newblock In {\em In Proc. Workshop on Subtyping and Dependent Types in
  Programming}, 2000.

\bibitem{chenPHD}
Gang Chen.
\newblock {\em Subtyping, Type Conversion and Transitivity Elimination}.
\newblock PhD thesis, University Paris 7, 1998.

\bibitem{TS11}
Claudio~Sacerdoti Coen and Enrico Tassi.
\newblock {Formalizing Overlap Algebras in Matita}.
\newblock {\em Mathematical Structures in Computer Science}, 21:1--31, 2011.

\bibitem{coq}
The {C}oq proof-assistant.
\newblock \\\url{http://coq.inria.fr}, 2009.

\bibitem{CC}
Thierry Coquand and G{\'e}rard~P. Huet.
\newblock The {C}alculus of {C}onstructions.
\newblock {\em Inf. Comput.}, 76(2/3):95--120, 1988.

\bibitem{dunfieldgreedy}
Joshua Dunfield.
\newblock Greedy bidirectional polymorphism.
\newblock In {\em ML Workshop (ML '09)}, pages 15--26, August 2009.
\newblock \url{http://www.cs.cmu.edu/~joshuad/papers/poly/}.

\bibitem{pfenning3}
Joshua Dunfield and Frank Pfenning.
\newblock Tridirectional typechecking.
\newblock In X.~Leroy, editor, {\em Conference Record of the 31st Annual
  Symposium on Principles of Programming Languages (POPL'04)}, pages 281--292,
  2004.

\bibitem{canonical-structures}
Fran\c{c}ois Garillot, Georges Gonthier, Assia Mahboubi, and Laurence Rideau.
\newblock Packaging mathematical structures.
\newblock In {\em Proceedings of the 22nd International Conference on Theorem
  Proving in Higher Order Logics}, TPHOLs '09, pages 327--342, Berlin,
  Heidelberg, 2009. Springer-Verlag.

\bibitem{huet2order}
G{\'e}rard~P. Huet.
\newblock A unification algorithm for typed lambda-calculus.
\newblock {\em Theor. Comput. Sci.}, 1(1):27--57, 1975.

\bibitem{isabelle}
The {Isabelle} proof-assistant.
\newblock \\\url{http://www.cl.cam.ac.uk/Research/HVG/Isabelle/}.

\bibitem{lego}
The {L}ego proof-assistant.
\newblock \\\url{http://www.dcs.ed.ac.uk/home/lego/}.

\bibitem{coercivesubtyping}
Zhaohui Luo.
\newblock Coercive subtyping.
\newblock {\em J. Logic and Computation}, 9(1):105--130, 1999.

\bibitem{munoz}
C{\'e}sar Mu{\~n}oz.
\newblock {\em A Calculus of Substitutions for Incomplete-Proof Representation
  in Type Theory}.
\newblock PhD thesis, INRIA, November 1997.

\bibitem{phd-norell}
Ulf Norell.
\newblock {\em Towards a practical programming language based on dependent type
  theory}.
\newblock PhD thesis, Department of Computer Science and Engineering, Chalmers
  University of Technology, SE-412 96 G\"{o}teborg, Sweden, September 2007.

\bibitem{mohring}
Christine Paulin-Mohring.
\newblock {\em D\'efinitions Inductives en Th\'eorie des Types d'Ordre
  Sup\'erieur}.
\newblock Habilitation \`a diriger les recherches, Universit\'e Claude Bernard
  Lyon I, December 1996.

\bibitem{piercelocaltype}
Benjamin~C. Pierce and David~N. Turner.
\newblock Local type inference.
\newblock {\em ACM Trans. Program. Lang. Syst.}, 22:1--44, January 2000.

\bibitem{pollackFAC02}
Robert Pollack.
\newblock Dependently typed records in type theory.
\newblock {\em Formal Aspects of Computing}, 13:386--402, 2002.

\bibitem{csc-phd}
Claudio {Sacerdoti Coen}.
\newblock {\em Mathematical Knowledge Management and Interactive Theorem
  Proving}.
\newblock PhD thesis, University of Bologna, 2004.
\newblock Technical Report UBLCS 2004-5.

\bibitem{mcs2008}
Claudio {Sacerdoti Coen} and Stefano Zacchiroli.
\newblock Spurious disambiguation errors and how to get rid of them.
\newblock {\em Journal of Mathematics in Computer Science, special issue on
  Management of Mathematical Knowledge}, 2:355--378, 2008.

\bibitem{saibi-inheritance}
Amokrane Sa\"{\i}bi.
\newblock Typing algorithm in type theory with inheritance.
\newblock In {\em Proceedings of the 24th ACM SIGPLAN-SIGACT symposium on
  Principles of programming languages}, POPL '97, pages 292--301, New York, NY,
  USA, 1997. ACM.

\bibitem{russell}
Matthieu Sozeau.
\newblock Subset coercions in {C}oq.
\newblock In {\em Types for Proofs and Programs}, volume 4502/2007 of {\em
  LNCS}, pages 237--252. Springer-Verlag, 2006.

\bibitem{SozeauO08}
Matthieu Sozeau and Nicolas Oury.
\newblock First-class type classes.
\newblock In {\em Proceedings of TPHOLs}, pages 278--293, 2008.

\bibitem{Spiwack}
Arnaud Spiwack.
\newblock {\em Verified Computing in Homological Algebra. A Journey Esploring
  the Power and Limits of Dependent Type Theory}.
\newblock PhD thesis, \'Ecole Polytechniqe, 2011.

\bibitem{strecker}
Martin Strecker.
\newblock {\em Construction and Deduction in Type Theories}.
\newblock PhD thesis, Universit{\"a}t Ulm, 1998.

\bibitem{StuckeySW06}
Peter~J. Stuckey, Martin Sulzmann, and Jeremy Wazny.
\newblock Type processing by constraint reasoning.
\newblock In {\em APLAS}, pages 1--25, 2006.

\bibitem{tassi-phd}
Enrico Tassi.
\newblock {\em Interactive Theorem Provers: issues faced as a user and tackled
  as a developer}.
\newblock PhD thesis, University of Bologna, 2008.

\bibitem{Werner}
Benjamin Werner.
\newblock {\em Une Th\'eorie des {C}onstructions {I}nductives}.
\newblock PhD thesis, Universit\'e Paris VII, May 1994.

\end{thebibliography}

    \insert\copyins{\hsize.57\textwidth
\vbox to 0pt{\vskip12 pt%
      \fontsize{6}{7 pt}\normalfont\upshape
      \everypar{}%
      \noindent\fontencoding{T1}%
  \textsf{This work is licensed under the Creative Commons
  Attribution-NoDerivs License. To view a copy of this license, visit
  \texttt{http://creativecommons.org/licenses/by-nd/2.0/} or send a
  letter to Creative Commons, 171 Second St, Suite 300, San Francisco,
    CA 94105, USA, or Eisenacher Strasse 2, 10777 Berlin, Germany}\vss}}
%      \par
%      \kern\zero}%

\unskip

\end{document}